\documentclass[a4paper,11pt]{article}

\usepackage{fullpage} 
\usepackage{hyperref}
\usepackage{amsmath}
\usepackage{amsfonts}
\usepackage{amssymb}
\usepackage{amsthm}
\usepackage{mathrsfs}
\usepackage{textcomp}
\usepackage{stmaryrd}
\usepackage{latexsym} 
\usepackage{enumitem}
\usepackage{proof}
\usepackage{longtable}
\usepackage{caption}
\usepackage{morefloats}
\usepackage[utf8]{inputenc}
\usepackage[OT4]{fontenc}
\usepackage[all]{xy}

\theoremstyle{plain}
\newtheorem{theorem}{Theorem}[section]
\newtheorem{proposition}[theorem]{Proposition}
\newtheorem{lemma}[theorem]{Lemma}

\newtheorem{corollary}[theorem]{Corollary}

\theoremstyle{definition}
\newtheorem{definition}[theorem]{Definition}

\newtheorem{example}[theorem]{Example}

\theoremstyle{plain}

\theoremstyle{definition}

\newcommand{\To}{\ensuremath{\Rightarrow}}

\newcommand{\pos}[2]{\ensuremath{{#1}_{|#2}}}

\newcommand{\cut}[2]{\ensuremath{{#1}_{\upharpoonright#2}}}
\newcommand{\ucut}[2]{\ensuremath{{#1}^{\upharpoonright#2}}}

\newcommand{\ceil}[1]{\ensuremath{\lceil#1\rceil}}

\newcommand{\Nbb}{\ensuremath{\mathbb{N}}}

\newcommand{\Pb}{\ensuremath{{\mathbb P}}}

\newcommand{\Sb}{\ensuremath{{\mathbb S}}}
\newcommand{\Ab}{\ensuremath{{\mathbb A}}}
\newcommand{\Bb}{\ensuremath{{\mathbb B}}}

\newcommand{\Pow}[1]{\ensuremath{\mathcal{P}(#1)}}

\newcommand{\la}{\ensuremath{\langle}}
\newcommand{\ra}{\ensuremath{\rangle}}

\newcommand{\Tb}{\ensuremath{{\mathbb T}}}

\newcommand{\D}{\ensuremath{{\mathcal D}}}

\newcommand{\Tc}{\ensuremath{{\mathcal T}}}

\makeatletter

\newcommand{\Rmnum}[1]{\expandafter\@slowromancap\romannumeral #1@}
\makeatother

\newcommand{\qle}{\sqsubseteq}
\newcommand{\qge}{\sqsupseteq}

\newcommand{\join}{\bigvee}
\newcommand{\meet}{\bigwedge}

\newcommand{\id}{\mathrm{id}}
\newcommand{\op}{\mathrm{op}}

\newcommand{\Fix}{\mathrm{Fix}}
\newcommand{\Max}{\mathrm{Max}}
\newcommand{\Min}{\mathrm{Min}}
\newcommand{\On}{\mathrm{On}}

\newcommand{\da}{{\downarrow}}
\newcommand{\ua}{{\uparrow}}

\newcommand{\Two}{\mathbf{2}}

\newcommand{\ulift}[1]{#1{\uparrow}}

\newcommand{\cons}{\mathtt{cons}}
\newcommand{\even}{\mathtt{even}}
\newcommand{\tail}{\mathtt{tl}}
\newcommand{\head}{\mathtt{hd}}

\newcommand{\add}{\mathtt{add}}
\newcommand{\zip}{\mathtt{zip}}

\newcommand{\mul}{\mathtt{mul}}
\newcommand{\mer}{\mathtt{merge}}

\newcommand{\tcut}{\mathtt{cut}}

\newcommand{\sfr}{\mathfrak{s}}
\newcommand{\dfr}{\mathfrak{d}}

\DeclareMathOperator{\Coloneqq}{{\,::=\,}}

\title{Coinduction: an elementary approach}

\author{\L{}ukasz Czajka \medskip \\
  \small
    Institute of Informatics, University of Warsaw\\
  \small
    Banacha 2, 02-097 Warszawa, Poland\\
  \small
    lukaszcz@mimuw.edu.pl
}

\date{May 22, 2019}

\allowdisplaybreaks

\setlist{itemsep=0pt,topsep=\parsep}

\begin{document}
\maketitle

\begin{abstract}
  The main aim of this paper is to promote a certain style of doing
  coinductive proofs, similar to inductive proofs as commonly done by
  mathematicians. For this purpose we provide a reasonably direct
  justification for coinductive proofs written in this style, i.e.,
  converting a coinductive proof into a non-coinductive argument is
  purely a matter of routine. In this way, we provide an elementary
  explanation of how to interpret coinduction in set theory.
\end{abstract}

\newcommand{\IND}{\mathrm{IND}}
\newcommand{\COIND}{\mathrm{COIND}}

\tableofcontents

\section{Introduction}

In its basic and most common form, coinduction is a method for
reasoning about the greatest fixpoints of monotone functions
on~$\Pow{A}$ for some set~$A$. Induction in turn may be seen as a way
of reasoning about the least fixpoints of monotone functions.

Let $F : \Pow{A} \to \Pow{A}$ be monotone. By the Knaster-Tarski
fixpoint theorem, the least fixpoint~$\mu F$ and the greatest
fixpoint~$\nu F$ of~$F$ exist and may be characterized as
\[
\mu F = \bigcap \{ X \in \Pow{A} \mid F(X) \subseteq X \}
\]
\[
\nu F = \bigcup \{ X \in \Pow{A} \mid X \subseteq F(X) \}.
\]
This yields the following proof principles
\[
\infer[(\IND)]{\mu F \subseteq X}{F(X) \subseteq X}\quad
\infer[(\COIND)]{X \subseteq \nu F}{X \subseteq F(X)}
\]
where $X \in \Pow{A}$. The rule~$(\COIND)$ is commonly used as the
principle underlying coinductive proofs. However, this rule is
arguably sometimes inconvenient to apply directly. Ordinarily, when
doing inductive proofs mathematicians do not directly employ the dual
rule~$(\IND)$, explicitly providing the set~$X$ and
calculating~$F(X)$. Nor do they think in terms of~$(\IND)$. Instead,
they show an inductive step, using the inductive hypothesis with
parameters smaller in an appropriate sense. There is a common
understanding when an inductive proof is correct. In ordinary
mathematical practice, nobody bothers to argue each time that an
inductive proof is indeed a valid application of a specific formal
induction principle. Induction is well-understood, and it is
sufficient for everyone that an inductive proof may be formalized ``in
principle''.

In contrast to induction, coinduction is not so well-established and
usually not trusted in the same way. One aim of this paper is to
promote and provide a reasonably direct justification for a certain
style of doing coinductive proofs: showing a coinductive step using a
coinductive hypothesis. As such, the paper has a flavour of a tutorial
with more space devoted to examples than to mathematical results.

From the point of view of someone familiar with research in
coinduction, the results of Section~\ref{sec_coind_tech} are probably
not very novel. They are known ``in principle'' to people studying
coinduction. However, the author believes that there is a need to
present coinductive techniques in a way accessible to a broad
audience, giving simple criteria to verify the correctness of
coinductive proofs and corecursive definitions without being forced to
reformulate them too much to fit a specific formal principle. Our
style of writing coinductive proofs is similar to how such proofs are
presented in
e.g.~\cite{EndrullisPolonsky2011,BezemNakataUustalu2012,NakataUustalu2010,LeroyGrall2009,KozenSilva2017},
but we justify them by direct reduction to transfinite induction. This
seems to provide a more approachable correctness criterion for someone
not familiar with infinite proofs in type
theory~\cite{Coquand1993,Gimenez1994}. Our method for justifying
(non-guarded) corecursive definitions usually boils down to solving
some recursive equations in natural numbers. The coalgebraic approach
to coinduction~\cite{JacobsRutten2011,Rutten2000} is perhaps more
abstract and conceptually satisfying, but not so straightforward to
apply directly. Even the rule~$(\COIND)$ is rather inconvenient in
certain more complex situations.

While the present paper does contain some (relatively simple and
essentially known) mathematical results and rigorous theory, much of
the paper has the status of a ``meta-explanation'' which establishes
certain conventions and describes how to interpret and verify informal
coinductive proofs and corecursive definitions using the theory. This
is analogous to an explanation of how, e.g., the more sophisticated
informal inductive constructions may be encoded in~ZFC set theory. We
do not provide a formal system in which coinductive proofs written in
our style may be formalized directly, but only describe how to convert
them so as to eliminate coinduction, by elaborating them to make
explicit references to the results of our theory. Without a precisely
defined formal system, this description is necessarily
informal. However, we believe the translation is straightforward
enough, so that after properly understanding the present paper it
should be clear that coinduction may be eliminated in the described
manner. It should also become clear how to verify such informally
presented coinductive proofs. Again, the word ``clear'' is used here
in the same sense that it is ``clear'' that common informal
presentations of inductive proofs may in principle be formalized in
ZFC set theory.

\subsection{Related work}

Coinduction and corecursion are by no means new topics. We do not
attempt here to provide an extensive overview of existing
literature. We only mention the pioneering work of~Aczel on
non-well-founded set theory~\cite{Aczel1988}, the final coalgebra
theorem of Aczel and Mendler~\cite{AczelMendler1989}, the subsequent
work of Barr~\cite{Barr1993}, and the work of Rutten~\cite{Rutten2000}
providing coalgebraic foundations for coinduction. A historical
overview of coinduction may be found in~\cite{Sangiorgi2011}. An
elementary introduction to coinduction and bisimulation
is~\cite{Sangiorgi2012}. For a coalgebraic treatment see
e.g.~\cite{JacobsRutten2011,Rutten2000}.

Our approach in Section~\ref{sec_corecursion} is largely inspired by
the work of Sijtsma~\cite{Sijtsma1989} on productivity of streams, and
the subsequent work on sized
types~\cite{AbelPientka2016,AbelPientka2013,Abel2012,Abel2010,Barthe2008,HughesParetoSabry1996}. In
fact, the central Corollary~\ref{cor_unique_solution} is a
generalization of Theorem~32 from~\cite{Sijtsma1989}. In contrast to
the work on sized types, we are not interested in this paper in
providing a formal system, but in explaining corecursion semantically,
in terms of ordinary set theory. Related is also the work on
productivity of streams and infinite data
structures~\cite{Isihara2008,EndrullisGrabmayerHendriks2008,Endrullis2010,ZantemaRaffelsieper2010,EndrullisHendriksKlop2013,Buchholz2005,TelfordTurner1997},
and some of the examples in Section~\ref{sec_corecursion_examples} are
taken from the cited papers. Productivity was first mentioned by
Dijkstra~\cite{Dijkstra1980}. The
articles~\cite{Coquand1993,Gimenez1994} investigate guarded
corecursive definitions in type theory. The
chapters~\cite{BertotCasteran2004Chapter13,Chlipala2013Chapter5} are a
practical introduction to coinduction in~Coq. The
papers~\cite{Czajka2018,Czajka2014,Czajka2015a} were to a large extent
a motivation for writing the present paper and they contain many
non-trivial coinductive proofs written in the style promoted here. The
article~\cite{KozenSilva2017} has a similar aim to the present paper,
but its approach is quite different. Our style of presenting
coinductive proofs is similar to how such proofs are presented in
e.g.~\cite{EndrullisPolonsky2011,BezemNakataUustalu2012,NakataUustalu2010,LeroyGrall2009,KozenSilva2017}.

\section{A crash-course in coinduction}\label{sec_crash}

In this section we give an elementary explanation of the most common
coinductive techniques. This is generalised and elaborated in more
detail in Section~\ref{sec_coind_tech}. Some of the examples,
definitions and theorems from the present section are leater repeated
and/or generalised in Section~\ref{sec_coind_tech}. This section
strives to strike a balance between generality and ease of
understanding. The explanation given here treats only guarded
corecursive definitions and only guarded proofs, but in practice this
suffices in many cases.

\subsection{Infinite terms and corecursion}

In this section we define many-sorted coterms. We also explain and
justify guarded corecursion using elementary notions.

\begin{definition}\label{def_coterms_a}
  A \emph{many-sorted algebraic signature}
  $\Sigma=\la\Sigma_s,\Sigma_c\ra$ consists of a collection of
  \emph{sort symbols}~$\Sigma_s=\{s_i\}_{i\in I}$ and a collection of
  \emph{constructors} $\Sigma_c=\{c_j\}_{j\in J}$. Each
  constructor~$c$ has an associated \emph{type}
  $\tau(c)=(s_1,\ldots,s_n;s)$ where $s_1,\ldots,s_n,s\in\Sigma_s$. If
  $\tau(c)=(;s)$ then~$c$ is a \emph{constant} of sort~$s$. In what
  follows we use $\Sigma,\Sigma'$, etc., for many-sorted algebraic
  signatures, $s,s'$, etc., for sort symbols, and $f,g,c,d$, etc., for
  constructors.

  The set~$\Tc^\infty(\Sigma)$, or just~$\Tc(\Sigma)$, of
  \emph{coterms over~$\Sigma$} is the set of all finite and infinite
  terms over~$\Sigma$, i.e., all finite and infinite labelled trees
  with labels of nodes specified by the constructors of~$\Sigma$ such
  that the types of labels of nodes agree. More precisely, a term~$t$
  over~$\Sigma$ is a partial function from~$\Nbb^*$ to~$\Sigma_c$
  satisfying:
  \begin{itemize}
  \item $t(\epsilon)\da$, and
  \item if $t(p) = c \in \Sigma_c$ with $\tau(c)=(s_1,\ldots,s_n;s)$
    then
    \begin{itemize}
    \item $t(pi)=d \in \Sigma_c$ with
      $\tau(d)=(s_1',\ldots,s_{m_i}';s_i)$ for $i < n$,
    \item $t(pi)\ua$ for $i \ge n$,
    \end{itemize}
  \item if $t(p)\ua$ then~$t(pi)\ua$ for every $i\in\Nbb$,
  \end{itemize}
  where $t(p)\ua$ means that~$t(p)$ is undefined, $t(p)\da$ means
  that~$t(p)$ is defined, and~$\epsilon \in \Nbb^*$ is the empty
  string. We use obvious notations for coterms, e.g., $f(g(t,s),c)$
  when $c,f,g \in \Sigma_c$ and $t,s \in \Tc(\Sigma)$, and the types
  agree. We say that a term~$t$ \emph{is of sort~$s$} if $t(\epsilon)$
  is a constructor of type $(s_1,\ldots,s_n;s)$ for some
  $s_1,\ldots,s_n\in\Sigma_s$. By~$\Tc_s(\Sigma)$ we denote the set of
  all terms of sort~$s$ from~$\Tc(\Sigma)$.
\end{definition}

\begin{example}
  Let $A$ be a set. Let $\Sigma$ consist of two sorts~$\sfr$
  and~$\dfr$, one constructor~$\cons$ of type $(\dfr,\sfr;\sfr)$ and a
  distinct constant $a \in A$ of sort~$\dfr$ for each element
  of~$A$. Then~$\Tc_\sfr(\Sigma)$ is the set of streams over~$A$. We
  also write $\Tc_\sfr(\Sigma) = A^\omega$ and $\Tc_\dfr(\Sigma) =
  A$. Instead of $\cons(a,t)$ we usually write $a : t$, and we assume
  that~$:$ associates to the right, e.g., $x : y : t$ is $x : (y :
  t)$. We also use the notation $x : t$ to denote the application of
  the constructor for~$\cons$ to~$x$ and~$t$. We define the functions
  $\head : A^\omega \to A$ and $\tail : A^\omega \to A^\omega$ by
  \[
  \begin{array}{rcl}
    \head(a : t) &=& a \\
    \tail(a : t) &=& t
  \end{array}
  \]
  Specifications of many-sorted signatures may be conveniently given
  by coinductively interpreted grammars. For instance, the
  set~$A^\omega$ of streams over a set~$A$ could be specified by
  writing
  \[
  A^\omega \Coloneqq \cons(A, A^\omega).
  \]
  A more interesting example is that of finite and infinite binary trees
  with nodes labelled either with~$a$ or~$b$, and leaves labelled with
  one of the elements of a set~$V$:
  \[
  T \Coloneqq V \parallel a(T, T) \parallel b(T, T).
  \]
  As such specifications are not intended to be formal entities but
  only convenient visual means for describing sets of coterms, we will
  not define them precisely. It is always clear what many-sorted
  signature is meant.
\end{example}

For the sake of brevity we often use $\Tc = \Tc(\Sigma)$ and $\Tc_s =
\Tc_s(\Sigma)$, i.e., we omit the signature~$\Sigma$ when clear from
the context or irrelevant.

\begin{definition}
  The class of \emph{constructor-guarded} functions is defined
  inductively as the class of all functions $h : \Tc_s^m \to \Tc_{s'}$
  (for arbitrary $m \in \Nbb$, $s,s' \in \Sigma_s$) such that there
  are a constructor~$c$ of type $(s_1,\ldots,s_{k};s')$ and functions
  $u_i : \Tc_s^m \to \Tc_{s_i}$ ($i=1,\ldots,k$) such that
  \[
  h(y_1,\ldots,y_m) =
  c(u_1(y_1,\ldots,y_m),\ldots,u_k(y_1,\ldots,y_m))
  \]
  for all $y_1,\ldots,y_m \in \Tc_s$, and for each $i=1,\ldots,k$ one
  of the following holds:
  \begin{itemize}
  \item $u_i$ is constructor-guarded, or
  \item $u_i$ is a constant function, or
  \item $u_i$ is a projection function, i.e., $s_i = s$ and there is
    $1\le j \le m$ with $u_i(y_1,\ldots,y_m) = y_j$ for all
    $y_1,\ldots,y_m \in \Tc_s$.
  \end{itemize}
  Let~$S$ be a set. A function $h : S \times \Tc_s^m \to \Tc_{s'}$ is
  constructor-guarded if for every $x \in S$ the function $h_x :
  \Tc_s^m \to \Tc_{s'}$ defined by $h_x(y_1,\ldots,y_m) =
  h(x,y_1,\ldots,y_m)$ is constructor-guarded. A function $f : S \to
  \Tc_s$ is defined by \emph{guarded corecursion} from $h : S \times
  \Tc_s^m \to \Tc_s$ and $g_i : S \to S$ ($i=1,\ldots,m$) if~$h$ is
  constructor-guarded and~$f$ satisfies
  \[
  f(x) = h(x, f(g_1(x)), \ldots, f(g_m(x)))
  \]
  for all $x \in S$.
\end{definition}

The following theorem is folklore in the coalgebra community. We
sketch an elementary proof. In fact, each set of many-sorted coterms
is a final coalgebra of an appropriate set-functor. Then
Theorem~\ref{thm_corecursion_a} follows from more general
principles. See e.g.~\cite{JacobsRutten2011,Rutten2000} for a more
general coalgebraic explanation of corecursion.

\begin{theorem}\label{thm_corecursion_a}
  For any constructor-guarded function $h : S \times \Tc_s^m \to
  \Tc_{s}$ and any $g_i : S \to S$ ($i=1,\ldots,m$), there exists a
  unique function $f : S \to \Tc_s$ defined by guarded corecursion
  from~$h$ and~$g_1,\ldots,g_m$.
\end{theorem}

\begin{proof}
  Let $f_0 : S \to \Tc_s$ be an arbitrary function. Define~$f_{n+1}$
  for $n \in \Nbb$ by $f_{n+1}(x) = h(x, f_n(g_1(x)), \ldots,
  f_n(g_m(x)))$. Using the fact that~$h$ is constructor-guarded, one
  shows by induction on~$n$ that:
  \begin{itemize}
  \item[$(\star)$] $f_{n+1}(x)(p) = f_n(x)(p)$ for $x \in S$ and $p
    \in \Nbb^*$ with $|p|<n$
  \end{itemize}
  where~$|p|$ denotes the length of~$p$. Indeed, the base is
  obvious. We show the inductive step. Let $x \in S$. Because~$h$ is
  constructor-guarded, we have for instance
  \[
  f_{n+2}(x) = h(x, f_{n+1}(g_1(x)), \ldots, f_{n+1}(g_m(x))) =
  c_1(c_2, c_3(w, f_{n+1}(g_1(x))))
  \]
  Let $p \in \Nbb^*$ with $|p| \le n$. The only interesting case is
  when $p=11p'$, i.e., when~$p$ points inside~$f_{n+1}(g_1(x))$. But
  then $|p'| < |p| \le n$, so by the inductive hypothesis
  $f_{n+1}(g_1(x))(p') = f_n(g_1(x))(p')$. Thus $f_{n+2}(x)(p) =
  f_{n+1}(g_1(x))(p') = f_n(g_1(x))(p') = f_{n+1}(x)(p)$.

  Now we define $f : S \to \Tc_s$ by
  \[
  f(x)(p) = f_{|p|+1}(x)(p)
  \]
  for $x \in S$, $p \in \Nbb^*$. Using~$(\star)$ it is routine to
  check that~$f(x)$ is a well-defined coterm for each $x \in S$. To
  show that~$f : S \to \Tc_s$ is defined by guarded corecursion
  from~$h$ and~$g_1,\ldots,g_m$, using~$(\star)$ one shows by
  induction on the length of~$p \in \Nbb^*$ that for any $x \in S$:
  \[
  f(x)(p) = h(x, f(g_1(x)),\ldots,f(g_m(x)))(p).
  \]
  To prove that~$f$ is unique it suffices to show that it does not
  depend on~$f_0$. For this purpose, using~$(\star)$ one shows by
  induction on the length of~$p \in \Nbb^*$ that~$f(x)(p)$ does not
  depend on~$f_0$ for any $x \in S$.
\end{proof}

We shall often use the above theorem implicitly, just mentioning that
some equations define a function by guarded corecursion.

\begin{example}\label{ex_corec}
  Consider the equation
  \[
  \even(x : y : t) = x : \even(t)
  \]
  It may be rewritten as
  \[
  \even(t) = \head(t) : \even(\tail(\tail(t)))
  \]
  So $\even : A^\omega \to A^\omega$ is defined by guarded corecursion
  from $h : A^\omega \times A^\omega \to A^\omega$ given by
  \[
  h(t,t') = \head(t) : t'
  \]
  and $g : A^\omega \to A^\omega$ given by
  \[
  g(t) = \tail(\tail(t))
  \]
  By Theorem~\ref{thm_corecursion_a} there is a unique function $\even :
  A^\omega \to A^\omega$ satisfying the original equation.

  Another example of a function defined by guarded corecursion is
  $\zip : A^\omega \times A^\omega \to A^\omega$:
  \[
  \zip(x : t, s) = x : \zip(s, t)
  \]
  The following function $\mer : \Nbb^\omega \times \Nbb^\omega \to
  \Nbb^\omega$ is also defined by guarded corecursion:
  \[
  \mer(x : t_1, y : t_2) =
  \left\{
    \begin{array}{cl}
      x : \mer(t_1, y : t_2) & \text{ if } x \le y \\
      y : \mer(x : t_1, t_2) & \text{ otherwise }
    \end{array}
  \right.
 \]
\end{example}

\subsection{Coinduction}

In this section\footnote{This section is largely based
  on~\cite[Sections~2,3]{Czajka2018}
  and~\cite[Section~2]{Czajka2015a}.} we give a brief explanation of a
certain style of writing coinductive proofs. Our presentation of
coinductive proofs is similar to
e.g.~\cite{EndrullisPolonsky2011,BezemNakataUustalu2012,NakataUustalu2010,LeroyGrall2009,KozenSilva2017}.

There are many ways in which coinductive proofs written in our style
can be justified. With enough patience one could, in principle,
reformulate all proofs to directly employ the usual coinduction
principle in set theory based on the Knaster-Tarski fixpoint
theorem~\cite{Sangiorgi2012}. Whenever proofs and corecursive
definitions are guarded, one could formalize them in a proof assistant
based on type theory with a syntactic guardedness check, e.g., in
Coq~\cite{Coquand1993,Gimenez1994}. Non-guarded proofs can be
formalized in recent versions of Agda with sized
types~\cite{AbelPientka2013,Abel2012,Abel2010,Barthe2008}. One may
also use the coinduction principle
from~\cite{KozenSilva2017}. Finally, one can justify coinductive
proofs by indicating how to interpret them in ordinary set theory,
which is what we do in this section.

The purpose of this section is to explain how to justify coinductive
proofs by reducing coinduction to transfinite induction. The present
section does not provide a formal coinduction proof principle as such,
but only indicates how one could elaborate the proofs so as to
eliminate the use of coinduction. Theorem~\ref{thm_coinduction}
provides a formal coinduction principle, though some reformulation is
still needed to use it directly. The status of the present section is
that of a ``meta-explanation'', analogous to an explanation of how,
e.g., the informal presentations of inductive constructions found in
the literature may be encoded in ZFC set theory.

\begin{example}\label{ex_1}
  Let~$T$ be the set of all finite and infinite terms defined
  coinductively by
  \[
  T \Coloneqq V \parallel A(T) \parallel B(T, T)
  \]
  where~$V$ is a countable set of variables, and~$A$, $B$ are
  constructors. By $x,y,\ldots$ we denote variables, and by
  $t,s,\ldots$ we denote elements of~$T$. Define a binary
  relation~$\to$ on~$T$ coinductively by the following rules.
  \[
  \infer=[(1)]{x \to x}{} \quad
  \infer=[(2)]{A(t) \to A(t')}{t \to t'} \quad
  \infer=[(3)]{B(s,t) \to B(s',t')}{s \to s' & t \to t'} \quad
  \infer=[(4)]{A(t) \to B(t',t')}{t\to t'}
  \]

  Formally, the relation~${\to}$ is the greatest fixpoint of a
  monotone function
  \[
  F : \Pow{T \times T} \to \Pow{T \times T}
  \]
  defined by
  \[
  F(R) = \left\{ \la t_1, t_2 \ra \mid \exists_{x \in V}(t_1 \equiv
    t_2 \equiv x) \lor \exists_{t,t'\in T}(t_1 \equiv A(t) \land t_2
    \equiv A(t') \land R(t,t')) \lor \ldots \right\}.
  \]

  Alternatively, using the Knaster-Tarski fixpoint theorem, the
  relation~$\to$ may be characterised as the greatest binary relation
  on~$T$ (i.e. the greatest subset of $T\times T$ w.r.t.~set
  inclusion) such that ${\to} \subseteq F({\to})$, i.e., such that for
  every $t_1,t_2 \in T$ with $t_1 \to t_2$ one of the following holds:
  \begin{enumerate}
  \item $t_1 \equiv t_2 \equiv x$ for some variable $x \in V$,
  \item $t_1 \equiv A(t)$, $t_2 \equiv A(t')$ with $t \to t'$,
  \item $t_1 \equiv B(s,t)$, $t_2 \equiv B(s',t')$ with $s \to s'$ and
    $t \to t'$,
  \item $t_1 \equiv A(t)$, $t_2 \equiv B(t',t')$ with $t \to t'$.
  \end{enumerate}

  Yet another way to think about~$\to$ is that $t_1 \to t_2$ holds if
  and only if there exists a \emph{potentially infinite} derivation
  tree of $t_1 \to t_2$ built using the rules~$(1)-(4)$.

  The rules~$(1)-(4)$ could also be interpreted inductively to yield
  the least fixpoint of~$F$. This is the conventional interpretation,
  and it is indicated with a single line in each rule separating
  premises from the conclusion. A coinductive interpretation is
  indicated with double lines.

  The greatest fixpoint~$\to$ of~$F$ may be obtained by transfinitely
  iterating~$F$ starting with~$T \times T$. More precisely, define an
  ordinal-indexed sequence~$(\to^\alpha)_\alpha$ by:
  \begin{itemize}
  \item $\to^0 = T \times T$,
  \item $\to^{\alpha+1} = F(\to^\alpha)$,
  \item $\to^\lambda = \bigcap_{\alpha<\lambda} \to^\alpha$ for a limit
    ordinal~$\lambda$.
  \end{itemize}
  Then there exists an ordinal~$\zeta$ such that ${\to} =
  {\to^\zeta}$. The least such ordinal is called the \emph{closure
    ordinal}. Note also that ${\to^\alpha} \subseteq {\to^\beta}$ for
  $\alpha \ge \beta$ (we often use this fact implicitly). See
  Section~\ref{sec_prelim} below. The relation~$\to^\alpha$ is called
  the \emph{$\alpha$-approximant} of~$\to$, or the \emph{approximant
    of~$\to$ at stage~$\alpha$}. If $t \to^\alpha s$ then we say that
  $t \to s$ (holds) at (stage)~$\alpha$. Note that the
  $\alpha$-approximants depend on a particular definition of~$\to$
  (i.e.~on the function~$F$), not solely on the relation~$\to$
  itself. We use~$R^\alpha$ for the $\alpha$-approximant of the
  relation~$R$.

  It is instructive to note that the coinductive rules for~$\to$ may
  also be interpreted as giving rules for the $\alpha+1$-approximants,
  for any ordinal~$\alpha$.
  \[
  \infer[(1)]{x \to^{\alpha+1} x}{}\quad
  \infer[(2)]{A(t) \to^{\alpha+1} A(t')}{
    t \to^\alpha t'}\quad
  \infer[(3)]{B(s,t) \to^{\alpha+1} B(s',t')}{
    s \to^\alpha s' & t \to^\alpha t'}\quad
  \infer[(4)]{A(t) \to^{\alpha+1} B(t',t')}{
    t\to^\alpha t'}
  \]

  Usually, the closure ordinal for the definition of a coinductive
  relation is~$\omega$. In general, however, it is not difficult to come
  up with a coinductive definition whose closure ordinal is greater
  than~$\omega$. For instance, consider the relation $R \subseteq \Nbb
  \cup \{\infty\}$ defined coinductively by the following two rules.
  \[
  \infer={R(n+1)}{R(n) & n \in \Nbb} \quad\quad
  \infer={R(\infty)}{\exists n \in \Nbb . R(n)}
  \]
  We have $R = \emptyset$, $R^n = \{m \in \Nbb \mid m \ge n\} \cup
  \{\infty\}$ for $n \in \Nbb$, $R^\omega = \{\infty\}$, and only
  $R^{\omega+1}=\emptyset$. Thus the closure ordinal of this
  definition is $\omega+1$.
\end{example}

Usually, we are interested in proving by coinduction statements of the
form $\psi(R_1,\ldots,R_m)$ where
\[
\psi(X_1,\ldots,X_m) \equiv \forall x_1 \ldots x_n . \varphi(\vec{x})
\to X_1(g_1(\vec{x}),\ldots,g_k(\vec{x})) \land \ldots \land
X_m(g_1(\vec{x}),\ldots,g_k(\vec{x})).
\]
and $R_1,\ldots,R_m$ are coinductive relations on~$T$, i.e, relations
defined as the greatest fixpoints of some monotone functions on the
powerset of an appropriate cartesian product of~$T$, and
$\psi(R_1,\ldots,R_m)$ is~$\psi(X_1,\ldots,X_m)$ with~$R_i$
substituted for~$X_i$. Statements with an existential quantifier may
be reduced to statements of this form by skolemizing, as explained in
Example~\ref{ex_skolem} below.

Here $X_1,\ldots,X_m$ are meta-variables for which relations on~$T$
may be substituted. In the statement~$\varphi(\vec{x})$ only
$x_1,\ldots,x_n$ occur free. The meta-variables $X_1,\ldots,X_m$
\emph{are not allowed to occur} in~$\varphi(\vec{x})$. In general, we
abbreviate $x_1,\ldots,x_n$ with~$\vec{x}$. The
symbols~$g_1,\ldots,g_k$ denote some functions of~$\vec{x}$.

To prove~$\psi(R_1,\ldots,R_m)$ it suffices to show by transfinite
induction that $\psi(R_1^\alpha,\ldots,R_m^\alpha)$ holds for each
ordinal~$\alpha \le \zeta$, where~$R_i^\alpha$ is the
$\alpha$-approximant of~$R_i$. It is an easy exercise to check that
because of the special form of~$\psi$ (in particular because~$\varphi$
does not contain~$X_1,\ldots,X_m$) and the fact that each~$R_i^0$ is
the full relation, the base case~$\alpha=0$ and the case of~$\alpha$ a
limit ordinal hold. They hold for \emph{any}~$\psi$ of the above form,
\emph{irrespective} of $\varphi,R_1,\ldots,R_m$. Note
that~$\varphi(\vec{x})$ is the same in
all~$\psi(R_1^\alpha,\ldots,R_m^\alpha)$ for any~$\alpha$, i.e., it
does not refer to the $\alpha$-approximants or the
ordinal~$\alpha$. Hence it remains to show the inductive step
for~$\alpha$ a successor ordinal. It turns out that a coinductive
proof of~$\psi$ may be interpreted as a proof of this inductive step
for a successor ordinal, with the ordinals left implicit and the
phrase ``coinductive hypothesis'' used instead of ``inductive
hypothesis''.

\begin{example}
  On terms from~$T$ (see Example~\ref{ex_1}) we define the operation
  of substitution by guarded corecursion.
  \[
  \begin{array}{rclcrcl}
    y[t/x] &=& y \quad\text{ if } x \ne y &\quad&
    (A(s))[t/x] &=& A(s[t/x]) \\
    x[t/x] &=& t &\quad&
    (B(s_1,s_2))[t/x] &=& B(s_1[t/x],s_2[t/x])
  \end{array}
  \]
  We show by coinduction: if $s \to s'$ and $t \to t'$ then $s[t/x]
  \to s'[t'/x]$, where~$\to$ is the relation from
  Example~\ref{ex_skolem}. Formally, the statement we show by
  transfinite induction on~$\alpha \le \zeta$ is: for $s,s',t,t' \in
  T$, if $s \to s'$ and $t \to t'$ then $s[t/x] \to^\alpha
  s'[t'/x]$. For illustrative purposes, we indicate the
  $\alpha$-approximants with appropriate ordinal superscripts, but it
  is customary to omit these superscripts.

  Let us proceed with the proof. The proof is by coinduction with case
  analysis on $s \to s'$. If $s \equiv s' \equiv y$ with $y \ne x$,
  then $s[t/x] \equiv y \equiv s'[t'/x]$. If $s \equiv s' \equiv x$
  then $s[t/x] \equiv t \to^{\alpha+1} t' \equiv s'[t'/x]$ (note that
  ${\to} \equiv {\to^\zeta} \subseteq {\to^{\alpha+1}}$). If $s \equiv
  A(s_1)$, $s' \equiv A(s_1')$ and $s_1 \to s_1'$, then $s_1[t/x]
  \to^\alpha s_1'[t'/x]$ by the coinductive hypothesis. Thus $s[t/x]
  \equiv A(s_1[t/x]) \to^{\alpha+1} A(s_1'[t'/x]) \equiv s'[t'/x]$ by
  rule~$(2)$. If $s \equiv B(s_1,s_2)$, $s' \equiv B(s_1',s_2')$ then
  the proof is analogous. If $s \equiv A(s_1)$, $s' \equiv
  B(s_1',s_1')$ and $s_1 \to s_1'$, then the proof is also
  similar. Indeed, by the coinductive hypothesis we have $s_1[t/x]
  \to^\alpha s_1'[t'/x]$, so $s[t/x] \equiv A(s_1[t/x]) \to^{\alpha+1}
  B(s_1'[t'/x],s_1'[t'/x]) \equiv s'[t'/x]$ by rule~$(4)$.
\end{example}

With the following example we explain how proofs of existential
statements should be interpreted.

\begin{example}\label{ex_skolem}
  Let~$T$ and~$\to$ be as in Example~\ref{ex_1}. We want to show: for
  all $s,t,t' \in T$, if $s \to t$ and $s \to t'$ then there exists
  $s' \in T$ with $t \to s'$ and $t' \to s'$. The idea is to skolemize
  this statement. So we need to find a Skolem function $f : T^3 \to T$
  which will allow us to prove the Skolem normal form:
  \begin{itemize}
  \item[$(\star)$] if $s \to t$ and $s \to t'$ then $t \to f(s,t,t')$
    and $t' \to f(s,t,t')$.
  \end{itemize}
  The rules for~$\to$ suggest a definition of~$f$:
  \[
  \begin{array}{rcl}
    f(x, x, x) &=& x \\
    f(A(s), A(t), A(t')) &=& A(f(s,t,t')) \\
    f(A(s),A(t),B(t',t')) &=& B(f(s,t,t'),f(s,t,t')) \\
    f(A(s),B(t,t),A(t')) &=& B(f(s,t,t'),f(s,t,t')) \\
    f(A(s),B(t,t),B(t',t')) &=& B(f(s,t,t'),f(s,t,t')) \\
    f(B(s_1,s_2), B(t_1,t_2), B(t_1',t_2')) &=&
    B(f(s_1,t_1,t_1'),f(s_2,t_2,t_2')) \\
    f(s, t, t') &=& \text{some arbitrary term if none of the above matches}
  \end{array}
  \]
  This is a definition by guarded corecursion, so there exists a
  unique function $f : T^3 \to T$ satisfying the above equations. The
  last case in the above definition of~$f$ corresponds to the case in
  Definition~\ref{def_guarded_corecursion} where all~$u_i$ are
  constant functions. Note that any fixed term has a fixed constructor
  (in the sense of Definition~\ref{def_guarded_corecursion}) at the
  root. In the sense of Definition~\ref{def_guarded_corecursion} also
  the elements of~$V$ are nullary constructors.

  We now proceed with a coinductive proof of~$(\star)$. Assume $s \to
  t$ and $s \to t'$. If $s \equiv t \equiv t' \equiv x$ then
  $f(s,t,t') \equiv x$, and $x \to x$ by rule~$(1)$. If $s \equiv A(s_1)$,
  $t \equiv A(t_1)$ and $t' \equiv A(t_1')$ with $s_1 \to t_1$ and
  $s_1 \to t_1'$, then by the coinductive hypothesis $t_1 \to
  f(s_1,t_1,t_1')$ and $t_1' \to f(s_1,t_1,t_1')$. We have $f(s,t,t')
  \equiv A(f(s_1,t_1,t_1'))$. Hence $t \equiv A(t_1) \to f(s,t,t')$
  and $t \equiv A(t_1') \to f(s,t,t')$, by rule~$(2)$. If $s \equiv
  B(s_1,s_2)$, $t \equiv B(t_1,t_2)$ and $t' \equiv B(t_1',t_2')$,
  with $s_1 \to t_1$, $s_1 \to t_1'$, $s_2 \to t_2$ and $s_2 \to
  t_2'$, then by the coinductive hypothesis we have $t_1 \to
  f(s_1,t_1,t_1')$, $t_1' \to f(s_1,t_1,t_1')$, $t_2 \to
  f(s_2,t_2,t_2')$ and $t_2' \to f(s_2,t_2,t_2')$. Hence $t \equiv
  B(t_1,t_2) \to B(f(s_1,t_1,t_1'),f(s_2,t_2,t_2')) \equiv f(s,t,t')$
  by rule~$(3)$. Analogously, $t' \to f(s,t,t')$ by rule~$(3)$. Other
  cases are similar.

  Usually, it is inconvenient to invent the Skolem function
  beforehand, because the definition of the Skolem function and the
  coinductive proof of the Skolem normal form are typically
  interdependent. Therefore, we adopt an informal style of doing a
  proof by coinduction of a statement
  \[
  \begin{array}{rcl}
    \psi(R_1,\ldots,R_m) &=& \forall_{x_1, \ldots, x_n \in T}
    \,.\, \varphi(\vec{x}) \to \\ &&\quad \exists_{y \in T}
    . R_1(g_1(\vec{x}),\ldots,g_k(\vec{x}), y) \land \ldots \land R_m(g_1(\vec{x}),\ldots,g_k(\vec{x}),y)
  \end{array}
  \]
  with an existential quantifier. We intertwine the corecursive
  definition of the Skolem function~$f$ with a coinductive proof of
  the Skolem normal form
  \[
  \begin{array}{l}
    \forall_{x_1, \ldots, x_n \in T} \,.\, \varphi(\vec{x}) \to
    \\ \quad\quad
    R_1(g_1(\vec{x}),\ldots,g_k(\vec{x}),f(\vec{x}))
    \land \ldots \land R_m(g_1(\vec{x}),\ldots,g_k(\vec{x}),f(\vec{x}))
  \end{array}
  \]
  We proceed as if the coinductive hypothesis
  was~$\psi(R_1^\alpha,\ldots,R_m^\alpha)$ (it really is the Skolem
  normal form). Each element obtained from the existential quantifier
  in the coinductive hypothesis is interpreted as a corecursive
  invocation of the Skolem function. When later we exhibit an element
  to show the existential subformula
  of~$\psi(R_1^{\alpha+1},\ldots,R_m^{\alpha+1})$, we interpret this
  as the definition of the Skolem function in the case specified by
  the assumptions currently active in the proof. Note that this
  exhibited element may (or may not) depend on some elements obtained
  from the existential quantifier in the coinductive hypothesis, i.e.,
  the definition of the Skolem function may involve corecursive
  invocations.

  To illustrate our style of doing coinductive proofs of statements
  with an existential quantifier, we redo the proof done above. For
  illustrative purposes, we indicate the arguments of the Skolem
  function, i.e., we write~$s'_{s,t,t'}$ in place
  of~$f(s,t,t')$. These subscripts $s,t,t'$ are normally omitted.

  We show by coinduction that if $s \to t$ and $s \to t'$ then there
  exists $s' \in T$ with $t \to s'$ and $t' \to s'$. Assume $s \to t$
  and $s \to t'$. If $s \equiv t \equiv t' \equiv x$ then take
  $s'_{x,x,x} \equiv x$. If $s \equiv A(s_1)$, $t \equiv A(t_1)$ and
  $t' \equiv A(t_1')$ with $s_1 \to t_1$ and $s_1 \to t_1'$, then by
  the coinductive hypothesis we obtain~$s'_{s_1,t_1,t_1'}$ with $t_1
  \to s'_{s_1,t_1,t_1'}$ and $t_1' \to s'_{s_1,t_1,t_1'}$. More
  precisely: by corecursively applying the Skolem function to
  $s_1,t_1,t_1'$ we obtain~$s'_{s_1,t_1,t_1'}$, and by the coinductive
  hypothesis we have $t_1 \to s'_{s_1,t_1,t_1'}$ and $t_1' \to
  s'_{s_1,t_1,t_1'}$. Hence $t \equiv A(t_1) \to A(s'_{s_1,t_1,t_1'})$
  and $t \equiv A(t_1') \to A(s'_{s_1,t_1,t_1'})$, by rule~$(2)$. Thus
  we may take $s'_{s,t,t'} \equiv A(s'_{s_1,t_1,t_1'})$. If $s \equiv
  B(s_1,s_2)$, $t \equiv B(t_1,t_2)$ and $t' \equiv B(t_1',t_2')$,
  with $s_1 \to t_1$, $s_1 \to t_1'$, $s_2 \to t_2$ and $s_2 \to
  t_2'$, then by the coinductive hypothesis we
  obtain~$s'_{s_1,t_1,t_1'}$ and~$s'_{s_2,t_2,t_2'}$ with $t_1 \to
  s'_{s_1,t_1,t_1'}$, $t_1' \to s'_{s_1,t_1,t_1'}$, $t_2 \to
  s'_{s_2,t_2,t_2'}$ and $t_2' \to s'_{s_2,t_2,t_2'}$. Hence $t \equiv
  B(t_1,t_2) \to B(s'_{s_1,t_1,t_1'},s'_{s_2,t_2,t_2'})$ by
  rule~$(3)$. Analogously, $t' \to
  B(s'_{s_1,t_1,t_1'},s'_{s_2,t_2,t_2'})$ by rule~$(3)$. Thus we may
  take $s'_{s,t,t'} \equiv
  B(s'_{s_1,t_1,t_1'},s'_{s_2,t_2,t_2'})$. Other cases are similar.

  It is clear that the above proof, when interpreted in the way
  outlined before, implicitly defines the Skolem function~$f$. It
  should be kept in mind that in every case the definition of the
  Skolem function needs to be guarded. We do not explicitly mention
  this each time, but verifying this is part of verifying the proof.
\end{example}

When doing proofs by coinduction the following criteria need to be
kept in mind in order to be able to justify the proofs according to
the above explanations.
\begin{itemize}
\item When we conclude from the coinductive hypothesis that some
  relation~$R(t_1,\ldots,t_n)$ holds, this really means that only its
  approximant~$R^\alpha(t_1,\ldots,t_n)$ holds. Usually, we need to
  infer that the next approximant~$R^{\alpha+1}(s_1,\ldots,s_n)$ holds
  (for some other elements~$s_1,\ldots,s_n$) by
  using~$R^\alpha(t_1,\ldots,t_n)$ as a premise of an appropriate
  rule. But we cannot, e.g., inspect (do case reasoning
  on)~$R^\alpha(t_1,\ldots,t_n)$, use it in any lemmas, or otherwise
  treat it as~$R(t_1,\ldots,t_n)$.
\item An element~$e$ obtained from an existential quantifier in the
  coinductive hypothesis is not really the element itself, but a
  corecursive invocation of the implicit Skolem function. Usually, we
  need to put it inside some constructor~$c$, e.g.~producing~$c(e)$,
  and then exhibit~$c(e)$ in the proof of an existential
  statement. Applying at least one constructor to~$e$ is necessary to
  ensure guardedness of the implicit Skolem function. But we cannot,
  e.g., inspect~$e$, apply some previously defined functions to it, or
  otherwise treat it as if it was really given to us.
\item In the proofs of existential statements, the implicit Skolem
  function cannot depend on the ordinal~$\alpha$. However, this is the
  case as long as we do not violate the first point, because if the
  ordinals are never mentioned and we do not inspect the approximants
  obtained from the coinductive hypothesis, then there is no way in
  which we could possibly introduce a dependency on~$\alpha$.
\end{itemize}

Equality on coterms may be characterised coinductively.

\begin{definition}\label{def_bisimilarity_a}
  Let~$\Sigma$ be a many-sorted algebraic signature, as in
  Definition~\ref{def_coterms_a}. Let $\Tc = \Tc(\Sigma)$. Define
  on~$\Tc$ a binary relation~${=}$ of \emph{bisimilarity} by the
  coinductive rules
  \[
  \infer={f(t_1,\ldots,t_n) = f(s_1,\ldots,s_n)}{t_1 = s_1 & \ldots &
    t_n = s_n}
  \]
  for each constructor $f \in \Sigma_c$.
\end{definition}

It is intuitively obvious that on coterms bisimilary is the same as
identity. The following easy (and well-known) proposition makes this
precise.

\begin{proposition}\label{prop_bisimilarity_a}
  For $t,s \in \Tc$ we have: $t = s$ iff $t \equiv s$.
\end{proposition}

\begin{proof}
  Recall that each term is formally a partial function from~$\Nbb^*$
  to~$\Sigma_c$. We write $t(p) \approx s(p)$ if either both
  $t(p),s(p)$ are defined and equal, or both are undefined.

  Assume $t = s$. It suffices to show by induction of the length of $p
  \in \Nbb^*$ that $\pos{t}{p} = \pos{s}{p}$ or $t(p)\ua,s(p)\ua$,
  where by~$\pos{t}{p}$ we denote the subterm of~$t$ at
  position~$p$. For $p = \epsilon$ this is obvious. Assume $p =
  p'j$. By the inductive hypothesis, $\pos{t}{p'} = \pos{s}{p'}$ or
  $t(p')\ua, s(p')\ua$. If $\pos{t}{p'} = \pos{s}{p'}$ then
  $\pos{t}{p'} \equiv f(t_0,\ldots,t_n)$ and $\pos{s}{p'} \equiv
  f(s_0,\ldots,s_n)$ for some $f \in \Sigma_c$ with $t_i = s_i$ for
  $i=0,\ldots,n$. If $0 \le j \le n$ then $\pos{t}{p} \equiv t_j = s_j
  = \pos{s}{p}$. Otherwise, if $j > n$ or if $t(p')\ua,s(p')\ua$, then
  $t(p)\ua,s(p)\ua$ by the definition of coterms.

  For the other direction, we show by coinduction that for any $t \in
  \Tc$ we have $t = t$. If $t \in \Tc$ then $t \equiv
  f(t_1,\ldots,t_n)$ for some $f \in \Sigma_c$. By the coinductive
  hypothesis we obtain $t_i = t_i$ for $i=1,\ldots,n$. Hence $t = t$
  by the rule for~$f$.
\end{proof}

For coterms $t,s\in \Tc$, we shall theorefore use the notations $t =
s$ and $t \equiv s$ interchangeably, employing
Proposition~\ref{prop_bisimilarity_a} implicitly.

\begin{example}
  Recall the coinductive definitions of~$\zip$ and~$\even$ from
  Example~\ref{ex_corec}.
  \[
  \begin{array}{rcl}
    \even(x : y : t) &=& x : \even(t) \\
    \zip(x : t, s) &=& x : \zip(s, t)
  \end{array}
  \]
  By coinduction we show
  \[
  \zip(\even(t),\even(\tail(t))) = t
  \]
  for any stream $t \in A^\omega$. Let $t \in A^\omega$. Then $t = x :
  y : s$ for some $x, y \in A$ and $s \in A^\omega$. We have
  \[
  \begin{array}{rcl}
    \zip(\even(t),\even(\tail(t))) &=& \zip(\even(x : y : s), \even(y
    : s)) \\
    &=& \zip(x : \even(s), \even(y : s)) \\
    &=& x : \zip(\even(y : s), \even(s)) \\
    &=& x : y : s \quad\text{ (by~CH) }\\
    &=& t
  \end{array}
  \]
  In the equality marked with~(by~CH) we use the coinductive
  hypothesis, and implicitly a bisimilarity rule from
  Definition~\ref{def_bisimilarity_a}.
\end{example}

The above explanation of coinduction is generalised and elaborated in
much more detail in Section~\ref{sec_coind_tech}. The
papers~\cite{Czajka2018,Czajka2014,Czajka2015a} contain many
non-trivial coinductive proofs written in the style promoted
here. Also~\cite{KozenSilva2017} may be helpful as it gives many
examples of coinductive proofs written in a style similar to the one
used here. The book~\cite{Sangiorgi2012} is an elementary introduction
to coinduction and bisimulation (but the proofs there are written in a
different style than here, not referring to the coinductive hypothesis
but explicitly constructing a backward-closed set). The
chapters~\cite{BertotCasteran2004Chapter13,Chlipala2013Chapter5}
explain coinduction in~Coq from a practical viewpoint. A reader
interested in foundational matters should also
consult~\cite{JacobsRutten2011,Rutten2000} which deal with the
coalgebraic approach to coinduction.

\section{Preliminaries}\label{sec_prelim}

In this section we provide the necessary background on order
theory. We also introduce some new or non-standard definitions and
easy lemmas which will be needed in subsequent developments. For more
background on order theory see e.g.~\cite{DaveyPriestley2002}.

\begin{definition}\label{def_order}
  A \emph{partial order} is a pair $\Pb = \la P, \le \ra$ where~$P$ is
  a set and~$\le$ is an antisymmetric, reflexive and transitive binary
  relation on~$P$. We often confuse~$\Pb$ with~$P$ or~$\le$. The
  \emph{dual} of a partial order $\Pb = \la P, \le \ra$ is a partial
  order $\Pb^\op = \la P, \ge \ra$ where $x \ge y$ iff $y \le
  x$. If~$A$ is a set, and~$\Pb_a = \la P_a, \le_a \ra$ is a partial
  order for each $a \in A$, then the \emph{product}~$\prod_{a \in A}
  \Pb_a = \la \prod_{a\in A}P_a, \le\ra$ is a partial order with~$\le$
  defined by: $p \le q$ iff $p(a) \le_a q(a)$ for each $a \in A$. If
  $A = \{a_1,\ldots,a_n\}$ is finite, then we write $\prod_{a \in
    A}P_a=P_{a_1}\times\ldots\times P_{a_n}$. If $P_a=P$ for each
  $a\in A$ then we write $\prod_{a\in A}P=P^A=A\to P$.

  An element $x \in P$ is \emph{maximal} (\emph{minimal}) if there is
  no $y \in P$ with $y > x$ ($y < x$). The set of all maximal
  (minimal) elements of~$P$ is denoted by~$\Max(P)$ ($\Min(P)$). A
  function $f : P \to Q$ is \emph{max-preserving}
  (\emph{min-preserving}) if $f(\Max(P))\subseteq \Max(Q)$
  ($f(\Min(P)) \subseteq\Min(Q)$). The \emph{least element}
  (\emph{greatest element}) of a set $X \subseteq P$ is an element $x
  \in X$ such that $x \le y$ ($x \ge y$) for all $y \in X$. A
  \emph{well-order} is a partial order in which every nonempty subset
  has the least element.

  An \emph{up-set} (\emph{down-set}) is a subset $U \subseteq P$ such
  that if $x \in U$ and $y \ge x$ ($y \le x$) then $y \in U$. A
  \emph{chain} is a subset $C \subseteq P$ satisfying: for all $x, y
  \in C$, $x \le y$ or $y \le x$. A \emph{directed set} in a parital
  order~$P$ is a nonempty subset $D \subseteq P$ such that for all $x,
  y \in D$ there exists~$z$ such that $z \ge x, y$. A
  \emph{bottom}~$\bot$ (\emph{top}~$\top$) of~$P$, is an element
  of~$P$ satisfying $\bot \le x$ ($x \le \top$) for any $x \in P$. We
  sometimes write~$\bot_P$ and~$\top_P$ when ambiguity may arise. An
  \emph{upper bound} (\emph{lower bound}) of a subset $D \subseteq P$
  is an element $x \in P$ such that $x \ge y$ ($x \le y$) for all $y
  \in D$, which we denote $D \le x$ ($x \le D$). A \emph{supremum} or
  \emph{least upper bound} or \emph{join} (\emph{infimum} or
  \emph{greatest lower bound} or \emph{meet}) of a subset $D \subseteq
  P$ is an element $\join D \in P$ ($\meet D \in P$) such that $D \le
  \join D$ ($\meet D \le D$) and for any $s \in P$ with $D \le s$ ($s
  \le D$) we have $s \le \join D$ ($s \ge \meet D$). We sometimes
  denote the supremum of~$D$ by~$\sup D$ and the infimum by~$\inf D$.

  A partial order is \emph{chain-complete} if every chain has a
  supremum. A \emph{complete partial order} (CPO) is a partial order
  with bottom in which every directed set has a supremum. A partial
  order is a \emph{complete lattice} if every set has a supremum.

  A function $f : P \to Q$ between partial orders is \emph{monotone}
  if it preserves the ordering, i.e., $x \le y$ implies $f(x) \le
  f(y)$. A function $f : P \to Q$ between CPOs is \emph{continuous} if
  for every directed set $D \subseteq P$, $f(D)$ is directed and
  $f(\join D) = \join f(D)$. A \emph{fixpoint} of an function~$f : P
  \to P$ on a partial order~$P$ is an element $x \in P$ such that
  $f(x) = x$. The set of all fixpoints of~$f$ is denoted
  by~$\Fix(f)$. The \emph{least fixpoint~$\mu f$} (\emph{greatest
    fixpoint~$\nu f$}) of a function~$f$ is a fixpoint of~$f$
  such that $\mu f \le x$ ($\nu f \ge x$) for any fixpoint~$x$ of~$f$.

  An \emph{initial (final) sequence} of a function~$f$ on~$P$ is an
  ordinal-indexed sequence~$(f^\alpha)_\alpha$ of elements of~$P$
  satisfying:
  \begin{itemize}
  \item $f^0 = \bot$ ($f^0 = \top$),
  \item $f^{\alpha+1} = f(f^\alpha)$,
  \item $f^\lambda = \join_{\alpha<\lambda} f^\alpha$ ($f^\lambda =
    \meet_{\alpha<\lambda} f^\alpha$) for a limit ordinal~$\lambda$.
  \end{itemize}
  A \emph{limit} of an initial (final) sequence of~$f$ is an element
  $x \in P$ for which there exists an ordinal~$\zeta$ such that
  $f^\alpha = x$ for $\alpha \ge \zeta$. The least such~$\zeta$ is
  called the \emph{closure ordinal} of the sequence.

  For an ordinal~$\alpha$, we denote by~$\On(\alpha)$ the set of all
  ordinals~$\le \alpha$.
\end{definition}

The following lemma is folklore.

\begin{lemma}\label{lem_sequence}
  Let $(f^\alpha)_\alpha$ be the initial (final) sequence of a
  monotone function~$f$. Then $f^\alpha \le f^\beta$ ($f^\alpha \ge
  f^\beta$) for $\alpha \le \beta$.
\end{lemma}

\begin{proof}
  Suppose $(f^\alpha)_\alpha$ is the initial sequence of~$f$. The
  proof for the final sequence is dual. We show by induction
  on~$\beta$ that $f^\alpha \le f^\beta$ for all $\alpha \le
  \beta$. The base case $\beta = 0$ is obvious.

  If $\beta = \gamma + 1$ then $f^\beta = f(f^{\gamma})$ and by the
  inductive hypothesis $f^\gamma \ge f^\alpha$ for $\alpha \le
  \gamma$. Hence, it suffices to show $f^\beta \ge f^\gamma$. If
  $\gamma = 0$ then obviously $f^\beta \ge f^\gamma = \bot$. If
  $\gamma = \delta + 1$ then $f^\gamma \ge f^\delta$, and thus
  $f^\beta = f(f^\gamma) \ge f(f^\delta) = f^\gamma$ by the
  monotonicity of~$f$. If~$\gamma$ is a limit ordinal then
  \[
  f^\beta = f(f^\gamma) = f(\join_{\alpha<\gamma}f^\alpha) \ge
  \join_{\alpha<\gamma}f(f^\alpha) = \join_{\alpha<\gamma}f^{\alpha+1}
  = \join_{\alpha<\gamma}f^\alpha = f^\gamma
  \]
  where the inequality follows from the monotonicity of~$f$ and the
  definition of supremum.

  Thus assume~$\beta$ is a limit ordinal. But then by definition
  $f^\beta = \join_{\alpha < \beta} f^\alpha \ge f^\alpha$ for $\alpha
  \le \beta$.
\end{proof}

In the following lemma we collect simple well-known properites of
lattices and CPOs.

\begin{lemma}~
  \begin{itemize}
  \item In a complete lattice each subset has an infimum.
  \item Any complete lattice has the bottom and top elements.
  \item The dual of a complete lattice is also a complete
    lattice.
  \item For any set~$A$, the power set~$\Pow{A}$ is a complete
    lattice.
  \item If~$P_a$ is a CPO for each $a \in A$, then~$\prod_{a \in
    A}P_a$ is a CPO with~$\bot_{\prod_{a \in A}P_a}$ defined by
    $\bot_{\prod_{a \in A}P_a}(a) = \bot_{P_a}$.
  \item Every continuous function is monotone.
  \item Every CPO is chain-complete.
  \end{itemize}
\end{lemma}

It is also true that every chain-complete partial order is
a~CPO~\cite[Theorem~8.11]{DaveyPriestley2002}.

An initial (final) sequence of a function on a partial order need not
exist. Even if it exists, its limit need not exist. However, the
situation is more definite for monotone functions on~CPOs or complete
lattices.

\begin{theorem}\label{thm_fixpoint_cpo}
  Every monotone function~$f$ on a CPO has the least fixpoint~$\mu
  f$. Moreover, $\mu f$ is the limit of the initial sequence of~$f$.
\end{theorem}

\begin{proof}
  See e.g.~\cite[Theorem~10.5 and Exercise~8.19]{DaveyPriestley2002}.
\end{proof}

\begin{theorem}\label{thm_fixpoint_complete_lattice}
  Every monotone function~$f$ on a complete lattice~$L$ has the least
  and greatest fixpoints. Moreover, $\mu f$ is the limit of the
  initial, and~$\nu f$ of the final, sequence of~$f$.
\end{theorem}

\begin{proof}
  The part about~$\mu f$ follows from the previous theorem, because
  every complete lattice is a CPO. The part about~$\nu f$ also follows
  from the previous theorem, by applying it to the dual of~$L$.
\end{proof}

The following theorem implies that every CPO has a maximal element.

\begin{theorem}[Kuratowski-Zorn Lemma]\label{thm_zorn}
  If~$P$ is a partial order in which every non-empty chain has an
  upper bound, then for every $x \in P$ there exists a maximal $y \ge
  x$.
\end{theorem}

\begin{proof}
  See e.g.~\cite[Chapter~10]{DaveyPriestley2002}.
\end{proof}

\begin{lemma}\label{lem_monotone_fixpoint}
  Let $\Ab,\Bb$ be CPOs, and let $F : \Bb^\Ab \to \Bb^\Ab$ be
  monotone. If~$F(f)$ is monotone for each monotone $f \in \Bb^\Ab$,
  then the least fixpoint of~$F$ is monotone.
\end{lemma}

\begin{proof}
  Since~$F$ is monotone, its least fixpoint~$\mu F$ is the limit of
  the initial sequence~$(f^\alpha)_\alpha$ of~$F$. It suffices to show
  by induction on~$\alpha$ that each~$f^\alpha$ is monotone. If
  $\alpha = 0$ then this is obvious, because $f^0(x) = \bot$ for each
  $x \in \Ab$.  For $\alpha=\beta+1$, $f^\alpha = F(f^\beta)$ is
  monotone, because~$f^\beta$ is monotone by the inductive
  hypothesis. Thus let~$\alpha$ be a successor ordinal. Then
  $f^\alpha=\join_{\beta<\alpha}f^\beta$. By Lemma~\ref{lem_sequence},
  $\{f^\beta\mid\beta<\alpha\}$ is a chain in~$\Bb^\Ab$. Thus
  $\{f^\beta(x)\mid\beta<\alpha\}$ is a chain in~$\Bb$ for any
  $x\in\Ab$. Let $x,y\in\Ab$ and $x\le y$. Then $f^\beta(x)\le
  f^\beta(y)$ for $\beta<\alpha$, because~$f^\beta$ is monotone by the
  inductive hypothesis. Hence $f^\beta(x) \le
  \join_{\beta<\alpha}f^\beta(y)$. This holds for any $\beta<\alpha$,
  so
  $\join_{\beta<\alpha}f^\beta(x)\le\join_{\beta<\alpha}f^\beta(y)$. Thus
  $f^\alpha(x) \le f^\alpha(y)$. Therefore~$f^\alpha$ is monotone.
\end{proof}

\section{Coinductive techniques}\label{sec_coind_tech}

In this section we give an elementary presentation of coinductive
techniques.

In Section~\ref{sec_corecursion} we develop a theory to justify
possibly non-guarded corecursive definitions. The approach is to
extend the codomains to sized CPOs (see
Definition~\ref{def_sized_cpo}). In principle, this approach is fairly
general, because any final coalgebra in the category of sets may be
converted into a sized CPO (see the appendix). It is important to note
that the theory is formulated in such a way as to make it unnecessary
in most cases to deal directly with any CPO structure. Usually, to
prove that a function is well-defined by corecursion, it suffices to
show that a certain \emph{prefix production function} $\eta : \Nbb^k
\to \Nbb$ satisfies $\eta(n_1,\ldots,n_k) > \min_{i=1,\ldots,k}n_i$.

In Section~\ref{sec_corecursion_examples} we apply the theory to some
concrete examples. The examples involve many-sorted coterms. We also
develop a style of justifying corecursive definitions. This style is
close enough to our theory to be considered rigorous -- only some
straightforward checks are left implicit.

In Section~\ref{sec_coinduction} we develop a style of doing
coinductive proofs. Some complex examples are presented, with
explanations of how to rigorously justify their correctness.

In Section~\ref{sec_nested_coinduction} we give some examples of
definitions and proofs mixing coinduction with induction, or nesting
coinduction.

As already mentioned, the theory and the results of this section are
not really new. The aim of this section is to give an explanation of
coinduction understandable to a broad audience, and to introduce a
certain style of doing coinductive proofs. For this purpose, we give a
new presentation of ``essentially known'' facts, which may serve as a
reasonably direct justification for coinductive proofs.

\subsection{Corecursion}\label{sec_corecursion}

We are mostly interested in corecursion as a definition method for
functions with a set of possibly infinite objects as codomain. The
following example illustrates the kind of arguments which we want to
make precise.

\begin{example}
  A stream over a set~$A$ is an infinite sequence from~$A^\omega$. For
  $s \in A^\omega$ and $n \in \Nbb$, by~$s_n$ we denote the $n$-th
  element of~$s$. If $a \in A$ and $s \in A^\omega$, then by $a : s$
  we denote the stream~$s$ with~$a$ prepended, i.e., $(a : s)_0 = a$
  and $(a : s)_{n + 1} = s_n$. Consider the equation
  \[
  \even(x : y : t) = x : \even(t)
  \]
  Intuitively, this equation uniquely determines a function~$\even$ on
  streams such that $(\even(s))_n = s_{2n}$. In this simple case,
  using inductive reasoning one could show that the function~$\even$
  defined by $(\even(s))_n = s_{2n}$ is indeed the unique solution of
  the given equation. The problem is how to prove existence and
  uniqueness without finding an explicit definition of the function,
  which is often inconvenient or difficult.

  Informally, one way would be to argue as follows. We show by
  induction that for every $n \in \Nbb$ and any stream~$s$, $\even(s)$
  approximates a stream up to depth~$n$, i.e., at least the first~$n$
  elements of~$\even(s)$ are well-defined. Then it will follow that
  every element of~$\even(s)$ is well-defined, so~$\even(s)$ is a
  stream. For $n = 0$ it is obvious that~$\even(s)$ approximates a
  stream up to depth~$0$. Assume that for every stream~$s$, $\even(s)$
  approximates a stream up to depth~$n$. Let~$s$ be a stream. Since $s
  = x : y : s'$ for some stream~$s'$, we have $\even(s) = x :
  \even(s')$. By the inductive hypothesis, $\even(s')$ approximates a
  stream up to depth~$n$, so~$\even(s)$ approximates a stream up to
  depth~$n+1$.

  Of course, this argument is not rigorous, because we did not
  formally define what it means to approximate a stream up to depth
  $n\in\Nbb$ -- only an informal explanation was given. More formally,
  the proof could be formulated as follows.

  Let $P = A^* \cup A^\omega$ be ordered by~$\qle$ where: $s \qle s'$
  iff~$s$ is a prefix of~$s'$. One easily checks that $\la P, \qle\ra$
  is a CPO. For $s \in P$, by $|s| \in \Nbb \cup \{\infty\}$ we denote
  the length of~$s$. The function $F : P^{A^\omega} \to P^{A^\omega}$
  defined for $f \in P^{A^\omega}$, $s \in A^\omega$ by
  \[
  F(f)(s) = x : f(s') \text{ where } s = x : y : s'
  \]
  is monotone. Therefore, by Theorem~\ref{thm_fixpoint_cpo} it has the
  least fixpoint~$\even$. By induction we show that for every $n \in
  \Nbb$, $|\even(s)| \ge n$ for any $s \in A^\omega$. This is obvious
  for $n = 0$. Assume $|\even(s)| \ge n$ for every $s \in
  A^\omega$. Let $s \in A^\omega$. We have $\even(s) = F(\even)(s) = x
  : \even(s')$ where $s = x : y : s'$. From this and the inductive
  hypothesis we obtain $|\even(s)| \ge n + 1$. Therefore
  $|\even(s)|=\infty$ for every $s \in A^\omega$. Hence $\even \in
  A^\omega \to A^\omega$, i.e., it is maximal
  in~$P^{A^\omega}$. Since~$\even$ is maximal and it is the least
  fixpoint of~$F$, it must be the unique fixpoint of~$F$. Because
  every solution (in $A^\omega \to A^\omega$) of
  \[
  \even(x : y : s) = x : \even(s)
  \]
  is a fixpoint of~$F$, we conclude that this equation has a unique
  solution in $A^\omega \to A^\omega$ (namely, the fixpoint~$\even$
  of~$F$). \hfill$\Box$
\end{example}

In what follows we develop a theory which generalizes the above kind
of reasoning. To formulate the theory, we introduce a CPO structure on
each set of infinite objects we are interested in. The original
objects are maximal elements of the~CPO, with other elements of
the~CPO being their ``approximations''.

More specifically, let~$A$ and~$B$ be sets. We are interested in the
existence of a unique fixpoint $f : A \to B$ of a function $F:B^A\to
B^A$. The strategy for finding~$f$ is to find a CPO~$\Bb$ and a
monotone function \mbox{$F^+ : \Bb^A \to \Bb^A$} such that $\Max(\Bb)
= B$, $F^+(g)(x) = F(g)(x)$ for $x \in A$ and $g \in B^A$ (i.e.~$F^+$
agrees with~$F$ on maximal elements of~$\Bb^A$), and the least
fixpoint~$f$ of~$F^+$ is in~$B^A$ (i.e.~it is maximal in
$\Bb^A$). Then~$f$ is the unique fixpoint of~$F^+$, so it is also the
unique fixpoint of~$F$, because any fixpoint of~$F$ is a fixpoint
of~$F^+$. To show that the least fixpoint of~$F^+$ is maximal, we need
a notion of the size of an element of a CPO. This leads to the
following definition.

\begin{definition}\label{def_sized_cpo}
  A \emph{sized CPO} is a tuple $\la \Ab, \zeta, s, \tcut \ra$
  where~$\Ab$ is a~CPO, $\zeta$ is a \emph{size ordinal}, $s : \Ab \to
  \On(\zeta)$ is a \emph{size function}, and $\tcut : \On(\zeta)
  \times \Ab \to \Ab$ is a \emph{cut function}, such that the
  following conditions are satisfied for $x \in \Ab$ and $\alpha \le
  \zeta$:
  \begin{enumerate}
  \item $s$ is surjective and continuous,
  \item $s(x) = \zeta$ iff $x \in \Ab$ is maximal,
  \item $\tcut$ is monotone in both arguments,
  \item $s(\tcut(\alpha,x)) = \alpha$ if $s(x) > \alpha$,
  \item $\tcut(\alpha,x) = x$ if $s(x) \le \alpha$.
  \end{enumerate}
  Usually, we confuse a sized CPO with its underlying CPO. Thus e.g.~by
  a function between sized CPOs we just mean a function between their
  underlying CPOs. We say that a CPO~$\Ab$ is a sized CPO if there
  exists a size ordinal~$\zeta$, a size function $s : \Ab \to
  \On(\zeta)$ and a cut function $\tcut : \On(\zeta) \times \Ab \to
  \Ab$ such that $\la \Ab, \zeta, s, \tcut \ra$ is a sized CPO. Given
  a sized CPO~$\Ab$ we use~$\zeta_\Ab$ for its associated size
  ordinal, $s_\Ab$ for the associated size function, and~$\tcut_\Ab$
  for the associated cut function. We often drop the subscripts when
  clear from the context.

  Let~$S$ be a nonempty set. The \emph{flat sized CPO}~$S_\bot$ on~$S$
  is defined as $\la \la S \cup \{\bot\}, \le \ra, 1, s, \tcut\ra$
  where the following holds for $x,y \in S_\bot$:
  \begin{itemize}
  \item $x \le y$ iff $x = \bot$ or $x = y$,
  \item $s(x) = 1$ if $x \ne \bot$, $s(\bot) = 0$,
  \item $\tcut(0,x) = \bot$, $\tcut(1,x) = x$.
  \end{itemize}
  It is not difficult to check that~$S_\bot$ is indeed a sized CPO.

  Let $\Ab,\Bb$ be CPOs and $A,B$ their sets of maximal elements.  For
  $f^* : \Ab \to \Bb$, the \emph{restriction} $\cut{f^*}{A} : A \to
  \Bb$ of~$f^*$ is defined by $\cut{f^*}{A}(x)=f^*(x)$ for $x \in
  A$. Then~$f^*$ is an \emph{extension} of~$\cut{f^*}{A}$. A function
  between CPOs is \emph{regular} if it is monotone and
  max-preserving. Let~$S$ be an arbitrary set. A function $f : S
  \times \Ab \to \Bb$ is regular if\footnote{By $\lambda y . f(x, y)$
    we denote a function $f' : \Ab \to \Bb$ defined by
    $f'(y)=f(x,y)$. We will sometimes use the lambda notation in what
    follows.} $\lambda y . f(x,y)$ is regular for each $x \in S$.
\end{definition}

Intuitively, in a sized CPO~$\Ab$ the cut function~$\tcut(\alpha,x)$
``cuts'' an element~$x$ of size $>\alpha$ to its approximation of
size~$\alpha$, i.e., $\tcut(\alpha,x) \le x$ for every $x \in
\Ab$. Indeed, let $x \in \Ab$. If $s(x) \le \alpha$ then
$\tcut(\alpha,x) = x \le x$. So assume $s(x) > \alpha$. Then
$\tcut(\alpha,x) \le \tcut(s(x),x) = x$.

In the rest of this section we assume that $S,Q,\ldots$ are arbitrary
sets, and $\Ab,\Bb,\ldots$ are sized CPOs, and $A,B,\ldots$ are their
corresponding sets of maximal elements, unless otherwise stated.

\begin{lemma}\label{lem_unique_fixpoint}
  Suppose $F : \Ab^S \to \Ab^S$ is a monotone function satisfying
  \[
  \min_{x\in S}s(F(g)(x)) > \min_{x \in S} s(g(x))
  \]
  for each non-maximal $g \in \Ab^S$. Then~$F$ has a unique fixpoint.
  Moreover, this fixpoint is maximal (i.e.~a member of~$A^S$).
\end{lemma}

\begin{proof}
  Because~$F$ is monotone, by Theorem~\ref{thm_fixpoint_cpo} it has
  the least fixpoint~$f$. It suffices to show that $f \in A^S$. Assume
  otherwise. Then~$f$ is not maximal, so
  \[
  \min_{x \in S}s(f(x)) < \min_{x \in S}s(F(f)(x)) = \min_{x \in
    S}s(f(x)).
  \]
  Contradiction.
\end{proof}

\begin{lemma}\label{lem_corecursion_regular}
  Let~$\Ab$ be a CPO and~$\Bb$ a sized CPO. Let
  $h:\Ab\times\Bb^m\to\Bb$ and $g_i : \Ab \to \Ab$ ($i=1,\ldots,m$) be
  regular. Suppose
  \[
  (\star)\quad\quad s(h(x,y_1,\ldots,y_m)) > \min_{i=1,\ldots,m}s(y_i)
  \]
  for all $x \in A$ and all $y_1,\ldots,y_m \in \Bb$ with some~$y_k$
  non-maximal. Then there exists the least fixpoint~$f^*$ of a
  function $F^* : \Bb^\Ab\to\Bb^\Ab$ defined by
  \[
  F^*(f)(x) = h(x, f(g_1(x)),\ldots,f(g_m(x)))
  \]
  for $f \in \Bb^\Ab$ and $x \in \Ab$. Moreover, $f^*$ is regular and
  $\cut{f^*}{A} \in B^A$ is the unique function in~$B^A$ satisfying
  \[
  \cut{f^*}{A}(x) = h(x, \cut{f^*}{A}(g_1(x)), \ldots,
  \cut{f^*}{A}(g_m(x)))
  \]
  for $x \in A$.
\end{lemma}

\begin{proof}
  Since~$h$ is monotone, so is~$F^*$. Indeed, assume $f \le f'$ where
  $f,f' \in \Bb^\Ab$. To show $F^*(f) \le F^*(f')$ it suffices to
  prove $F^*(f)(x) \le F^*(f')(x)$ for $x \in \Ab$. But this follows
  from $f \le f'$ and the monotonicity of~$h$. Therefore, since~$F^*$
  is monotone, by Theorem~\ref{thm_fixpoint_cpo} it has the least
  fixpoint~$f^*$.

  Let $F : \Bb^A \to \Bb^A$ be defined by $F(f)(x) = h(x,
  f(g_1(x)),\ldots,f(g_m(x)))$. Note that indeed $F(f) \in \Bb^A$ for
  $f \in \Bb^A$, because each~$g_i$ is max-preserving.

  We show that for non-maximal $f \in \Bb^A$ we have $\min_{x\in
    A}s(F(f)(x)) > \min_{x\in A}s(f(x))$. Let $f \in \Bb^A$ be
  non-maximal. Let $A'\subseteq A$ be the set of all $x \in A$ such
  that~$f(g_i(x))$ is not maximal for some~$i$.

  First assume $A'=\emptyset$, i.e., $f(g_i(x))$ is maximal for all
  $i=1,\ldots,m$ and all $x\in A$. Then $F(f)(x) =
  h(x,f(g_1(x)),\ldots,f(g_m(x)))$ is maximal for $x\in A$,
  because~$h$ is max-preserving. Hence
  \[
  \min_{x\in A}s(F(f)(x)) = \zeta > \min_{x\in A}(s(f(x)))
  \]
  because~$F(f)(x)$ is maximal for all $x\in A$, but there is $x\in A$
  for which~$f(x)$ is not maximal.

  Thus assume $A'\ne\emptyset$. Since, for $x\in A$,
  $s(h(x,f(g_1(x)),\ldots,f(g_m(x))))=\zeta$ if~$f(g_i(x))$ is
  maximal for all $i\in I$, and $A\ne\emptyset$, we have
  \[
  \min_{x \in A} s(h(x,f(g_1(x)),\ldots,f(g_m(x))))=\min_{x \in A'}
  s(h(x,f(g_1(x)),\ldots,f(g_m(x))))
  \]
  Hence
  \[
  \begin{array}{rcl}
    \min_{x\in A}s(F(f)(x)) &=& \min_{x \in A} s(h(x,f(g_1(x)),\ldots,f(g_m(x)))) \\
    &=& \min_{x \in A'} s(h(x,f(g_1(x)),\ldots,f(g_m(x)))) \\
    &>& \min_{x \in A'} \min_{i=1,\ldots,m} s(f(g_i(x))) \\
    &\ge& \min_{x \in A} s(f(x))
  \end{array}
  \]
  where the strict inequality follows from~$(\star)$.

  Therefore, for non-maximal $f \in \Bb^A$ we have $\min_{x\in
    A}s(F(f)(x)) > \min_{x\in A}s(f(x))$. Thus by
  Lemma~\ref{lem_unique_fixpoint} the function~$F$ has a unique
  fixpoint~$u$. Recall that~$f^*$ is the least fixpoint of~$F^*$. Note
  that~$\cut{f^*}{A}$ is a fixpoint of~$F$. Indeed, for $x \in A$ we
  have
  \[
  f^*(x) = F^*(f^*)(x) = h(x, f^*(g_1(x)),\ldots,f^*(g_m(x))) =
  F(f^*)(x).
  \]
  Therefore, $\cut{f^*}{A} = u$, so it is the unique function in~$B^A$
  satisfying
  \[
  \cut{f^*}{A}(x) = h(x, \cut{f^*}{A}(g_1(x)), \ldots,
  \cut{f^*}{A}(g_m(x)))
  \]
  for $x \in A$.

  It remains to check that~$f^*$ is regular. Since $\cut{f^*}{A} \in
  B^A$, the function~$f^*$ is max-preserving. Because~$h$ and
  all~$g_i$ are monotone, for monotone~$f$ the function~$F^*(f)$ is
  monotone. By Lemma~\ref{lem_monotone_fixpoint} we thus conclude
  that~$f^*$ is monotone. Hence~$f^*$ is regular.
\end{proof}

\begin{corollary}\label{cor_general_corecursion}
  Let $h:S\times\Bb^m\to\Bb$ be regular. Let $g_i : S \to S$
  ($i=1,\ldots,m$). Suppose
  \[
  (\star)\quad\quad s(h(x,y_1,\ldots,y_m)) > \min_{i=1,\ldots,m}s(y_i)
  \]
  for all $x\in S$ and all $y_1,\ldots,y_m\in\Bb$ with some~$y_k$
  non-maximal. Then there exists a unique function $f : S \to B$
  satisfying
  \[
  f(x) = h(x, f(g_1(x)),\ldots,f(g_m(x)))
  \]
  for $x \in S$.
\end{corollary}

\begin{proof}
  Let $\Sb = S_\bot$ be the flat CPO on~$S$. There exists a regular
  extension $g_i^* : \Sb \to \Sb$ of each~$g_i$, defined by
  \[
  g_i^*(x) =
  \left\{
    \begin{array}{cl}
      x & \text{ if } x \in S \\
      \bot & \text{ otherwise }
    \end{array}
  \right.
  \]
  for $x \in \Sb$. Analogously, there exists a regular $h^* : \Sb
  \times \Bb^m \to \Bb$ defined by
  \[
  h^*(x,\bar{y}) =
  \left\{
    \begin{array}{cl}
      h(x,\bar{y}) & \text{ if } x \in S \\
      \bot & \text{ otherwise }
    \end{array}
  \right.
  \]
  for $x \in \Sb$ and $\bar{y} \in \Bb^m$. Moreover, $h^*$
  satisfies~$(\star)$ in
  Lemma~\ref{lem_corecursion_regular}. Therefore, we may apply
  Lemma~\ref{lem_corecursion_regular} to obtain the required
  function~$f$.
\end{proof}

At this point it is worthwhile to emphasize one aspect of our
approach. Ultimately, we really only care about the maximal elements
in a CPO, and only about functions between sets of maximal
elements. That we introduce a structure of a CPO is only to be able to
rigorously justify certain methods for defining corecursive
functions. But once these methods have been shown correct, to apply
them we usually do not need to directly deal with the CPO structure at
all. The following makes this more apparent.

\begin{definition}\label{def_corecursion}
  A function $f : S \to Q$ is \emph{defined by substitution} from $h :
  Q_1\times\ldots\times Q_m \to Q$ and $g_i:S\to Q_i$ ($i=1,\ldots,m$)
  if $f(x) = h(g_1(x),\ldots,g_m(x))$ for $x\in S$. A function $f : S
  \to Q$ is \emph{defined by cases} from functions $g_i:S\to Q$ and
  condition functions $h_i:S\to\{0,1\}$ for $i=1,\ldots,m$, if for $x
  \in S$:
  \begin{itemize}
  \item $f(x) = g_i(x)$ if $h_i(x) = 1$,
  \item $f(x) = g_0(x)$ if $h_i(x) = 0$ for all $i=1,\ldots,m$,
  \item there is no $x \in S$ with $h_i(x) = h_j(x) = 1$ for $i \ne
    j$.
  \end{itemize}
  A function $f : S \to Q$ is \emph{defined by corecursion} from $h :
  S \times Q^m \to Q$ and $g_i : S \to S$ ($i=1,\ldots,m$) if it is
  the unique function in~$Q^S$ satisfying
  \[
  f(x) = h(x,f(g_1(x)),\ldots,f(g_m(x)))
  \]
  for all $x \in S$. We say that~$h$ is a \emph{prefix function}
  for~$f$, and each~$g_i$ is an \emph{argument function} for~$f$.
  Note that given~$h$ and~$g_i$, there might not exist any function
  defined by corecursion from~$h$ and~$g_i$.

  A \emph{production function}
  $\eta_f:\On(\zeta_{\Ab_1})\times\ldots\times\On(\zeta_{\Ab_n})\to
  \On(\zeta_\Bb)$ for \mbox{$f:A_1\times\ldots\times A_n\to B$} is any
  function satisfying
  \[
  \eta_f(s(x_1),\ldots,s(x_n)) = s(f^*(x_1,\ldots,x_n))
  \]
  for $x_i \in \Ab_i$ ($i=1,\ldots,n$), where
  $f^*\in\Ab_1\times\ldots\times\Ab_n\to\Bb$ is a regular extension
  of~$f$. We then also say that~$\eta_f$ is a production function
  for~$f^*$, or that~$f^*$ is \emph{associated with}~$\eta_f$.  If a
  production function~$\eta_f$ for $f : A_1\times\ldots\times A_n \to
  B$ is clear from the context, then we use~$f^*$ to denote the
  regular extension of~$f$ associated with~$\eta_f$.

  Any production function~$\eta_h$ for a prefix function~$h$ for~$f$
  is called a \emph{global prefix production function for~$f$}. If $x
  \in S$ and $h : S \times B^m \to B$ is a prefix function for $f : S
  \to B$, then any production function~$\eta_h^x : \On(\zeta_\Bb)^m
  \to \On(\zeta_\Bb)$ for the \emph{$x$-local prefix function}
  $\lambda \bar{y} . h(x,\bar{y})$ is called an \emph{$x$-local prefix
    production function for~$f$}. We use the term \emph{prefix
    production function} for either a local or a global prefix
  production function, depending on the context.
\end{definition}

\begin{lemma}\label{lem_prodfun_regular}
  Any production function $\eta_f : \On(\zeta_{\Ab_1}) \times \ldots
  \times \On(\zeta_{\Ab_n}) \to \On(\zeta_\Bb)$ for a function $f :
  A_1 \times \ldots \times A_n \to B$ is regular.
\end{lemma}

\begin{proof}
  Let $f^* : \Ab_1\times\ldots\times\Ab_n\to\Bb$ be the regular
  extension of~$f$ associated with~$\eta_f$. Let $\alpha_i \le \beta_i
  \le \zeta_{\Ab_i}$ for $i=1,\ldots,n$. Because the size functions
  for each~$\Ab_i$ are surjective, for every $i=1,\ldots,n$ there is
  $y_i \in \Ab_i$ such that $s(y_i) = \beta_i$. Let $x_i =
  \tcut(\alpha_i,y_i)$. Because of the monotonicity of the cut
  function we have $x_i \le \tcut(\beta_i,y_i) = y_i$. Also $s(x_i) =
  \alpha_i$ by the definition of~$\tcut$. Hence
  \[
  \begin{array}{rcl}
    \eta_f(\alpha_1,\ldots,\alpha_n) &=& \eta_f(s(x_1),\ldots,s(x_n)) \\
    &=& s(f^*(x_1,\ldots,x_n)) \\
    &\le& s(f^*(y_1,\ldots,y_n)) \\
    &=& \eta_f(s(y_1),\ldots,s(y_n)) \\
    &=& \eta_f(\beta_1,\ldots,\beta_n)
  \end{array}
  \]
  where the inequality follows from the fact that~$f^*$ and~$s$ are
  monotone. Therefore~$\eta_f$ is monotone.

  To show that $\eta_f$ is max-preserving, we need to prove
  $\eta_f(\zeta_{\Ab_1},\ldots,\zeta_{\Ab_n})=\zeta_{\Bb}$. Let
  $x_i\in A_i$ for $i=1,\ldots,n$. Then $f^*(x_1,\ldots,x_n)$ is
  maximal, because~$f^*$ is max-preserving. Thus \(
  \eta_f(\zeta_{\Ab_1},\ldots,\zeta_{\Ab_n}) =
  \eta_f(s(x_1),\ldots,s(x_n)) = s(f^*(x_1,\ldots,x_n)) =
  \zeta_{\Bb}. \)
\end{proof}

The following corollary implies that to determine whether there exists
a function defined by corecursion it suffices to bound the values of
local prefix production functions. Thus no analysis of the underlying
CPO structure is needed, as long as we are able to calculate the
production functions.

\begin{corollary}\label{cor_unique_solution}
  Let $h : S \times B^m \to B$ and $g_i : S \to S$
  ($i=1,\ldots,m$). Suppose for each $x \in S$, a function~$\eta_h^x$
  is an $x$-local prefix production function, i.e., a production
  function for $\lambda \bar{y} . h(x,\bar{y})$. Assume
  \[
  (\star)\quad\quad\eta_h^x(\alpha_1,\ldots,\alpha_m) > \min_{i=1,\ldots,m}\alpha_i
  \]
  for each $x \in S$ and all $\alpha_1,\ldots,\alpha_m\le\zeta_\Bb$
  such that $\alpha_k < \zeta_\Bb$ for some $1 \le k \le m$. Then
  there exists a function defined by corecursion from~$h$ and~$g_i$
  ($i=1,\ldots,m$), i.e., a unique function $f : S \to B$ satisfying
  \[
  f(x) = h(x, f(g_1(x)), \ldots, f(g_m(x)))
  \]
  for all $x \in S$.
\end{corollary}

\begin{proof}
  Follows from Corollary~\ref{cor_general_corecursion}.
\end{proof}

Note that for any function $f : A_1\times\ldots\times A_n \to B$ there
exists a production function. Simply take the function~$\eta_f$
defined by:
\[
\eta_f(\alpha_1,\ldots,\alpha_n) =
\left\{
\begin{array}{cl}
  \zeta_\Bb & \text{ if } \alpha_i = \zeta_{\Ab_i} \text{ for } i=1,\ldots,n \\
  0 & \text{ otherwise }
\end{array}
\right.
\]
Then the regular function~$f^* : \Ab_1 \times \ldots \times \Ab_n \to
\Bb$ associated with~$\eta_f$ is defined by
\[
f^*(x_1,\ldots,x_n) =
\left\{
\begin{array}{cl}
  f(x_1,\ldots,x_n) & \text{ if } x_i \in A_i \text{ for } i=1,\ldots,n \\
  \bot & \text{ otherwise }
\end{array}
\right.
\]
The point is to be able to find ``sensible'' production functions, and
then use them to verify~$(\star)$ in
Corollary~\ref{cor_unique_solution}. Below we show how to compute
production functions for functions defined by substitution, cases or
corecursion.

Let $\Two = \{0,1\}_\bot$ be the flat sized CPO on $\{0,1\}$. In what
follows we assume that~$\Two$ is the sized CPO associated with
$\{0,1\}$, e.g., a production function for $f : A \to \{0,1\}$ is
assumed to have $\On(1) = \{0,1\}$ as its codomain. Recall that
$s_\Two(0)=s_\Two(1) = 1$ and $s_\Two(\bot) = 0$.

\begin{lemma}\label{lem_production_functions}~
  \begin{itemize}
  \item The function $\eta(\alpha_1,\ldots,\alpha_n) = \alpha_i$ is a
    continuous production function for the $i$-th projection function
    $\pi_i : A_1 \times \ldots \times A_n \to A_i$ defined by
    $\pi_i(x_1,\ldots,x_n) = x_i$.
  \item The function $\eta(\alpha) = \alpha$ is a continuous
    production function for the identity function $\id : A \to A$.
  \item The function $\eta : \{0,1\}^n \to \{0,1\}$ defined by
    $\eta(\alpha_1,\ldots,\alpha_n)=\min_{i=1,\ldots,n}\alpha_i$ is a
    continuous production function for any function $f : \{0,1\}^m \to
    \{0,1\}$.
  \end{itemize}
\end{lemma}

\begin{proof}
  Follows from definitions.
\end{proof}

\begin{lemma}\label{lem_substitution}
  If a function $f : A_1\times\ldots\times A_n \to B$ is defined by
  substitution from functions $h : B_1\times\ldots\times B_m\to B$ and
  $g_i : A_1\times\ldots\times A_n \to B_i$ ($i=1,\ldots,m$), and
  $\eta_h$ and $\eta_{g_i}$ are production functions for~$h$ and~$g_i$
  respectively, then the function~$\eta_f$ defined by
  \[
  \eta_f(\alpha_1,\ldots,\alpha_n) =
  \eta_h(\eta_{g_1}(\alpha_1,\ldots,\alpha_n),\ldots,\eta_{g_m}(\alpha_1,\ldots,\alpha_n))
  \]
  is a production function for~$f$. Moreover, if~$\eta_h$ and
  all~$\eta_{g_i}$ are continuous, then so is~$\eta_f$.
\end{lemma}

\begin{proof}
  Follows directly from definitions.
\end{proof}

\begin{corollary}\label{cor_permutation}~
  \begin{itemize}
  \item If $\eta_g$ is a (continuous) production
    function for $g: A^m\to B$, then
    \[
    \eta(\alpha_1,\ldots,\alpha_m) =
    \eta_g(\alpha_{\tau(1)},\ldots,\alpha_{\tau(m)})
    \]
    is a (continuous) production function for $f : A^m \to B$ defined
    by
    \[
    f(x_1,\ldots,x_m)=g(x_{\tau(1)},\ldots,x_{\tau(m)})
    \]
    where $\tau : \{1,\ldots,m\} \to \{1,\ldots,m\}$.
  \item If $\eta_g$ is a (continuous) production function for $g :
    A_1\times\ldots\times A_n \to B$, then
    \[
    \eta(\alpha_1,\ldots,\alpha_n,\beta_1,\ldots,\beta_k) =
    \eta_g(\alpha_1,\ldots,\alpha_n)
    \]
    is a (continuous) production function for $f :
    A_1\times\ldots\times A_n\times B_1 \times\ldots\times B_k \to B$
    defined by $f(x_1,\ldots,x_n,y_1,\ldots,y_k)=g(x_1,\ldots,x_n)$
    for $x_i\in A_i$ ($i=1,\ldots,n$), $y_i\in B_i$ ($i=1,\ldots,k$).
  \end{itemize}
\end{corollary}

\begin{proof}
  Follows from the first point of Lemma~\ref{lem_production_functions}
  and from Lemma~\ref{lem_substitution}.
\end{proof}

\begin{lemma}\label{lem_cases}
  If a function $f : A_1\times\ldots\times A_n \to B$ is defined by
  cases from $g_i : A_1\times\ldots\times A_n \to B$ ($i=0,\ldots,m$)
  and $h_i : A_1\times\ldots\times A_n \to \{0,1\}$ ($i=1,\ldots,m$),
  and $\eta_{g_i}$ is a production function for~$g_i$, and
  $\eta_{h_i}$ is a production function for~$h_i$, then the function
  $\eta_f$ defined by
  \[
  \eta_f(\alpha_1,\ldots,\alpha_n) = \left\{
    \begin{array}{cl}
      \displaystyle{\min_{i=0,\ldots,m} \eta_{g_i}(\alpha_1,\ldots,\alpha_n)} &
      \text{ if } \eta_{h_i}(\alpha_1,\ldots,\alpha_n) = 1 \text{ for every } i=1,\ldots,m \\
      0 & \text{ otherwise }
    \end{array}
  \right.
  \]
  is a production function for~$f$. Moreover, if all~$\eta_{g_i}$ and
  all~$\eta_{h_i}$ are continuous, then so is~$\eta_f$.
\end{lemma}

\begin{proof}
  Let $g_i^* : \Ab_1\times\ldots\times\Ab_n \to \Bb$ and \mbox{$h_i^*
    : \Ab_1\times\ldots\times\Ab_n \to \Two$} be the regular
  extensions associated with~$\eta_{g_i}$ and~$\eta_{h_i}$
  respectively. Define \mbox{$f^* : \Ab_1\times\ldots\times\Ab_n \to
    \Bb$} by
  \[
  f^*(x_1,\ldots,x_n) =
  \left\{
    \begin{array}{cl}
      \tcut(\kappa(x_1,\ldots,x_n), g_k^*(x_1,\ldots,x_n)) &
      \text{ if } h_i^*(x_1,\ldots,x_n) \ne \bot \text{ for } 1
      \le i \le m, \\
      & \text{ and } k \text{ is least s.t. } h_k^*(x_1,\ldots,x_n) = 1 \\
      \tcut(\kappa(x_1,\ldots,x_n), g_0^*(x_1,\ldots,x_n)) & \text{ if } h_i^*(x_1,\ldots,x_n) = 0 \text{ for } 1 \le
      i \le m \\
      \bot & \text{ otherwise }
    \end{array}
  \right.
  \]
  where $x_i \in \Ab_i$ ($i=1,\ldots,n$) and $\kappa(x_1,\ldots,x_n) =
  \min_{i=0,\ldots,m} \eta_{g_i}(s(x_1),\ldots,s(x_n))$. One easily
  checks that~$f^*$ is an extension of~$f$. Hence~$f^*$ is
  max-preserving. To show that~$f^*$ is regular it thus suffices to
  check that it is monotone. Assume $x_i \le y_i$ for
  $i=1,\ldots,n$. We need to show $f^*(x_1,\ldots,x_n) \le
  f^*(y_1,\ldots,y_n)$. If $f^*(x_1,\ldots,x_n) = \bot$ then this is
  obvious. So assume, e.g., $h_k^*(x_1,\ldots,x_n) = 1$ and
  $h_i^*(x_1,\ldots,x_n) \ne \bot$ for $i=1,\ldots,m$. Then
  $h_i^*(y_1,\ldots,y_n) = h_i^*(x_1,\ldots,x_n)$ for $i=1,\ldots,m$,
  because each~$h_i^*$ is monotone. Thus it suffices to show
  $\tcut(\kappa(x_1,\ldots,x_n), g_k^*(x_1,\ldots,x_n)) \le
  \tcut(\kappa(y_1,\ldots,y_n),
  g_k^*(y_1,\ldots,y_n))$. Because~$g_i^*$ for $i=0,\ldots,m$ and~$s$
  are monotone, $s(g_i^*(x_1,\ldots,x_n)) \le
  s(g_i^*(y_1,\ldots,y_n))$ for $i=0,\ldots,m$. Hence
  \[
  \begin{array}{rcl}
    \kappa(x_1,\ldots,x_n) &=& \min_{i=0,\ldots,m}
    \eta_{g_i}(s(x_1),\ldots,s(x_n)) \\ &=& \min_{i=0,\ldots,m}
    s(g_i^*(x_1,\ldots,x_n)) \\ &\le& \min_{i=0,\ldots,m}
    s(g_i^*(y_1,\ldots,y_n)) \\ &=& \min_{i=0,\ldots,m}
    \eta_{g_i}(s(y_1),\ldots,s(y_n))\\ &=& \kappa(y_1,\ldots,y_n).
  \end{array}
  \]
  Therefore
  \[
  \tcut(\kappa(x_1,\ldots,x_n), g_k^*(x_1,\ldots,x_n)) \le
  \tcut(\kappa(y_1,\ldots,y_n), g_k^*(y_1,\ldots,y_n))
  \]
  because~$g_k^*$ and~$\tcut$ are monotone.

  We now check that the function~$\eta_f$ defined in the statement of
  the theorem is a production function for~$f^*$. Let $x_i \in \Ab_i$
  for $i=1,\ldots,n$. If $h_i^*(x_1,\ldots,x_n) = \bot$ for some
  $i=1,\ldots,m$, then $\eta_{h_i}(s(x_1),\ldots,s(x_n)) = 0$ and
  $f^*(x_1,\ldots,x_n) = \bot$. Hence
  \[
  \eta_f(s(x_1),\ldots,s(x_n)) = 0 = s(f^*(x_1,\ldots,x_n)).
  \]
  Thus assume, e.g., $h_i^*(x_1,\ldots,x_n) = 0$ for all
  $i=1,\ldots,m$. Then $\eta_{h_i}(s(x_1),\ldots,s(x_n)) = 1$ for
  every $i=1,\ldots,m$, and $f^*(x_1,\ldots,x_n) =
  \tcut(\kappa(x_1,\ldots,x_n),g_0^*(x_1,\ldots,x_n))$. Therefore
  \[
  \begin{array}{rcl}
    \eta_f(s(x_1),\ldots,s(x_n)) &=&
    \min_{i=0,\ldots,m}\eta_{g_i}(s(x_1),\ldots,s(x_n)) \\
    &=& \kappa(x_1,\ldots,x_n) \\
    &=& \min\{\kappa(x_1,\ldots,x_n), \eta_{g_0}(s(x_1),\ldots,s(x_n))\} \\
    &=& \min\{\kappa(x_1,\ldots,x_n), s(g_0^*(x_1,\ldots,x_n))\} \\
    &=& s(\tcut(\kappa(x_1,\ldots,x_n), g_0^*(x_1,\ldots,x_n))) \\
    &=& s(f^*(x_1,\ldots,x_n)).
  \end{array}
  \]

  It remains to show that if all~$\eta_{g_i}$ and~$\eta_{h_i}$ are
  continuous, then so is~$\eta_f$. Let $D \subseteq
  \Ab_1\times\ldots\times\Ab_n$ be a directed set. First assume
  $\eta_{h_i}(\join D) = 0$ for some $1\le i \le m$. Then
  $\eta_f(\join D) = 0$ by the definition of~$\eta_f$. Also $\join
  \eta_{h_i}(D) = \eta_{h_i}(\join D) = 0$ by continuity
  of~$\eta_{h_i}$. Hence $\eta_{h_i}(d) = 0$ for every $d \in D$. So
  $\eta_f(d) = 0$ for $d \in D$. Hence $\join \eta_f(D) = 0 =
  \eta_f(\join D)$.

  So assume $\eta_{h_i}(\join D) = 1$ for every $1 \le i \le m$. Then
  \[
  \eta_f(\join D) = \min_{i=0,\ldots,m}\eta_{g_i}(\join D) =
  \min_{i=0,\ldots,m}\join\eta_{g_i}(D)
  \]
  by the continuity of~$\eta_{g_i}$. Let~$D^*$ be the set of all $d
  \in D$ such that $\eta_{h_i}(d) = 1$ for all $i=1,\ldots,m$. We have
  $\join \eta_{h_i}(D) = \eta_{h_i}(\join D) = 1$ for every
  $i=1,\ldots,m$. So for every $i=1,\ldots,m$ there is $d_i \in D$
  such that $\eta_{h_i}(d_i) = 1$. Because~$D$ is directed and
  each~$\eta_{h_i}$ is monotone, we thus have $D^* \ne \emptyset$ (an
  element greater or equal all of $d_1,\ldots,d_m$ is in~$D^*$). Hence
  for every $d \in D$ there is $d^* \in D^*$ such that $d^* \ge d$
  (take an element greater or equal than~$d$ and than some element
  of~$D^*$). Thus for every $d \in D$ there is $d^* \in D^*$ such that
  $\min_{i=0,\ldots,m}\eta_{g_i}(d) \le
  \min_{i=0,\ldots,m}\eta_{g_i}(d^*)$. Therefore
  \[
  \join \eta_f(D) = \join_{d^* \in \D^*}
  \min_{i=0,\ldots,m}\eta_{g_i}(d^*) = \join_{d \in D}
  \min_{i=0,\ldots,m}\eta_{g_i}(d)
  \]
  Hence it suffices to show
  \[
  \min_{i=0,\ldots,m}\join\eta_{g_i}(D) = \join_{d \in D}
  \min_{i=0,\ldots,m}\eta_{g_i}(d).
  \]
  Let $L = \min_{i=0,\ldots,m}\join\eta_{g_i}(D)$ and $R = \join_{d\in
    D} \min_{i=0,\ldots,m}\eta_{g_i}(d)$. Without loss of generality,
  assume $L = \join\eta_{g_0}(D)$. We need to show $L \le R$ and $R
  \le L$. For $R \le L$ it suffices to show that $L \ge
  \min_{i=0,\ldots,m}\eta_{g_i}(d)$ for $d \in D$. But $L \ge
  \eta_{g_0}(d) \ge \min_{i=0,\ldots,m}\eta_{g_i}(d)$ for $d \in
  D$. For $L \le R$ it suffices to show $R \ge \eta_{g_0}(d)$ for $d
  \in D$. So let $d \in D$ and assume $R < \eta_{g_0}(d)$. We have
  $\join \eta_{g_i}(D) \ge \eta_{g_0}(d)$ for $i=0,\ldots,m$, so for
  every $i=0,\ldots,m$ there exists $d_i \in D$ with $\eta_{g_i}(d_i)
  > R$. Because~$D$ is directed there is a $d' \in D$ such that $d'
  \ge d_i$ for $i=0,\ldots,m$. Then $\eta_{g_i}(d') > R$ for
  $i=0,\ldots,m$, because each~$\eta_{g_i}$ is monotone. Hence
  $\min_{i=0,\ldots,m}\eta_{g_i}(d') > R$. This contradicts the
  definition of~$R$.
\end{proof}

The following theorem shows how to calculate a production function for
a function defined by corecursion.

\begin{theorem}\label{thm_corecursion}
  Let $h : A_1\times\ldots\times A_n\times B^m\to B$ and
  $g_i:A_1\times\ldots\times A_n\to A_1\times\ldots\times A_n$
  (\mbox{$i=1,\ldots,m$}) where
  \[
  g_i(x_1,\ldots,x_n) = \la
  g_i^1(x_1,\ldots,x_n),\ldots,g_i^n(x_1,\ldots,x_n)\ra
  \]
  for $x_j \in A_j$ ($i=1,\ldots,m$, $j=1,\ldots,n$). Let $\eta_h$ be
  a production function for~$h$, and $\eta_{i,j}$ a production
  function for~$g_i^j$ ($i=1,\ldots,m$, $j=1,\ldots,n$). Assume
  that~$\eta_h$ satisfies:
  \[
  (\star)\quad\quad
  \eta_h(\zeta_{\Ab_1},\ldots,\zeta_{\Ab_n},\beta_1,\ldots,\beta_m)
  > \min_{i=1,\ldots,m}\beta_i
  \]
  for all $\beta_1,\ldots,\beta_m\le \zeta_{\Bb}$ with $\beta_i <
  \zeta_{\Bb}$ for some $1 \le i \le m$.

  If $f : A \to B$ is a function defined by corecursion from~$h$
  and~$g_i$ ($i=1,\ldots,m$), then there exists a production
  function~$\eta_f$ for~$f$ satisfying for all $\bar{\alpha}\in
  \On(\zeta_{\Ab_1})\times\ldots\times\On(\zeta_{\Ab_n})$ the equation
  \[
  \eta_f(\bar{\alpha}) =
  \eta_h\left(\bar{\alpha},\, \eta_f(\eta_{1,1}(\bar{\alpha}),\ldots,\eta_{1,n}(\bar{\alpha})),
    \ldots,
    \eta_f(\eta_{m,1}(\bar{\alpha}),\ldots,\eta_{m,n}(\bar{\alpha}))
  \right)
  \]
  Moreover, if~$\eta_h$ and all~$\eta_{i,j}$ ($i=1,\ldots,m$,
  $j=1,\ldots,n$) are continuous, then so is~$\eta_f$.
\end{theorem}

\begin{proof}
  Let $h^*$ be the regular extension associated with~$\eta_h$,
  and~$g_{i,j}^*$ the regular extension associated with~$\eta_{i,j}$
  for $i=1,\ldots,m$ and $j=1,\ldots,n$. Let $\Ab =
  \Ab_1\times\ldots\times\Ab_n$. For $i=1,\ldots,m$, let~$g_{i}^* :
  \Ab \to \Ab$ be defined by
  \[
  g_i^*(\bar{x}) = \la g_{i,1}^*(\bar{x}), \ldots, g_{i,n}^*(\bar{x}) \ra
  \]
  for $\bar{x} \in \Ab$. Let $F^* : \Bb^{\Ab}
  \to \Bb^{\Ab}$ be defined by
  \[
  F^*(f')(\bar{x})=h^*(\bar{x}, f'(g_1^*(\bar{x})),\ldots,f'(g_m^*(\bar{x})))
  \]
  for $\bar{x} \in \Ab$. Then~$f$ is the restriction of the least
  fixpoint~$f^*$ of~$F^*$, by~$(\star)$ and
  Lemma~\ref{lem_corecursion_regular}. Let~$(f^\alpha)_\alpha$ be the
  initial sequence of~$F^*$. Let $W =
  \On(\zeta_{\Ab_1})\times\ldots\times\On(\zeta_{\Ab_n})$. Let
  $\eta_i : W \to W$ be defined by
  \[
  \eta_i(w) = \la \eta_{i,1}(w), \ldots, \eta_{i,n}(w) \ra
  \]
  for $w \in W$. If $\bar{x} = \la x_1,\ldots,x_n\ra \in \Ab$ then we
  write~$s(\bar{x})$ for $\la s(x_1),\ldots,s(x_n)\ra$. Note that
  \[
  \eta_i(s(\bar{x})) = s(g_i^*(\bar{x}))
  \]
  for $\bar{x} \in \Ab$.

  By transfinite induction on~$\alpha$ we show that there exists a
  production function~$\eta^\alpha$ for~$f^\alpha$. For $\alpha=0$ we
  may define~$\eta^0$ by $\eta^0(w) = 0$ for $w \in W$, because
  $f^0(\bar{x})=\bot$ for any $\bar{x}\in\Ab$.

  For $\alpha=\beta+1$ we define
  \[
  \eta^{\beta+1}(w) = \eta_h(w,\eta^\beta(\eta_1(w)),\ldots,
  \eta^\beta(\eta_m(w)))
  \]
  Then the equality
  \[
  f^{\beta+1}(\bar{x}) =
  h^*(\bar{x},f^\beta(g_1^*(\bar{x})),\ldots,f^\beta(g_m^*(\bar{x})))
  \]
  and the inductive hypothesis imply that $\eta^{\beta+1}(s(\bar{x}))
  = s(f^{\beta+1}(\bar{x}))$.

  Finally, let~$\alpha$ be a limit ordinal. For $\bar{x}\in\Ab$ we
  have
  \[
  f^\alpha(\bar{x}) = \join_{\beta<\alpha}
  h^*(\bar{x},f^\beta(g_1^*(\bar{x})),\ldots,f^\beta(g_m^*(\bar{x})))
  \]
  Because~$s_\Bb$ is continuous and $f^\beta \le f^{\beta+1}$, we
  obtain
  \[
  \begin{array}{rcl}
    s(f^\alpha(\bar{x})) &=& s(\join_{\beta<\alpha} f^\beta(\bar{x})) \\
    &=& s(\join_{\beta<\alpha} f^{\beta+1}(\bar{x})) \\
    &=& s\left(\join_{\beta<\alpha}
      h^*(\bar{x},f^\beta(g_1^*(\bar{x})),\ldots,f^\beta(g_m^*(\bar{x})))
    \right) \\
    &=& \join_{\beta < \alpha} s\left(
      h^*(\bar{x},f^\beta(g_1^*(\bar{x})),\ldots,f^\beta(g_m^*(\bar{x})))
    \right) \\
    &=& \join_{\beta<\alpha} \eta_h(s(\bar{x}),
    \eta^\beta(\eta_1(s(\bar{x}))), \ldots,
    \eta^\beta(\eta_m(s(\bar{x}))))
  \end{array}
  \]
  where in the last equality we use the inductive
  hypothesis. Therefore, we may define
  \[
  \eta^\alpha(w) = \join_{\beta<\alpha} \eta_h(w,
  \eta^\beta(\eta_1(w)), \ldots, \eta^\beta(\eta_m(w)))
  \]
  The join always exists, because~$s_{\Ab_i}$ is surjective for
  $i=1,\ldots,m$. Indeed, for any $w\in W$ there is $\bar{x} \in \Ab$
  such that $w=s(\bar{x})$, by surjectivity. So
  \[
  \begin{array}{rcl}
    \eta^\alpha(w) &=& \join_{\beta<\alpha} \eta_h(w,
    \eta^\beta(\eta_1(w)), \ldots,
    \eta^\beta(\eta_m(w))) \\
    &=& \join_{\beta<\alpha} \eta_h(s(\bar{x}),
    \eta^\beta(\eta_1(s(\bar{x}))), \ldots,
    \eta^\beta(\eta_m(s(\bar{x})))) \\
    &=& \join_{\beta<\alpha} s(f^\beta(\bar{x}))
  \end{array}
  \]
  which exists because $\{ f^\beta(\bar{x}) \mid \beta < \alpha \}$ is
  a chain and~$s_\Bb$ is continuous.

  Let~$\kappa$ be the closure ordinal of~$(f^\alpha)_\alpha$. Then
  $f^{\kappa+1}=f^\kappa$, so for $\bar{x}\in\Ab$:
  \[
  \eta^{\kappa+1}(s(\bar{x})) = s(f^{\kappa+1}(\bar{x})) =
  s(f^\kappa(\bar{x})) = \eta^\kappa(s(\bar{x}))
  \]
  Since~$s_{\Ab_i}$ is surjective for $i=1,\ldots,m$, we thus have
  $\eta^{\kappa+1} = \eta^\kappa$. So for $w\in W$:
  \[
  \eta^{\kappa}(w) = \eta^{\kappa+1}(w) \\
  = \eta_h(w,\eta^\kappa(\eta_1(w)), \ldots, \eta^\kappa(\eta_m(w)))
  \]
  Therefore, $\eta^\kappa$ is the required production function for
  $f^*=f^\kappa$.

  If $\eta_h$ and all $\eta_{i,j}$ are continuous, then it follows by
  transfinite induction on~$\alpha$ that each~$\eta^\alpha$ is
  continuous.
\end{proof}

\subsection{Coterms}\label{sec_corecursion_examples}

The above general theory for defining corecursive functions will now
be illustrated with some concrete examples. The examples will involve
many-sorted coterms.

\begin{definition}\label{def_coterms}
  A \emph{many-sorted algebraic signature}
  $\Sigma=\la\Sigma_s,\Sigma_f\ra$ consists of a collection of
  \emph{sort symbols}~$\Sigma_s=\{s_i\}_{i\in I}$ and a collection of
  \emph{function symbols} $\Sigma_f=\{f_j\}_{j\in J}$. Each function
  symbol~$f$ has an associated \emph{type}
  $\tau(f)=(s_1,\ldots,s_n;s)$ where $s_1,\ldots,s_n,s\in\Sigma_s$. If
  $\tau(f)=(;s)$ then~$f$ is a \emph{constant} of sort~$s$. In what
  follows we use $\Sigma,\Sigma'$, etc., for many-sorted algebraic
  signatures, $s,s'$, etc., for sort symbols, and $f,g,c$, etc., for
  function symbols.

  The set~$\Tc^\infty(\Sigma)$, or just~$\Tc(\Sigma)$, of
  \emph{coterms over~$\Sigma$} is the set of all finite and infinite
  terms over~$\Sigma$, i.e., all finite and infinite labelled trees
  with labels of nodes specified by the function symbols of~$\Sigma$
  such that the types of labels of nodes agree. More precisely, a
  coterm over~$\Sigma$ is a function $t : \Nbb^* \to \Sigma_f \cup
  \{\bot\}$, where $\bot\notin\Sigma_f$, satisfying:
  \begin{itemize}
  \item $t(\epsilon) \ne \bot$, and
  \item if $t(p) = f \in \Sigma_f$ with $\tau(f)=(s_1,\ldots,s_n;s)$
    then
    \begin{itemize}
    \item $t(pi)=g \in \Sigma_f$ with
      $\tau(g)=(s_1',\ldots,s_{m_i}';s_i)$ for $i < n$,
    \item $t(pi)=\bot$ for $i \ge n$,
    \end{itemize}
  \item if $t(p)=\bot$ then~$t(pi)=\bot$ for every $i\in\Nbb$,
  \end{itemize}
  where $\epsilon \in \Nbb^*$ is the empty string. We use obvious
  notations for coterms, e.g., $f(g(t,s),c)$ when $c,f,g \in \Sigma_f$
  and $t,s \in \Tc(\Sigma)$, and the types agree. We say that a
  coterm~$t$ \emph{is of sort~$s$} if $t(\epsilon)$ is a function
  symbol of type $(s_1,\ldots,s_n;s)$ for some
  $s_1,\ldots,s_n\in\Sigma_s$. By~$\Tc_s(\Sigma)$ we denote the set of
  all coterms of sort~$s$ from~$\Tc(\Sigma)$. We also
  write~$\Tc_*(\Sigma)$ for~$\Tc(\Sigma)$, i.e., by~$*$ we denote a
  special sort of all coterms.

  The \emph{$n$-th approximant} of a coterm $t\in\Tc(\Sigma)$ is a
  function $\ucut{t}{n} : \Nbb^* \to \Sigma_f \cup \{\bot\}$ such that
  \begin{itemize}
  \item $\ucut{t}{n}(p)=t(p)$ if $|p| < n$, or $|p| = n > 0$
    and~$t(p)$ is a constant, i.e., $\tau(t(p))=(;s)$ for some $s \in
    \Sigma_s$,
  \item $\ucut{t}{n}(p)=\bot$ otherwise,
  \end{itemize}
  where by $|p|$ we denote the length of $p\in\Nbb^*$. In other words,
  $\ucut{t}{n}$ is~$t$ cut at depth~$n$, but we do not change
  constants in leaves at depth~$n > 0$ into~$\bot$. By~$\Tc^n(\Sigma)$
  we denote the set of all $n$-th approximants of coterms
  over~$\Sigma$. We extend the notation $\ucut{t}{m}$ to approximants
  $t\in\Tc^n(\Sigma)$ in the obvious way. We also use obvious
  notations for approximants, e.g., the first approximant
  of~$f(g(t_1),h(t_2),c)$ is denoted by~$f(\bot,\bot,c)$. We say that
  an approximant $t \in \Tc^n(\Sigma)$ \emph{is of sort}~$s \in
  \Sigma_s$ if either $t(\epsilon) = \bot$ or~$t(\epsilon)$ is a
  function symbol of type $(s_1,\ldots,s_n;s)$ for some
  $s_1,\ldots,s_n\in\Sigma_s$. By~$\Tc_s^n(\Sigma)$ we denote the set
  of all $t \in \Tc^n(\Sigma)$ of sort~$s$.

  The partial order $\Nbb_\infty = \la \Nbb \cup \{\infty\}, \le \ra$
  is ordered by the usual order on~$\Nbb$ extended with $n\le\infty$
  for $n\in\Nbb_\infty$. Note that~$\Nbb_\infty$ is isomorphic
  to~$\On(\omega)$. We extend the arithmetical operations on~$\Nbb$
  to~$\Nbb_\infty$ in an obvious way, with $\infty-n=\infty$,
  $n+\infty=\infty+n=\infty+\infty=\infty$, $n\cdot\infty=\infty\cdot
  n=\infty\cdot\infty=\infty$, where $n\in\Nbb$.

  If~$A_i$ for $i\in I$ are sets, then by $\amalg_{i\in I}A_i$ we
  denote the coproduct of the~$A_i$s, i.e., the set of all pairs $\la
  i,a \ra$ such that $i\in I$ and $a \in A_i$. We define the partial
  order $\Tb(\Sigma) = \la \amalg_{n\in\Nbb_\infty} \Tc^n(\Sigma),
  \qle\ra$ by: $\la i,t\ra\qle\la j,s\ra$ iff $i \le j$ and
  $\ucut{s}{i} = t$. The \emph{size}~$|t|\in\Nbb_\infty$ of $t \in
  \Tb(\Sigma)$ is the first component of~$t$. We will often confuse
  $\la i,t\ra \in \Tb(\Sigma)$ with~$t \in \Tc^i(\Sigma)$. For a sort
  symbol $s \in \Sigma_s$, by~$\Tb_s(\Sigma)$ we denote the subset
  of~$\Tb(\Sigma)$ consisting of all $\la i,t\ra$ such that~$t$ is of
  sort~$s$. We also use the notation~$\Tb_*(\Sigma)$
  for~$\Tb(\Sigma)$.

  We define $\tcut : \Nbb_\infty \times \Tb(\Sigma) \to \Tb(\Sigma)$
  by $\tcut(n, \la i, t \ra) = \la i, t\ra$ if $i \le n$, and
  $\tcut(n, \la i, t \ra) = \la n, \ucut{t}{n}\ra$ if $i > n$. Note
  that if $t \in \Tb_s(\Sigma)$ then $\tcut(n,t) \in \Tb_s(\Sigma)$.

  Let $f \in \Sigma_f$ be a function symbol of type
  $(s_1,\ldots,s_n;s)$ and let $1\le i\le n$. The \emph{$i$-th
    destructor for~$f$} is a function
  $d_{f,i}:\Tc_s(\Sigma)\to\Tc_{s_i}(\Sigma)$ defined by:
  \[
  d_{f,i}(t) =
  \left\{
  \begin{array}{rl}
    t_i &\text{ if } t \equiv f(t_1,\ldots,t_n) \\
    t' &\text{ otherwise }
  \end{array}
  \right.
  \]
  where $t' \in \Tc_{s_i}(\Sigma)$ is arbitrary. The \emph{constructor
    for~$f$} is a function
  $c_f:\Tc_{s_1}(\Sigma)\times\ldots\times\Tc_{s_n}(\Sigma)\to\Tc_s(\Sigma)$
  defined by
  \[
  c_f(t_1,\ldots,t_n) = f(t_1,\ldots,t_n)
  \]
  The \emph{test for~$f$} is a function $o_f:\Tc_s\to\{0,1\}$ defined
  by
  \[
  o_f(t) =
  \left\{
  \begin{array}{rl}
    1 &\text{ if } t \equiv f(t_1,\ldots,t_n) \\
    0 &\text{ otherwise }
  \end{array}
  \right.
  \]

  If $t \in \Tb(\Sigma)$ and $t \equiv \la i, t'\ra$ with $t' \in
  \Tc^\infty(\Sigma)$ (this may happen for $i < \infty$ if e.g.~$t'$
  is a constant), then by~$\ulift{t}$ we denote $\la \infty, t'\ra$.
\end{definition}

\begin{lemma}\label{lem_coterms_cpo}
  The partial order~$\Tb(\Sigma)$ is a CPO. Also, for each $s \in
  \Sigma_s$, the partial order~$\Tb_s(\Sigma)$ is a CPO.
\end{lemma}

\begin{proof}
  The bottom of~$\Tb(\Sigma)$ is $\la 0,\bot\ra$, where~$\bot$ is the
  sole element of~$\Tc^0(\Sigma)$.

  Let $D \subseteq \Tb(\Sigma)$ be a directed set. Let~$n$ be the
  supremum of the first coordinates of elements of~$D$. Define
  $t\in\Tc^n(\Sigma)$ by:
  \begin{itemize}
  \item $t(p) = f$ if $f \in \Sigma_f$ and there is $s \in D$ with
    $s(p) = f$,
  \item $t(p) = \bot$ otherwise.
  \end{itemize}
  where $p \in \Nbb^*$. This is a good definition, because~$D$ is
  directed. By definition we of course have $\la n,t\ra \qge s$ for
  all $s \in D$. If~$n < \infty$ then $\la n, t\ra \in D$, so it is
  the supremum of~$D$. Assume $n=\infty$. Suppose $u \in
  \Tc^m(\Sigma)$ and $\la m, u\ra \qge s$ for all $s \in D$. Then $m =
  \infty$. Let $p \in \Nbb^*$. Then there exists $s \in D$ with $|s|
  \ge |p|$, and so $t(p)=s(p)$. Since $u \qge s$, we obtain
  $u(p)=s(p)=t(p)$. Thus $u = t$, so $\la \infty, t\ra$ is the
  supremum of~$D$.

  That for each $s \in \Sigma_s$, the order~$\Tb_s(\Sigma)$ is a CPO
  follows from the fact that if all elements of a directed set are of
  sort~$s$, then its supremum is also of sort~$s$.
\end{proof}

\begin{lemma}\label{lem_coterms_sized_cpo}
  The tuple $\la \Tb(\Sigma), \infty, |\cdot|, \tcut \ra$ is a sized
  CPO.\footnote{Since $\Nbb_\infty$ and $\On(\omega)$ are isomorphic
    we identify them without loss of generality. So $\infty =
    \omega$.} Also for each $s \in \Sigma_s$, the tuple $\la
  \Tb_s(\Sigma), \infty, |\cdot|, \tcut \ra$ is a sized CPO.
\end{lemma}

\begin{proof}
  The only part which is not completely obvious is the continuity of
  the size function~$|\cdot|$. Let $D \subseteq \Tb(\Sigma)$ be a directed
  set. The set~$|D|$ is directed. By the definition of~$\join D$ in
  the proof of Lemma~\ref{lem_coterms_cpo}, we have $|\join D| = \join
  |D|$.
\end{proof}

Our definition of~$\Tb(\Sigma)$ may seem somewhat convoluted. One may
wonder why we do not simply use
$\bigcup_{n\in\Nbb_\infty}\Tc^n(\Sigma)$, or even the set of all
coterms with some arbitrary subterms changed into~$\bot$, with an
obvious ``information'' ordering. The answer is that then there would
be no cut function~$\tcut$ with the desired properties. Also, the
construction of~$\Tb(\Sigma)$ is a slightly modified instance of a
more general sized CPO construction for an arbitrary final coalgebra
in the category of sets (see the appendix).

For the sake of brevity we often use $\Tc = \Tc(\Sigma)$, $\Tc_s =
\Tc_s(\Sigma)$, $\Tb = \Tb(\Sigma)$ and $\Tb_s = \Tb_s(\Sigma)$, i.e.,
we omit the signature~$\Sigma$ when clear from the context or
irrelevant. We also confuse~$\Tb$ and~$\Tb_s$ with the sized CPOs from
Lemma~\ref{lem_coterms_sized_cpo}.

\begin{lemma}\label{lem_coterms_prodfuns}
  Assume $f \in \Sigma_f$ has type $(s_1,\ldots,s_n;s)$.
  \begin{enumerate}
  \item The function $\eta_{d_{f,i}}:\Nbb_\infty\to\Nbb_\infty$
    defined by $\eta_{d_{f,i}}(n)=\max(0,n-1)$ is a production
    function for~$d_{f,i}$.
  \item If all elements of~$\Tc_{s_i}$ are constants, then
    $\eta_\infty : \Nbb_\infty \to \Nbb_\infty$ defined by
    \[
    \eta_\infty(n) =
    \left\{
      \begin{array}{cl}
        \infty & \text{ if } n > 0 \\
        0 & \text{ if } n = 0
      \end{array}
    \right.
    \]
    is a production function for~$d_{f,i}$.
  \item The function $\eta_{c_f}:\Nbb_\infty^n\to\Nbb_\infty$ defined
    by $\eta_{c_f}(m_1,\ldots,m_n)=\min_{i=1,\ldots,n}m_i+1$ is a
    production function for~$c_f$.
  \item For any $k \in \Nbb$ the function $\eta_{o_f}^k : \Nbb_\infty
    \to \{0,1\}$ defined by
    \[
    \eta_{o_f}^k(n) =
    \left\{
      \begin{array}{cl}
        1 & \text{ if } n > k \\
        0 & \text{ if } n \le k
      \end{array}
    \right.
    \]
    is a production function for~$o_f$.
  \item If all elements of~$\Tc_s$ are constants and $g : \Tc_s^m \to
    \Tc_s$ then the function $\eta_g^\infty : \Nbb_\infty^m \to
    \Nbb_\infty$ given by
    \[
    \eta_g^\infty(n_1,\ldots,n_m) = \left\{
      \begin{array}{cl}
        \infty & \text{ if } n_1,\ldots,n_m > 0 \\
        0 & \text{ otherwise }
      \end{array}
    \right.
    \]
    is a production function for~$g$.
  \item If all elements of~$\Tc_s$ are constants and $\chi : \Tc_s^m
    \to \{0,1\}$ then the function $\eta_{\chi}^{\infty} :
    \Nbb_\infty^m \to \{0,1\}$ defined by
    \[
    \eta_\chi^{\infty}(n_1,\ldots,n_m) = \left\{
      \begin{array}{cl}
        1 & \text{ if } n_1,\ldots,n_m > 0 \\
        0 & \text{ otherwise }
      \end{array}
    \right.
    \]
    is a production function for~$\chi$.
  \end{enumerate}
\end{lemma}

\begin{proof}~
  \begin{enumerate}
  \item The $i$-th destructor $d_{f,i}:\Tc_s(\Sigma)\to\Tc_s(\Sigma)$
    extends to a regular
    $d_{f,i}^*:\Tb_s(\Sigma)\to\Tb_{s_i}(\Sigma)$:
    \[
    d_{f,i}^*(t) =
    \left\{
      \begin{array}{rl}
        t_i &\text{ if } t \equiv f(t_1,\ldots,t_n) \\
        t' &\text{ otherwise, where } t' \in \Tb_{s_i}(\Sigma) \text{ is
          arbitrary with } |t'| = \max(0, |t| - 1)
      \end{array}
    \right.
    \]
  \item We may take the regular extension
    \[
    d_{f,i}^\infty(t) =
    \left\{
      \begin{array}{cl}
        \ulift{t_i} & \text{ if } t \equiv f(t_1,\ldots,t_n) \\
        t' & \text{ if } t \equiv g(t_1,\ldots,t_k)
        \text{ with } g \ne f\text{, where } t' \in \Tb_{s_i}(\Sigma)
        \\ & \text{ is arbitrary with } |t'| = \infty \\
        \bot & \text{ otherwise }
      \end{array}
    \right.
    \]
    This is well-defined, because~$t_i$ above is a constant.
  \item The constructor~$c_f$ extends to a regular
    $c_f^*:\Tb_{s_1}(\Sigma)\times\ldots\times\Tb_{s_n}(\Sigma)\to\Tb(\Sigma)$
    as follows:
    \[
    c_f^*(t_1,\ldots,t_n) = f(\ucut{t_1}{m},\ldots,\ucut{t_n}{m})
    \]
    where $m = \min_{i=1,\ldots,n}|t_i|$.
  \item The test~$o_f$ extends to a regular
    $o_f^*:\Tb_s(\Sigma)\to\Two$ as follows:
    \[
    o_f^*(t) =
    \left\{
      \begin{array}{cl}
        1 & \text{ if } |t| > k \text{ and } t \equiv f(t_1,\ldots,t_n) \\
        0 & \text{ if } |t| > k \text{ and } t \equiv
        g(t_1,\ldots,t_k) \text{ for } g \in \Sigma_f, g \ne f \\
        \bot & \text{ otherwise, if } |t| \le k
      \end{array}
    \right.
    \]
  \item A regular extension $g^\infty : \Tb_s^m \to \Tb_s$ of~$g$ is
    given by
    \[
    g^\infty(t_1,\ldots,t_m) =
    \left\{
      \begin{array}{cl}
        g(\ulift{t_1},\ldots,\ulift{t_m}) & \text{ if } t_i \ne \bot \text{ for }
        i=1,\ldots,m \\
        \bot & \text{ otherwise }
      \end{array}
    \right.
    \]
  \item A regular extension $\chi^\infty : \Tb_s^m \to \Two$ of~$\chi$
    is given by
    \[
    \chi^\infty(t_1,\ldots,t_m) =
    \left\{
      \begin{array}{cl}
        \chi(\ulift{t_1},\ldots,\ulift{t_m}) &
        \text{ if } t_i \ne \bot \text{ for } i=1,\ldots,m \\
        \bot & \text{ otherwise }
      \end{array}
    \right.
    \]
  \end{enumerate}
\end{proof}

The following simple lemma implies that all the production functions
from the above lemma are continuous.

\begin{lemma}\label{lem_infinite_continuous}
  A production function $\eta : \Nbb_\infty^m \to \Nbb_\infty$ is
  continuous iff $\eta(\join D) = \join \eta(D)$ for any infinite
  directed $D \subseteq \Nbb_\infty^m$.
\end{lemma}

\begin{proof}
  The implication from left to right is obvious. For the converse, let
  $D \subseteq \Nbb_\infty^m$ be a finite directed set. Then~$\join D$
  is the largest element of~$D$. So $\eta(\join D)$ is the largest
  element of~$\eta(D)$, because~$\eta$ is monotone by
  Lemma~\ref{lem_prodfun_regular}. Thus $\join \eta(D) = \eta(\join
  D)$.
\end{proof}

Because any $p \in \Nbb_\infty^m \setminus \Nbb^m$ is a join of an
infinite chain $C \subseteq \Nbb^m$, the above lemma implies that the
values of continuous functions in $\Nbb_\infty^m \to \Nbb_\infty$ are
uniquely determined by their values on~$\Nbb^m$. We shall thus often
treat continuous functions as if they were defined on~$\Nbb^m$, and
leave their values at infinity implicit.

\begin{lemma}\label{lem_inequality_at_infinity}
  If $\eta : \Nbb_\infty^m \to \Nbb_\infty$ is continuous and for
  every $n_1,\ldots,n_m \in \Nbb$ we have
  \[
  \eta(n_1,\ldots,n_m) > \min_{i=1,\ldots,m}n_i,
  \]
  then also for every $n_1,\ldots,n_m \in \Nbb_\infty$ such that $n_k
  < \infty$ for some $1 \le k \le m$, we have $\eta(n_1,\ldots,n_m) >
  \min_{i=1,\ldots,m}n_i$.
\end{lemma}

\begin{proof}
  Without loss of generality, we consider the case $m=2$ and show that
  if $\eta(n_1,n_2) > \min(n_1,n_2)$ for $n_1,n_2\in\Nbb$ then
  $\eta(n,\infty) > n$ for $n \in \Nbb$. For every $k \ge n$ we have
  $\eta(n,k) > \min(n,k) = n$, i.e., $\eta(n,k) \ge n + 1$. Because $D
  = \{ \la n, k \ra \mid k \ge n \}$ is directed, $\join D = \la n,
  \infty \ra$ and~$\eta$ is continuous, we have
  \[
  \eta(n,\infty) = \eta(\join D) = \join \eta(D) = \join_{k\ge
    n}\eta(n,k) \ge \join_{k\ge n}n + 1 = n + 1 > n.
  \]
\end{proof}

The method for showing well-definedness of functions given by
corecursive equations is to use Lemma~\ref{lem_coterms_prodfuns} and
lemmas~\ref{lem_production_functions}-\ref{lem_cases} and
Theorem~\ref{thm_corecursion} from Section~\ref{sec_corecursion} to
calculate production functions, and then apply
Corollary~\ref{cor_unique_solution}. For convenience of reference, we
reformulate Corollary~\ref{cor_unique_solution} specialized to
many-sorted coterms, in its most useful form.

\begin{corollary}\label{cor_unique_solution_coterms}
  Let~$S$ be an arbitrary set. Let $h : S \times \Tc_s^{m}\to \Tc_s$
  and $g_i:S\to S$ ($i=1,\ldots,m$). For each $x \in S$, let
  $\eta_h^x:\Nbb_\infty^{m}\to \Nbb_\infty$ be a continuous production
  function for $\lambda \bar{y} . h(x, \bar{y})$. If
  \[
  (\star)\quad\quad \eta_h^x(n_1,\ldots,n_m) > \min_{i=1,\ldots,m}n_i
  \]
  for all $x \in S$ and all $n_1,\ldots,n_m \in \Nbb$, then there
  exists a function $f : S \to \Tc_s$ defined by corecursion from~$h$
  and $g_1,\ldots,g_m$.
\end{corollary}

\begin{proof}
  Follows from Corollary~\ref{cor_unique_solution} and
  Lemma~\ref{lem_inequality_at_infinity}.
\end{proof}

Note that if $S = \Tc_{s_1}\times\ldots\times\Tc_{s_k}$, the local
prefix production functions~$\eta_h^x$ above are all the same and a
global prefix production function $\eta_h : \Nbb_\infty^{k+m} \to
\Nbb_\infty$ satisfies $\eta_h(\infty,\ldots,\infty,n_1,\ldots,n_m) =
\eta_h^x(n_1,\ldots,n_m)$, then~$(\star)$ in
Corollary~\ref{cor_unique_solution_coterms} implies~$(\star)$ in
Theorem~\ref{thm_corecursion}. This situation is usually the case, and
we will often avoid mentioning it explicitly.

\begin{example}
  Let $A$ be a set. Let $\Sigma$ consist of two sorts~$\sfr$
  and~$\dfr$, one function symbol~$\cons$ of type $(\dfr,\sfr;\sfr)$
  and a distinct constant symbol $a \in A$ of sort~$\dfr$ for each
  element of~$A$. Then~$\Tc_\sfr(\Sigma)$ is the set of streams
  over~$A$. We also write $\Tc_\sfr(\Sigma) = A^\omega$ and
  $\Tc_\dfr(\Sigma) = A$. Instead of $\cons(a,t)$ we usually write $a
  : t$, and we assume that~$:$ associates to the right, e.g., $x : y :
  t$ is $x : (y : t)$. We also use the notation $x : t$ to denote the
  application of the constructor for~$\cons$ to~$x$ and~$t$. Instead
  of~$d_{\cons,1}$ we write~$\head$, and instead of~$d_{\cons,2}$ we
  write~$\tail$. Instead of $o_a(x) = 1$, where $a \in A$, we write
  $x=a$. For~$\tail$ we shall use the continuous production function
  $\eta_\tail(n) = \max(0, n - 1)$, and for~$\cons$ we shall use the
  function $\eta_\cons(n) = n + 1$. Since all elements of $\Tc_\dfr$
  are constants, we may use
  \[
  \eta_\infty(n) = \left\{
    \begin{array}{cl}
      \infty & \text{ if } n > 0 \\
      0 & \text{ if } n = 0
    \end{array}
  \right.
  \]
  as a continuous production function for~$\head$. For~$o_\cons$ we
  use the continuous production function~$\eta_{o_\cons}^0$. See
  Lemma~\ref{lem_coterms_prodfuns}.

  Consider the equation
  \[
  \even(x : y : t) = x : \even(t)
  \]
  We shall show that the above equation determines a unique function
  $\even : A^\omega \to A^\omega$.

  The equation may be rewritten as
  \[
  \even(t) = \head(t) : \even(\tail(\tail(t))) \text{ if } o_\cons(t) = 1
  \]
  So $\even$ is defined by corecursion from $h : A^\omega \times A^\omega \to
  A^\omega$ given by
  \[
  h(t,t') = \head(t) : t' \text{ if } o_\cons(t) = 1
  \]
  and $g : A^\omega \to A^\omega$ given by
  \[
  g(t) = \tail(\tail(t))
  \]
  The function~$h$ is defined by cases from $h_0 : A^\omega \times
  A^\omega \to A^\omega$ given by
  \[
  h_0(t,t') = \head(t) : t'
  \]
  and from the test function~$o_\cons$ (formally, we also need some
  function for~$g_0$ in Definition~\ref{def_corecursion}, i.e., for
  the case when none of the conditions holds, but in the present
  instance~$o_\cons$ never gives the value~$0$ so this does not
  matter). Using Lemma~\ref{lem_substitution} we conclude that for
  each $t \in \Tc$ a continuous $t$-local prefix production
  function\footnote{Recall that by Lemma~\ref{lem_infinite_continuous}
    we may consider continuous production functions as defined
    on~$\Nbb^m$ instead of~$\Nbb_\infty^m$.} $\xi^t : \Nbb \to \Nbb$
  is defined by $\xi^t(n) = n + 1$. Therefore, $\even$ is well-defined
  (i.e.~it exists and is unique) by
  Corollary~\ref{cor_unique_solution_coterms}. Using
  Lemma~\ref{lem_substitution} we see that a continuous production
  function $\eta_g : \Nbb \to \Nbb$ for~$g$ is defined by
  \[
  \eta_g(n) =
  \left\{
    \begin{array}{cl}
      n - 2 & \text{ if } n \ge 2 \\
      0 & \text{ otherwise }
    \end{array}
  \right.
  \]
  From Lemma~\ref{lem_substitution} and Lemma~\ref{lem_cases} a
  continuous production function $\eta_h : \Nbb^2 \to \Nbb$ for~$h$ is
  given by
  \[
  \eta_h(n,m) =
  \left\{
    \begin{array}{cl}
      m + 1 & \text{ if } n > 0 \\
      0 & \text{ if } n = 0
    \end{array}
  \right.
  \]
  Therefore, by Theorem~\ref{thm_corecursion} there exists a
  continuous production function $\eta_\even : \Nbb \to \Nbb$
  for~$\even$ satisfying
  \[
  \eta_\even(n) =
  \left\{
    \begin{array}{cl}
      \eta_\even(n-2) + 1 & \text{ if } n \ge 2 \\
      \eta_\even(0) + 1 & \text{ if } n = 1 \\
      0 & \text{ if } n = 0
    \end{array}
  \right.
  \]
  Thus $\eta_\even(n) = \ceil{\frac{n}{2}}$ for $n \in \Nbb$.

  Usually, we do not get into so much detail when justifying
  well-definedness of a function given by some corecursive
  equations. Formally, sets of equations of the form
  \[
  \begin{array}{c}
    f(t_1) = s_1 \\
    \vdots \\
    f(t_k) = s_k
  \end{array}
  \]
  where $t_1,\ldots,t_k$ are some patterns and $s_1,\ldots,s_k$ some
  expressions possibly involving~$f$, are always interpreted as
  defining a function~$f$ by (corecursion from a function defined by)
  cases from appropriate functions corresponding to the~$s_i$ and from
  some combinations of test functions corresponding to the
  patterns. These correspondences are usually straightforward and left
  implicit. To prove well-definedness of~$f$ we implicitly use
  lemmas~\ref{lem_production_functions}-\ref{lem_cases} to calculate
  all local prefix production functions, and then we show~$(\star)$ in
  Corollary~\ref{cor_unique_solution_coterms}. If we are also
  interested in a production function for~$f$, then we calculate
  production functions for the argument functions~$g_i$ and the prefix
  function~$h$ (using
  lemmas~\ref{lem_production_functions}-\ref{lem_cases}), and then we
  apply Theorem~\ref{thm_corecursion} to obtain recursive equations
  for a production function for~$f$. The resulting production
  functions are typically continuous. We leave this observation
  implicit and consider the production functions as functions defined
  on~$\Nbb^m$ (which we can do by Lemma~\ref{lem_infinite_continuous}
  and Lemma~\ref{lem_inequality_at_infinity}).

  Applying the remarks of the preceding paragraph, we now give
  arguments justifying the well-definedness of~$\even$ and the form of
  its production function in a style which we shall adopt from now on.

  A prefix production function\footnote{More precisely: for each $r
    \in A^\omega$ a continuous $r$-local prefix production
    function\ldots} for a function~$\even$ satisfying
  \[
  \even(x : y : t) = x : \even(t)
  \]
  is given by $\xi(n) = n + 1 > n$. Thus~$\even$ is well-defined and
  its production function $\eta_\even : \Nbb \to \Nbb$
  satisfies\footnote{We leave implicit the verification of~$(\star)$
    in Theorem~\ref{thm_corecursion}, which follows from the fact that
    a global prefix production function~$\xi'$ (see the definition
    of~$\eta_h$ above) satisfies $\xi'(\infty,n) = \xi(n)$, and
    $\xi(n) > n$ for $n \in \Nbb$.} for $n \in \Nbb$:
  \[
  \begin{array}{rcl}
    \eta_\even(0) &=& 0 \\
    \eta_\even(1) &=& 1 + \eta_\even(0) \\
    \eta_\even(n+2) &=& 1 + \eta_\even(n)
  \end{array}
  \]
  Hence $\eta_\even(n) = \ceil{\frac{n}{2}}$ for $n \in
  \Nbb$. \hfill$\Box$
\end{example}

The above definition of~$\even$ is actually an instance of a common
form of definition by guarded corecursion.

\begin{definition}\label{def_guarded_corecursion}
  A function $h : S \times \Tc_s^m \to \Tc_{s'}$ (for $m \in \Nbb$,
  $s,s' \in \Sigma_s$) is \emph{non-consuming} if for each $x \in S$
  there is a continuous production function
  $\eta_h^x:\Nbb_\infty^{m}\to \Nbb_\infty$ for $\lambda \bar{y}
  . h(x, \bar{y})$ satisfying
  \[
  \eta_h^x(n_1,\ldots,n_m) \ge \min_{i=1,\ldots,m}n_i
  \]
  for all $n_1,\ldots,n_m \in \Nbb$.

  The class of \emph{constructor-guarded} functions is defined
  inductively as the class of all functions $h : S \times \Tc_s^m \to
  \Tc_{s'}$ (for arbitrary $m \in \Nbb$, $s,s' \in \Sigma_s$) such
  that
  \[
  h(x,y_1,\ldots,y_m) = c(u_1(x,y_1,\ldots,y_m),\ldots,u_k(x,y_1,\ldots,y_m))
  \]
  where~$c$ is a constructor for a function symbol of type
  $(s_1,\ldots,s_k;s')$ and each $u_i : S \times \Tc_s^m \to
  \Tc_{s_i}$ is non-consuming.

  We say that a function $f : S \to \Tc_s$ is defined by \emph{guarded
    corecursion} from $h : S \times \Tc_s^m \to \Tc_s$ and $g_i : S
  \to S$ ($i=1,\ldots,m$) if it is defined by corecursion from~$h$ and
  $g_1,\ldots,g_m$, with~$h$ defined by cases from some
  constructor-guarded functions $h_j : S \times \Tc_s^m \to \Tc_s$
  ($j=0,\ldots,k$) and some condition functions $o_j : S \to \{0,1\}$
  ($j=1,\ldots,k$), i.e., the condition functions depend only on the
  first argument of~$h$.
\end{definition}

Note that every function $h : S \times \Tc_s^m \to \Tc_{s'}$ which
\begin{itemize}
\item depends only on its first argument, or
\item satisfies $h(x,y_1,\ldots,y_m) = y_i$ for all $x \in S$,
  $y_1,\ldots,y_m\in\Tc_s^m$, fixed~$i$, or
\item is constructor-guarded
\end{itemize}
is also non-consuming.

By Corollary~\ref{cor_unique_solution_coterms}, for every~$h$ and
$g_1,\ldots,g_m$ satisfying the requirements of
Definition~\ref{def_guarded_corecursion}, there exists a unique
function defined by guarded corecursion. When some corecursive
equations involving a function~$f$ straightforwardly translate to a
definition by guarded corecursion, then we say that the definition
of~$f$ is \emph{guarded}, which implies well-definedness
of~$f$. If~$f$ is defined by guarded corecursion, $S =
\Tc_{s_1}\times\ldots\times\Tc_{s_l}$ and there exist appropriate
production functions for the~$u_j$ ($j=1,\ldots,k$), then~$(\star)$ in
Theorem~\ref{thm_corecursion} holds, so we may then use
Theorem~\ref{thm_corecursion} to calculate a production function
for~$f$.

The functions~$\eta_f^\infty$ and~$\eta_\chi^\infty$ for various~$f$
and~$\chi$ (see Lemma~\ref{lem_coterms_prodfuns}) will be used
implicitly in calculations of production functions in the following
examples.

\begin{example}
  Consider the equations over streams of natural numbers:
  \[
  \begin{array}{rcl}
    \add(x : t, y : s) &=& (x + y) : \add(t,s) \\
    \zip(x : t, s) &=& x : \zip(s, t) \\
    D &=& 0 : 1 : 1 : \zip(\add(\tail(D),\tail(\tail(D))),\even(\tail(D)))
  \end{array}
  \]
  We show that these equations define unique functions $\add :
  \Nbb^\omega \times \Nbb^\omega \to \Nbb^\omega$, $\zip : \Nbb^\omega
  \times \Nbb^\omega \to \Nbb^\omega$, and a unique stream\footnote{To
    make the definition of~$D$ consistent with our theory, which
    considers only functions, we could provide~$D$ with one dummy
    argument.}~$D \in \Nbb^\omega$.

  The function~$\add$ is well-defined, because its definition is
  guarded. A production function $\eta_\add : \Nbb \times \Nbb \to
  \Nbb$ for~$\add$ satisfies
  \[
  \begin{array}{rcl}
    \eta_\add(0,m) &=& 0 \\
    \eta_\add(n,0) &=& 0 \\
    \eta_\add(n + 1, m + 1) &=& \eta_\add(n, m) + 1
  \end{array}
  \]
  Thus $\eta_\add(n, m) = \min(n,m)$.

  The definition of~$\zip$ is also guarded, so~$\zip$ is
  well-defined. A production function $\eta_\zip : \Nbb \times \Nbb
  \to \Nbb$ for~$\zip$ satisfies
  \[
  \begin{array}{rcl}
    \eta_\zip(0, m) &=& 0 \\
    \eta_\zip(n + 1, m) &=& \eta_\zip(m, n) + 1
  \end{array}
  \]
  The equations for~$\eta_\zip$ are equivalent to
  \[
  \begin{array}{rcl}
    \eta_\zip(0, m) &=& 0 \\
    \eta_\zip(n + 1, 0) &=& 1 \\
    \eta_\zip(n + 1, m + 1) &=& \eta_\zip(n, m) + 2
  \end{array}
  \]
  Thus $\eta_\zip(n,m) = \min(2n,2m+1)$.

  Using the formulas for~$\eta_\even$, $\eta_\add$ and~$\eta_\zip$ we
  now calculate a prefix production function~$\xi$ for~$D$. For
  $n < 2$ we have
  \[
  \begin{array}{rcl}
    \xi(n) &=& 3 + \eta_\zip(\eta_\add(0,0), \eta_\even(0)) \\
    &=& 3 + \eta_\zip(0, 0) \\
    &=& 3
  \end{array}
  \]
  so $\xi(n) > n$ for $n < 2$. For $n \ge 2$ we have
  \[
  \begin{array}{rcl}
    \xi(n) &=& 3 + \eta_\zip(\eta_\add(n-1,n-2), \eta_\even(n-1))
    \\
    &=& 3 + \eta_\zip(\min(n-1,n-2), \ceil{\frac{n-1}{2}}) \\
    &=& 3 + \eta_\zip(n-2,\ceil{\frac{n-1}{2}}) \\
    &=& 3 + \min(2(n - 2), 2 \ceil{\frac{n-1}{2}} + 1)
  \end{array}
  \]
  We have $2(n - 2) = 2n - 4 > n - 3$ for $n \ge 2$. Also $2
  \ceil{\frac{n-1}{2}} + 1 \ge n - 1 + 1 = n > n - 3$. Hence for $n
  \ge 2$ we have $\xi(n) > 3 + n - 3 = n$. Thus $\xi(n) > n$
  for $n \in \Nbb$, and therefore~$D$ is well-defined.\hfill$\Box$
\end{example}

\begin{example}
  Consider the following specification of the Hamming string~$H$ of
  positive natural numbers not divisible by primes other than $2$, $3$
  and~$5$.
  \[
  \begin{array}{rcl}
    \mul(x, y : t) &=& x \cdot y : \mul(x, t) \\
    \mer(x : t_1, y : t_2) &=&
    \left\{
    \begin{array}{cl}
      x : \mer(t_1, y : t_2) & \text{ if } x \le y \\
      y : \mer(x : t_1, t_2) & \text{ otherwise }
    \end{array}
    \right.
    \\
    H &=& 1 : \mer(\mer(\mul(2, H), \mul(3, H)), \mul(5,H))
  \end{array}
  \]
  We show that $\mul : \Nbb \times \Nbb^\omega \to \Nbb^\omega$,
  $\mer : \Nbb^\omega \times \Nbb^\omega \to \Nbb^\omega$ and $H \in
  \Nbb^\omega$ are well-defined.

  The function~$\mul$ is well-defined, because the definition is
  guarded. A production function\footnote{Formally, we consider
    infinitely many functions $\lambda t . \mul(n, t)$ for each
    $n\in\Nbb$, and apply Theorem~\ref{thm_corecursion} to each of
    them.}~$\eta_\mul$ for~$\mul$ is given by $\eta_\mul(n) = n$. The
  definition of~$\mer$ is also guarded, so~$\mer$ is well-defined. A
  production function~$\eta_\mer$ for~$\mer$ satisfies\footnote{We use
    Lemma~\ref{lem_cases} and Theorem~\ref{thm_corecursion}.}:
  \[
  \begin{array}{rcl}
    \eta_\mer(0,m) &=& 0 \\
    \eta_\mer(n,0) &=& 0 \\
    \eta_\mer(n+1,m+1) &=& \min(\eta_\mer(n,m+1),\eta_\mer(n+1,m)) + 1
  \end{array}
  \]
  Thus $\eta_\mer(n,m) = \min(n,m)$. Note that the form of this
  production function or even its existence is not completely
  intuitive -- one would expect that the ``size'' of the resulting
  stream may depend on the elements of the argument streams, not only
  on their sizes. The trick is that we use cut functions in the proof
  of Lemma~\ref{lem_cases} to effectively select the least possible
  size, disregarding any side conditions.

  Therefore, a prefix production function~$\xi$ for~$H$ satisfies
  $\xi(n) = 1 + \min(\min(n, n), n) = n + 1 > n$. So~$H$ is
  well-defined.\hfill$\Box$
\end{example}

Specifications of many-sorted signatures may be conveniently given by
coinductively interpreted grammars. For instance, the set~$S$ of
streams over a set~$A$ could be specified by writing
\[
S \Coloneqq \cons(A, S).
\]
A more interesting example is that of finite and infinite binary trees
with nodes labelled either with~$a$ or~$b$, and leaves labelled with
one of the elements of a set~$V$:
\[
T \Coloneqq V \parallel a(T, T) \parallel b(T, T).
\]
As such grammars are not intended to be formal entities but only
convenient visual means for specifying sets of coterms, we will not
define them precisely. It is always clear what many-sorted signature
is meant.

\newcommand{\subst}{\mathtt{subst}}

\begin{example}\label{ex_epsilon_lambda_terms}
  We define the set~$\Lambda$ of infinitary $\epsilon$-$\lambda$-terms by
  \[
  \Lambda \Coloneqq V \parallel \Lambda \Lambda \parallel \lambda V
  . \Lambda \parallel \epsilon(\Lambda)
  \]
  where~$V$ is a set of variables. For $s,t\in\Lambda$ and $x \in V$,
  the operation of substitution $\subst_x : \Lambda \times \Lambda \to
  \Lambda$ is defined by guarded
  corecursion
  \[
  \begin{array}{rcl}
    x[t/x] &=& t \\
    y[t/x] &=& y \quad\text{ if } x \ne y \\
    (s_1s_2)[t/x] &=& (s_1[t/x])(s_2[t/x]) \\
    (\lambda y . s)[t/x] &=& \lambda y . s[t/x] \quad\text{ if } x \ne
    y \\
    (\lambda x . s)[t/x] &=& \lambda x . s \\
    (\epsilon(s))[t/x] &=& \epsilon(s[t/x])
  \end{array}
  \]
  where $s[t/x] = \subst_x(s,t)$.

  Note that substitution defined in this way may capture
  variables. For the sake of simplicity, we disregard this problem by
  assuming that in all terms the free variables are distinct from the
  bound ones.

  A production function~$\eta_\subst$ for~$\subst_x$ is given by the
  equations\footnote{We again implicitly use Lemma~\ref{lem_cases},
    and Theorem~\ref{thm_corecursion}.}
  \[
  \begin{array}{rcl}
    \eta_\subst(0,m) &=& 0 \\
    \eta_\subst(n,0) &=& 0 \\
    \eta_\subst(n+1,m) &=& \min(m,n+1,\eta_\subst(n,m)+1)
  \end{array}
  \]
  Thus $\eta_\subst(n,m) = \min(n,m)$.

  The definition of substitution on infinitary $\epsilon$-$\lambda$-terms
  will be used in an example in the next section.\hfill$\Box$
\end{example}

\subsection{Coinduction}\label{sec_coinduction}

Coinduction is a method of proving statements involving coinductive
relations, i.e., relations defined as greatest fixpoints of certain
monotone operators. Coinductive relations are most useful in
conjunction with infinite objects and corecursively defined functions.

\begin{definition}\label{def_coinductive_relation}
  Let $A$ be a set. An $n$-ary relation $R \subseteq A^n$ is a
  \emph{coinductive relation} if it is the greatest fixpoint of some
  monotone function $F : \Pow{A^n} \to \Pow{A^n}$. Since~$\Pow{A^n}$
  is a complete lattice, for any monotone function $F : \Pow{A^n} \to
  \Pow{A^n}$ there exists an associated coinductive relation $R = \nu
  F$, and it is the limit of the final sequence of~$F$. The final
  sequence of a function for a coinductive relation~$R$ will be
  denoted by~$(R^\alpha)_\alpha$. The $\alpha$-th element~$R^\alpha$
  of the final sequence is called the \emph{approximant of~$R$ at
    stage~$\alpha$}. If $\la x_1,\ldots,x_n\ra \in R$ then we say that
  $R(x_1,\ldots,x_n)$ (holds). If $R^\alpha(x_1,\ldots,x_n)$ then we
  say that $R(x_1,\ldots,x_n)$ (holds) at (stage)~$\alpha$.
\end{definition}

Note that the approximants~$R^\alpha$ of a coinductive relation~$R$
depend on the function~$F$ of which~$R$ is the greatest fixpoint,
i.e., they depend on a particular definition of~$R$, not on the
relation~$R$ itself.

\begin{example}\label{ex_coinductive_def}
  We define a set of coterms~$T$ by
  \[
  T \Coloneqq V \parallel A(T) \parallel B(T, T)
  \]
  where~$V$ is a countable set of variables, and~$A$, $B$ are
  constructors. By $x,y,\ldots$ we denote variables, and by
  $t,s,\ldots$ we denote coterms (i.e.~elements of~$T$).

  We define a coinductive relation ${\to} \subseteq T \times T$ by a
  set of coinductive rules:
  \[
  \infer=[(1)]{x \to x}{} \quad
  \infer=[(2)]{A(t) \to A(t')}{t \to t'} \quad
  \infer=[(3)]{B(s,t) \to B(s',t')}{s \to s' & t \to t'} \quad
  \infer=[(4)]{A(t) \to B(t',t')}{t\to t'}
  \]
  Formally, the relation~${\to}$ is the greatest fixpoint of a
  monotone \mbox{$F : \Pow{T \times T} \to \Pow{T \times T}$} defined
  by
  \[
  F(R) = \left\{ \la t_1, t_2 \ra \mid \exists_{x \in V}(t_1 \equiv
    t_2 \equiv x) \lor \exists_{t,t'\in T}(t_1 \equiv A(t) \land t_2
    \equiv B(t',t') \land R(t,t')) \lor \ldots \right\}.
  \]
  It is always straightforward to convert rules of the above form into
  an appropriate monotone function (provided the rules actually are
  monotone). We shall always leave this conversion implicit.

  Alternatively, using the Knaster-Tarski fixpoint theorem, the
  relation~$\to$ may be characterized as the greatest binary relation
  on~$T$ (i.e. the greatest subset of $T\times T$ w.r.t.~set
  inclusion) such that ${\to} \subseteq F({\to})$, i.e., such that for
  every $t_1,t_2 \in T$ with $t_1 \to t_2$ one of the following holds:
  \begin{enumerate}
  \item $t_1 \equiv t_2 \equiv x$ for some variable $x \in V$,
  \item $t_1 \equiv A(t)$, $t_2 \equiv A(t')$ with $t \to t'$,
  \item $t_1 \equiv B(s,t)$, $t_2 \equiv B(s',t')$ with $s \to s'$ and
    $t \to t'$,
  \item $t_1 \equiv A(t)$, $t_2 \equiv B(t',t')$ with $t \to t'$.
  \end{enumerate}

  Yet another way to think about~$\to$ is that $t_1 \to t_2$ holds iff
  there exists a \emph{potentially infinite} derivation tree of $t_1
  \to t_2$ built using the rules~$(1)-(4)$.

  The rules~$(1)-(4)$ could also be interpreted inductively to yield
  the least fixpoint of~$F$. This is the conventional interpretation,
  and it is indicated with single line in each rule separating
  premises from the conclusion. A coinductive interpretation is
  indicated with double lines.

  It is instructive to note that the coinductive rules may also be
  interpreted as giving (ordinary) rules for approximants at each
  successor ordinal stage~$\alpha+1$.
  \[
  \begin{array}{c}
    \infer[(1)]{x \to x \text{ at } \alpha+1}{} \quad
    \infer[(2)]{A(t) \to A(t') \text{ at } \alpha+1}{
      t \to t' \text{ at } \alpha} \quad
    \infer[(3)]{B(s,t) \to B(s',t') \text{ at } \alpha+1}{
      s \to s' \text{ at } \alpha & t \to t' \text{ at } \alpha} \\ \\
    \infer[(4)]{A(t) \to B(t',t') \text{ at } \alpha+1}{
      t\to t' \text{ at } \alpha}
  \end{array}
  \]
  This follows directly from the way~$F$ and the approximants are
  defined. We will often use this observation implicitly.

  Usually, the closure ordinal for the definition of a coinductive
  relation is~$\omega$. In general, however, it is not difficult to
  come up with a coinductive definition whose closure ordinal is
  greater than~$\omega$. For instance, consider the relation $R
  \subseteq \Nbb \cup \{\infty\}$ defined coinductively by the
  following two rules.
  \[
  \infer={R(n+1)}{R(n) & n \in \Nbb} \quad\quad
  \infer={R(\infty)}{\exists n \in \Nbb . R(n)}
  \]
  We have $R = \emptyset$, $R^n = \{m \in \Nbb \mid m \ge n\} \cup
  \{\infty\}$ for $n \in \Nbb$, $R^\omega = \{\infty\}$, and only
  $R^{\omega+1}=\emptyset$. Thus the closure ordinal of this
  definition is $\omega+1$.\hfill$\Box$
\end{example}

\begin{definition}\label{def_formulas}
  Let~$\Sigma$ be a first-order signature. The \emph{first-order
    language} over the signature~$\Sigma$ is defined in the standard
  way, except that we additionally allow free relational variables
  (but not bound ones -- quantification is only over individuals). We
  use the symbol~$\equiv$ to denote syntactic identity of terms and
  formulas.

  A \emph{sentence} is a formula without free variables (relational or
  individual). Given a $\Sigma$-structure~$\Ab$ and a
  sentence~$\varphi$, we write $\Ab \models \varphi$ if~$\varphi$ is
  true in~$\Ab$. If~$\Ab$ is clear or irrelevant, we sometimes simply
  say that~$\varphi$ holds.

  Since we will usually work with a fixed structure~$\Ab$, to save on
  notation we often confuse function and relation symbols in the
  language with corresponding functions and relations on~$\Ab$. We
  will also often confuse a structure~$\Ab$ with its carrier
  set. Moreover, we usually implicitly assume that in the
  signature~$\Sigma$ there is a corresponding constant (i.e.~a nullary
  function symbol) for every element of~$\Ab$.

  If $\Sigma \subseteq \Sigma'$ and~$\Ab$ is a $\Sigma$-structure,
  then a \emph{$\Sigma'$-expansion} of~$\Ab$ is a
  $\Sigma'$-structure~$\Ab'$ with the same carrier set and the same
  interpretation of symbols from~$\Sigma$ as~$\Ab$.

  We write $\varphi \equiv \varphi(\bar{x},\bar{X}) \equiv
  \varphi(x_1,\ldots,x_n,X_1,\ldots,X_m)$ for a formula with all free
  individual variables among $x_1,\ldots,x_n$, and all free relational
  variables among $X_1,\ldots,X_m$. We then write
  $\varphi(t_1,\ldots,t_n,R_1,\ldots,R_m)$ to denote~$\varphi$ with
  each~$x_i$ substituted with the term~$t_i$, and each~$X_i$
  substituted with the relation symbol~$R_i$.

  A formula~$\varphi$ is in \emph{prenex normal form} if
  $\varphi\equiv \forall_{x_1}\exists_{y_1}\forall_{x_2}\exists_{y_2}
  \ldots\forall_{x_n}\exists_{y_n}\psi$ where~$\psi$ is
  quantifier-free. It is a standard result in elementary logic that
  any first order formula may be effectively translated into an
  equivalent formula in prenex normal form. A formula~$\varphi$ is
  \emph{universal} if it is equivalent to a formula
  $\forall_{x_1}\forall_{x_2}\ldots\forall_{x_n}\psi$ with~$\psi$
  quantifier-free. A formula~$\varphi$ is \emph{standard} if it is
  equivalent to a conjunction of formulas of the form\footnote{The
    individual variables $x_1,\ldots,x_n$ may of course occur in the
    terms~$t_i^j$.}
  \[
  \forall_{x_1}\ldots\forall_{x_n}\left(\psi(x_1,\ldots,x_n,X_1,\ldots,X_m)
    \to X_{i_1}(t_1^1,\ldots,t_{n_1}^1) \land \ldots \land
    X_{i_k}(t_1^k,\ldots,t_{n_k}^k)\right)
  \]
  where~$\psi$ is quantifier-free.
\end{definition}

The following simple but important theorem states the coinduction
principle.

\begin{theorem}[Coinduction principle]\label{thm_coinduction}
  Let~$\Sigma$ be a first-order signature, $\varphi(X_1,\ldots,X_m)$ a
  standard formula over~$\Sigma$, and~$\Ab$ a $\Sigma$-structure. Let
  $R_1,\ldots,R_m$ be coinductive relations on~$\Ab$, with arities
  matching the arities of $X_1,\ldots,X_m$. Suppose the
  \emph{coinductive step} holds:
  \begin{itemize}
  \item for every ordinal~$\alpha$, if $\Ab \models
    \varphi(R_1^\alpha,\ldots,R_m^\alpha)$ then \mbox{$\Ab \models
      \varphi(R_1^{\alpha+1},\ldots,R_m^{\alpha+1})$}.
  \end{itemize}
  Then $\Ab \models \varphi(R_1^\alpha,\ldots,R_m^\alpha)$ for every
  ordinal~$\alpha$. In particular, $\Ab \models
  \varphi(R_1,\ldots,R_m)$.
\end{theorem}

\begin{proof}
  By transfinite induction on~$\alpha$ we show $\Ab \models
  \varphi(R_1^\alpha,\ldots,R_m^\alpha)$. For $\alpha=0$ this follows
  from the fact that~$\varphi$ is standard and each~$R_i^0$ is a full
  relation, i.e., $R_i^0 = \Ab^k$ for some~$k > 0$. For~$\alpha$ a
  successor ordinal this follows from the coinductive step. So
  assume~$\alpha$ is a limit ordinal. Since~$\varphi$ is a universal
  formula, it is equivalent to a formula
  $\forall_{x_1}\ldots\forall_{x_n}\psi$ with~$\psi$ quantifier-free
  in disjunctive normal form. So $\psi \equiv \psi_1 \lor \ldots \lor
  \psi_k$ with each disjunct~$\psi_i$ a conjunction of literals. We
  need to show that for all $a_1,\ldots,a_n \in \Ab$ we have $\Ab
  \models \psi(a_1,\ldots,a_n,R_1^\alpha,\ldots,R_m^\alpha)$.

  Let $a_1,\ldots,a_n \in \Ab$. Let $\beta < \alpha$. Then $\Ab
  \models \varphi(R_1^\beta,\ldots,R_m^\beta)$, so $\Ab \models
  \psi(a_1,\ldots,a_n,R_1^\beta,\ldots,R_m^\beta)$. Hence $\Ab \models
  \psi_i(a_1,\ldots,a_n,R_1^\beta,\ldots,R_m^\beta)$ for some $1 \le i
  \le k$. Since the number~$k$ of disjuncts is finite, there must be
  $1 \le i \le k$ with $\Ab \models
  \psi_i(a_1,\ldots,a_n,R_1^\beta,\ldots,R_m^\beta)$ for arbitrarily
  large~$\beta < \alpha$, i.e., for every $\gamma < \alpha$ there is
  $\gamma \le \beta < \alpha$ with $\Ab \models
  \psi_i(a_1,\ldots,a_n,R_1^\beta,\ldots,R_m^\beta)$.

  Assume $\psi_i \equiv \theta_1 \land \ldots \land \theta_r$ with
  each~$\theta_j$ a literal. Thus $\Ab \models
  \theta_j(a_1,\ldots,a_n,R_1^\beta,\ldots,R_m^\beta)$ for arbitrarily
  large $\beta<\alpha$, for $j=1,\ldots,r$. It suffices to show
  \[
  \Ab \models \theta_j(a_1,\ldots,a_n,R_1^\alpha,\ldots,R_m^\alpha)
  \]
  for every $1 \le j \le r$. Let $1 \le j \le r$.

  If $\theta_j(a_1,\ldots,a_n,X_1,\ldots,X_m) \equiv
  \theta_j(a_1,\ldots,a_n)$, i.e., $\theta_j$ does not depend on the
  relational variables $X_1,\ldots,X_m$, then $\Ab \models
  \theta_j(a_1,\ldots,a_n,R_1^\alpha,\ldots,R_m^\alpha)$, because $\Ab
  \models \theta_j(a_1,\ldots,a_n,R_1^\beta,\ldots,R_m^\beta)$ for
  some $\beta < \alpha$, i.e., $\Ab \models \theta_j(a_1,\ldots,a_n)$.

  Now assume $\theta_j(a_1,\ldots,a_n,X_1,\ldots,X_m) \equiv \neg
  X_p(t_1,\ldots,t_q)$. Then $\Ab \models \neg
  R_p^\beta(t_1,\ldots,t_q)$ for some $\beta < \alpha$. We have
  $R_p^\alpha = \bigcap_{\beta<\alpha}R_p^\beta$ because~$\alpha$ is a
  limit ordinal (recall the definition of the final sequence in
  Definition~\ref{def_order}). Hence $\Ab \not\models
  R_p^\alpha(t_1,\ldots,t_q)$, and thus $\Ab \models \neg
  R_p^\alpha(t_1,\ldots,t_q)$, i.e., $\Ab \models
  \theta_j(a_1,\ldots,a_n,R_1^\alpha,\ldots,R_m^\alpha)$.

  So finally assume $\theta_j(a_1,\ldots,a_n,X_1,\ldots,X_m) \equiv
  X_p(t_1,\ldots,t_q)$. Then $\Ab \models R_p^\beta(t_1,\ldots,t_q)$
  for arbitrarily large $\beta < \alpha$. By Lemma~\ref{lem_sequence},
  if $\Ab \models R_p^\beta(t_1,\ldots,t_q)$ then $\Ab \models
  R_p^\gamma(t_1,\ldots,t_q)$ for all $\gamma \le \beta$. Thus in fact
  $\Ab \models R_p^\beta(t_1,\ldots,t_q)$ for all $\beta < \alpha$,
  i.e., $\Ab \models
  \bigcap_{\beta<\alpha}R_p^\beta(t_1,\ldots,t_q)$. Since $R_p^\alpha
  = \bigcap_{\beta<\alpha}R_p^\beta$, we have $\Ab \models
  R_p^\alpha(t_1,\ldots,t_q)$. Hence $\Ab \models
  \theta_j(a_1,\ldots,a_n,R_1^\alpha,\ldots,R_m^\alpha)$.
\end{proof}

\begin{example}\label{ex_to_refl}
  Let~$T$ be the set of coterms, and~${\to}$ the coinductive relation,
  from Example~\ref{ex_coinductive_def}. We show by coinduction that
  for arbitrary $t \in T$ we have $t \to t$. For the coinductive step,
  assume the coinductive hypothesis (CH), i.e., that for $\beta \le
  \alpha$: for all $t \in T$ we have $t \to t$ at
  stage~$\beta$. Consider possible forms of~$t$. If $t \equiv x \in
  V$, then $t \to t$ at~$\alpha+1$ by rule~$(1)$. If $t \equiv A(t')$
  then $t' \to t'$ at~$\alpha$ by the~CH, so $t \equiv A(t') \to A(t')
  \equiv t$ at~$\alpha+1$ by rule~$(2)$. If $t \equiv B(t_1,t_2)$ then
  $t_1 \to t_1$ at~$\alpha$ and $t_2 \to t_2$ at~$\alpha$ by the~CH,
  so $t \to t$ at~$\alpha+1$ by rule~$(3)$. Therefore, for all $t \in
  T$ we have $t \to t$ at~$\alpha+1$, which shows the coinductive
  step.

  The correctness of the above reasoning relies on
  Theorem~\ref{thm_coinduction}. The signature~$\Sigma$ and the
  structure~$\Ab$ are left implicit. For every function and relation
  on~$T$ that we use in the proof there is a corresponding symbol
  in~$\Sigma$. The structure~$\Ab$ has the set~$T$ as its carrier, and
  the expected interpretation of all symbols from~$\Sigma$ (as the
  corresponding actual functions and relations on~$T$).

  Usually, we do not get into so much detail when doing coinductive
  proofs. The ordinal stages are also left implicit, unless they occur
  in the statement we ultimately want to show or the argument that the
  stage increases is not completely trivial. Below we give the proof
  again in a style which we adopt from now on.

  We show by coinduction that if $t \in T$ then $t \to t$. If $t
  \equiv x$ then this follows by rule~$(1)$. If $t \equiv A(t')$ then
  $t' \to t'$ by the~CH, so $t \to t$ by rule~$(2)$. If $t \equiv
  B(t_1,t_2)$ then $t_1 \to t_1$ and $t_2 \to t_2$ by the~CH, so $t
  \to t$ by rule~$(3)$.\hfill$\Box$
\end{example}

When doing a proof by coinduction one must be careful to ensure that
the implicit stages actually do increase. The most common way to
ensure this is to immediately provide the conclusion of the
coinductive hypothesis as a premise of some coinductive rule, since
applying a rule increases the stage. Note that $R^\alpha \subseteq
R^\beta$ for $\beta < \alpha$, by Lemma~\ref{lem_sequence}. This has
the important practical consequence that we do not have to
worry\footnote{As long as we are showing a statement with only
  positive occurences of the coinductive relations for which we
  (implicitly) track the stages.} to increase the stage by exactly
one, as it would at first sight seem necessary from the statement of
Theorem~\ref{thm_coinduction}. We may increase it by an arbitrary $n >
0$, and the proof is still correct. In particular, it is harmless to
apply rules repeatedly a positive number of times to a conclusion of
the coinductive hypothesis, e.g., to conclude~$R(x)$ (at~$\alpha$) by
the~CH, then to conclude~$R(s(x))$ (at~$\alpha+1$) by some rule~$(r)$
with~$R(x)$ (at~$\alpha$) as a premise, then conclude~$R(s(s(x)))$
(at~$\alpha+2$, so also at~$\alpha+1$ by Lemma~\ref{lem_sequence}) by
rule~$(r)$ with~$R(s(x))$ (at~$\alpha+1$) as a premise, finishing the
proof of the coinductive step.

In general, Lemma~\ref{lem_sequence} implies that we may always
decrease the stage of a coinductive relation. But to increase it we
need to apply at least one rule.

Note that because we are usually just interested in showing properties
of some coinductive relations on certain sets, we have some freedom in
choosing the signature~$\Sigma$ and the structure~$\Ab$ in
Theorem~\ref{thm_coinduction}, as well as the actual formula~$\varphi$
we want to prove. Hence the restriction on~$\varphi$ in
Theorem~\ref{thm_coinduction} to standard formulas is less limiting
than it might at first seem. For instance, suppose $\varphi(X) \equiv
\forall x ((\forall y \psi(x,y)) \to X(f(x)))$, i.e., $X$ does not
occur in~$\psi$. We are interested in showing $\Ab \models \varphi(R)$
for some structure~$\Ab$ and a coinductive relation~$R$. One cannot
apply Theorem~\ref{thm_coinduction} directly to~$\varphi$ because of
the negative occurence of the univeral quantifier $\forall y$ (the
prenex normal of~$\varphi$ has an existential quantifier). However,
one may add a new unary relation symbol~$r$ to the signature,
interpreted in an expansion~$\Ab'$ of~$\Ab$ by the relation $\{a \in
\Ab \mid \Ab \models \forall y \psi(a,y) \}$. Then $\Ab \models
\varphi(R)$ iff $\Ab' \models \forall x (r(x) \to R(f(x)))$. In
practice, we thus do not need to worry about negative (resp.~positive)
occurences of universal (resp.~existential) quantifiers which do not
have any relational variables within their scope.

\begin{example}\label{ex_subst}
  On coterms from~$T$ (from Example~\ref{ex_coinductive_def}) we
  define the operation of substitution by guarded corecursion.
  \[
  \begin{array}{rcl}
    y[t/x] &=& y \quad\text{ if } x \ne y \\
    x[t/x] &=& t \\
    (A(s))[t/x] &=& A(s[t/x]) \\
    (B(s_1,s_2))[t/x] &=& B(s_1[t/x],s_2[t/x])
  \end{array}
  \]
  We show by coinduction: if $s \to s'$ and $t \to t'$ then $s[t/x]
  \to s'[t'/x]$, where~$\to$ is the relation from
  Example~\ref{ex_coinductive_def}. Formally, the statement we show
  is: for $s,s',t,t' \in T$, if $s \to s'$ and $t \to t'$ then $s[t/x]
  \to s'[t'/x]$ at~$\alpha$. So we do not track the stages in the
  antecedent of the implication, as this is not necessary for the
  proof to go through. It is somewhat arbitrary how to choose the
  occurences of coinductive relations for which we track the
  stages. Generally, tracking stages for negative occurences makes the
  proof harder, while tracking them for positive occurences makes it
  easier. So we adopt the convention of tracking the stages only for
  positive occurences of coinductive relations, and leave this choice
  implicit.

  Let us proceed with the proof. The proof is by coinduction with case
  analysis on $s \to s'$. If $s \equiv s' \equiv y$ with $y \ne x$,
  then $s[t/x] \equiv y \equiv s'[t'/x]$. If $s \equiv s' \equiv x$
  then $s[t/x] \equiv t \to t' \equiv s'[t'/x]$ (at~$\alpha+1$ -- we
  implicitly use Lemma~\ref{lem_sequence} here). If $s \equiv A(s_1)$,
  $s' \equiv A(s_1')$ and $s_1 \to s_1'$, then $s_1[t/x] \to
  s_1'[t'/x]$ by the~CH. Thus $s[t/x] \equiv A(s_1[t/x]) \to
  A(s_1'[t'/x]) \equiv s'[t'/x]$ by rule~$(2)$. If $s \equiv
  B(s_1,s_2)$, $s' \equiv B(s_1',s_2')$ then the proof is
  analogous. If $s \equiv A(s_1)$, $s' \equiv B(s_1',s_1')$ and $s_1
  \to s_1'$, then the proof is also similar. Indeed, by the~CH we have
  $s_1[t/x] \to s_1'[t'/x]$, so $s[t/x] \equiv A(s_1[t/x]) \to
  B(s_1'[t'/x],s_1'[t'/x]) \equiv s'[t'/x]$ by rule~$(4)$.\hfill$\Box$
\end{example}

Let us reiterate the convention introduced in the above example.

\begin{paragraph}
  {\bf Important convention.} Unless explicitly stated otherwise, we
  track the stages only for positive occurences of coinductive
  relations, i.e., we do not treat negative occurences as relational
  variables in the formula we feed to
  Theorem~\ref{thm_coinduction}. For instance, let $f : T \to T$,
  let~$R \subseteq T$ be a coinductive relation, and suppose we want
  to show that for all $x \in T$ such that~$R(x)$ we
  have~$R(f(x))$. Then by default we take $\varphi(X) \equiv
  \forall_{x \in T} . R(x) \to X(f(x))$ to be the formula used with
  Theorem~\ref{thm_coinduction}. To override this convention one may
  mention the stages explicitly, e.g.: for all $x \in T$ such
  that~$R(x)$ at stage~$\alpha$ we have~$R(f(x))$ at
  stage~$\alpha$. Then the formula we take is $\varphi(X) \equiv
  \forall_{x \in T} . X(x) \to X(f(x))$. In conclusion, by default we
  track the stages of all positive occurences of coinductive
  relations, and only those negative occurences for which the stage is
  explicitly mentioned.
\end{paragraph}

\begin{definition}\label{def_bisimilarity}
  Let~$\Sigma$ be a many-sorted algebraic signature, as in
  Definition~\ref{def_coterms}. Let $\Tc = \Tc(\Sigma)$. Define
  on~$\Tc$ a binary relation~${=}$ of \emph{bisimilarity} by the
  coinductive rules
  \[
  \infer={f(t_1,\ldots,t_n) = f(s_1,\ldots,s_n)}{t_1 = s_1 & \ldots &
    t_n = s_n}
  \]
  for each $f \in \Sigma_f$.
\end{definition}

It is intuitively obvious that on coterms bisimilary is the same as
identity. The following easy theorem makes this precise.

\begin{theorem}\label{thm_bisimilarity}
  For $t,s \in \Tc$ we have: $t = s$ iff $t \equiv s$.
\end{theorem}

\begin{proof}
  Recall that each coterm is formally a function from~$\Nbb^*$
  to~$\Sigma_f \cup \{\bot\}$.

  Assume $t = s$. It suffices to show by induction of the length of $p
  \in \Nbb^*$ that $\pos{t}{p} = \pos{s}{p}$ or $\pos{t}{p} \equiv
  \pos{s}{p} \equiv \bot$, where by~$\pos{t}{p}$ we denote the subterm
  of~$t$ at position~$p$. For $p = \epsilon$ this is obvious. Assume
  $p = p'j$. By the inductive hypothesis (IH), $\pos{t}{p'} =
  \pos{s}{p'}$ or $\pos{t}{p'} \equiv \pos{s}{p'} \equiv \bot$. If
  $\pos{t}{p'} = \pos{s}{p'}$ then $\pos{t}{p'} \equiv
  f(t_0,\ldots,t_n)$ and $\pos{s}{p'} \equiv f(s_0,\ldots,s_n)$ for
  some $f \in \Sigma_f$ with $t_i = s_i$ for $i=0,\ldots,n$. If $0 \le
  j \le n$ then $\pos{t}{p} \equiv t_j = s_j = \pos{s}{p}$. Otherwise,
  if $j > n$ or if $\pos{t}{p'} \equiv \pos{s}{p'} \equiv \bot$, then
  $\pos{t}{p} \equiv \pos{s}{p} \equiv \bot$ by the definition of
  coterms.

  For the other direction, we show by coinduction that for any $t \in
  \Tc$ we have $t = t$. If $t \in \Tc$ then $t \equiv
  f(t_1,\ldots,t_n)$ for some $f \in \Sigma_f$. By the~CH we obtain
  $t_i = t_i$ for $i=1,\ldots,n$. Hence $t = t$ by the rule for~$f$.
\end{proof}

For coterms $t,s\in \Tc$, we shall theorefore use the notations $t =
s$ and $t \equiv s$ interchangeably, employing
Theorem~\ref{thm_bisimilarity} implicitly.

\begin{example}
  Recall the coinductive definitions of~$\zip$ and~$\even$ from
  Section~\ref{sec_corecursion_examples}.
  \[
  \begin{array}{rcl}
    \even(x : y : t) &=& x : \even(t) \\
    \zip(x : t, s) &=& x : \zip(s, t)
  \end{array}
  \]
  By coinduction we show
  \[
  \zip(\even(t),\even(\tail(t))) = t
  \]
  for any stream $t \in A^\omega$.

  Let $t \in A^\omega$. Then $t = x : y : s$ for some $x, y \in A$ and
  $s \in A^\omega$. We have
  \[
  \begin{array}{rcl}
    \zip(\even(t),\even(\tail(t))) &=& \zip(\even(x : y : s), \even(y
    : s)) \\
    &=& \zip(x : \even(s), \even(y : s)) \\
    &=& x : \zip(\even(y : s), \even(s)) \\
    &=& x : y : s \quad\text{ (by~CH) }\\
    &=& t
  \end{array}
  \]
  In the equality marked with~(by~CH) we use the coinductive
  hypothesis, and implicitly a bisimilarity rule from
  Definition~\ref{def_bisimilarity}.\hfill$\Box$
\end{example}

Theorem~\ref{thm_coinduction} gives a coinduction principle only for
standard formulas. By the discussion just above
Example~\ref{ex_subst}, this essentially means that we cannot do
coinductive proofs for formulas with some positive (resp.~negative)
occurences of existential (resp.~universal) quantifiers which have
some relational variables in their scope. However, even this is not so
much of a restriction as it may seem, because any formula without free
individual variables may be converted into Skolem normal form.

\begin{definition}
  Let $\varphi \equiv \forall_{x_1} \exists_{y_1} \ldots \forall_{x_n}
  \exists_{y_n} \psi(x_1,\ldots,x_n,y_1,\ldots,y_n,X_1,\ldots,X_k)$ be
  a formula over a signature~$\Sigma$, with~$\psi$
  quantifier-free. The \emph{Skolem normal form} of~$\varphi$ is
  \[
  \varphi^S \equiv \forall_{x_1} \ldots \forall_{x_n}
  \psi(x_1,\ldots,x_n,f_1(x_1),f_2(x_1,x_2),\ldots,f_n(x_1,\ldots,x_n),X_1,\ldots,X_k)
  \]
  where $f_1,\ldots,f_n$ are distinct new \emph{Skolem function
    symbols}, i.e., $f_1,\ldots,f_n \notin \Sigma$. The signature
  $\Sigma^S = \Sigma \cup \{f_1,\ldots,f_n\}$ is called a \emph{Skolem
    signature} for~$\varphi$. Thus~$\varphi^S$ is a formula
  over~$\Sigma^S$. The definition of Skolem normal form extends in a
  natural way to arbitrary formulas without free individual variables,
  by converting them into equivalent prenex normal form first. A
  \emph{Skolem expansion}~$\Ab^S$ of a $\Sigma$-structure~$\Ab$
  wrt.~$\varphi$ is a $\Sigma^S$-expansion of~$\Ab$. The functions
  interpreting Skolem function symbols in a Skolem expansion are
  called \emph{Skolem functions}.
\end{definition}

Let~$\Ab$ be a $\Sigma$-structure, and~$\varphi(X_1,\ldots,X_n)$ a
formula over~$\Sigma$. Let $R_1,\ldots,R_n$ be coinductive relations
on~$\Ab$ with matching arities. It is obvious that if there exists a
Skolem expansion~$\Ab^S$ of~$\Ab$ with $\Ab^S \models
\varphi^S(R_1^\alpha,\ldots,R_n^\alpha)$ for all ordinals~$\alpha$,
then $\Ab \models \varphi(R_1^\alpha,\ldots,R_n^\alpha)$ for all
ordinals~$\alpha$.

The method for showing by coinduction a
formula~$\varphi(R_1,\ldots,R_n)$ with existential quantifiers
occuring positively is to first convert~$\varphi$ into Skolem normal
form~$\varphi^S$ and find appropriate Skolem functions, and then show
using Theorem~\ref{thm_coinduction} that~$\varphi^S(R_1,\ldots,R_n)$
is true in the Skolem expansion. Usually, it is convenient to define
the required Skolem functions by corecursion, using methods from
Section~\ref{sec_corecursion}.

\begin{example}
  Let~$T$ be the set of coterms and~$\to$ the binary relation from
  Example~\ref{ex_coinductive_def}. We show: for all $s,t,t' \in T$,
  if $s \to t$ and $s \to t'$ then there exists $s' \in T$ with $t \to
  s'$ and $t' \to s'$. So we need to find a Skolem function $f : T
  \times T \times T \to T$ which will allow us to prove:
  \begin{itemize}
  \item[$(\star)$] if $s \to t$ and $s \to t'$ then $t \to f(s,t,t')$
    and $t' \to f(s,t,t')$.
  \end{itemize}
  The rules for~$\to$ suggest a definition of~$f$:
  \[
  \begin{array}{rcl}
    f(x, x, x) &=& x \\
    f(A(s), A(t), A(t')) &=& A(f(s,t,t')) \\
    f(A(s),A(t),B(t',t')) &=& B(f(s,t,t'),f(s,t,t')) \\
    f(A(s),B(t,t),A(t')) &=& B(f(s,t,t'),f(s,t,t')) \\
    f(A(s),B(t,t),B(t',t')) &=& B(f(s,t,t'),f(s,t,t')) \\
    f(B(s_1,s_2), B(t_1,t_2), B(t_1',t_2')) &=&
    B(f(s_1,t_1,t_1'),f(s_2,t_2,t_2')) \\
    f(s, t, t') &=& \text{some arbitrary coterm if none of the above matches}
  \end{array}
  \]
  The definition is guarded, so~$f$ is well-defined.

  We now proceed with a coinductive proof of~$(\star)$. Assume $s \to
  t$ and $s \to t'$. If $s \equiv t \equiv t' \equiv x$ then
  $f(s,t,t') = x$, and $x \to x$ by rule~$(1)$. If $s \equiv A(s_1)$,
  $t \equiv A(t_1)$ and $t' \equiv A(t_1')$ with $s_1 \to t_1$ and
  $s_1 \to t_1'$, then by the~CH $t_1 \to f(s_1,t_1,t_1')$ and $t_1'
  \to f(s_1,t_1,t_1')$. We have $f(s,t,t') \equiv
  A(f(s_1,t_1,t_1'))$. Hence $t \equiv A(t_1) \to f(s,t,t')$ and $t
  \equiv A(t_1') \to f(s,t,t')$, by rule~$(2)$. If $s \equiv
  B(s_1,s_2)$, $t \equiv B(t_1,t_2)$ and $t' \equiv B(t_1',t_2')$,
  with $s_1 \to t_1$, $s_1 \to t_1'$, $s_2 \to t_2$ and $s_2 \to
  t_2'$, then by the~CH we have $t_1 \to f(s_1,t_1,t_1')$, $t_1' \to
  f(s_1,t_1,t_1')$, $t_2 \to f(s_2,t_2,t_2')$ and $t_2' \to
  f(s_2,t_2,t_2')$. Hence $t \equiv B(t_1,t_2) \to
  B(f(s_1,t_1,t_1'),f(s_2,t_2,t_2')) \equiv f(s,t,t')$ by
  rule~$(3)$. Analogously, $t' \to f(s,t,t')$ by rule~$(3)$. Other
  cases are similar.

  Usually, it is inconvenient to invent Skolem functions beforehand,
  because definitions of the Skolem functions and the coinductive
  proof of the Skolem normal form are typically
  interdependent. Therefore, we adopt a style of doing a proof by
  coinduction of a formula~$\varphi(R_1,\ldots,R_m)$ in prenex normal
  form with existential quantifiers. We intertwine mutually
  corecursive definitions\footnote{Section~\ref{sec_corecursion}
    directly deals only with corecursive definitions of single
    functions, but mutual corecursion may be easily handled by
    considering an appropriate function on tuples of elements. See
    also Example~\ref{ex_mutual_coinduction} and
    Definition~\ref{def_mutual_corecursion}.} of Skolem functions with
  a coinductive proof of the Skolem normal
  form~$\varphi^S(R_1,\ldots,R_m)$. We pretend that the coinductive
  hypothesis is~$\varphi(R_1^\alpha,\ldots,R_m^\alpha)$. Each element
  obtained from an existential quantifier in the coinductive
  hypothesis is interpreted as a corecursive invocation of the
  corresponding Skolem function. When later we exhibit an element to
  show an existential subformula
  of~$\varphi(R_1^{\alpha+1},\ldots,R_m^{\alpha+1})$, we interpret
  this as the definition of the corresponding Skolem function in the
  case specified by the assumptions currently active in the
  proof. Note that this exhibited element may (or may not) depend on
  some elements obtained from existential quantifiers in the
  coinductive hypothesis, i.e., the definition of the corresponding
  Skolem function may involve corecursive invocations of Skolem
  functions.

  To illustrate the style of doing coinductive proofs of formulas with
  existential quantifiers, we redo the proof done above. For
  illustrative purposes, we indicate the arguments of the Skolem
  function, i.e., we write~$s'_{s,t,t'}$ in place
  of~$f(s,t,t')$. These subscripts $s,t,t'$ are normally omitted.

  We show by coinduction that if $s \to t$ and $s \to t'$ then there
  exists $s' \in T$ with $t \to s'$ and $t' \to s'$. Assume $s \to t$
  and $s \to t'$. If $s \equiv t \equiv t' \equiv x$ then take
  $s'_{x,x,x} = x$. If $s \equiv A(s_1)$, $t \equiv A(t_1)$ and $t'
  \equiv A(t_1')$ with $s_1 \to t_1$ and $s_1 \to t_1'$, then by
  the~CH we obtain\footnote{More precisely: by corecursively applying
    the Skolem function to $s_1,t_1,t_1'$ we
    obtain~$s'_{s_1,t_1,t_1'}$, and by the coinductive hypothesis we
    have $t_1 \to s'_{s_1,t_1,t_1'}$ and $t_1' \to
    s'_{s_1,t_1,t_1'}$.}~$s'_{s_1,t_1,t_1'}$ with $t_1 \to
  s'_{s_1,t_1,t_1'}$ and $t_1' \to s'_{s_1,t_1,t_1'}$. Hence $t \equiv
  A(t_1) \to A(s'_{s_1,t_1,t_1'})$ and $t \equiv A(t_1') \to
  A(s'_{s_1,t_1,t_1'})$, by rule~$(2)$. Thus we may take $s'_{s,t,t'}
  = A(s'_{s_1,t_1,t_1'})$. If $s \equiv B(s_1,s_2)$, $t \equiv
  B(t_1,t_2)$ and $t' \equiv B(t_1',t_2')$, with $s_1 \to t_1$, $s_1
  \to t_1'$, $s_2 \to t_2$ and $s_2 \to t_2'$, then by the~CH we
  obtain~$s'_{s_1,t_1,t_1'}$ and~$s'_{s_2,t_2,t_2'}$ with $t_1 \to
  s'_{s_1,t_1,t_1'}$, $t_1' \to s'_{s_1,t_1,t_1'}$, $t_2 \to
  s'_{s_2,t_2,t_2'}$ and $t_2' \to s'_{s_2,t_2,t_2'}$. Hence $t \equiv
  B(t_1,t_2) \to B(s'_{s_1,t_1,t_1'},s'_{s_2,t_2,t_2'})$ by
  rule~$(3)$. Analogously, $t' \to
  B(s'_{s_1,t_1,t_1'},s'_{s_2,t_2,t_2'})$ by rule~$(3)$. Thus we may
  take $s'_{s,t,t'} \equiv
  B(s'_{s_1,t_1,t_1'},s'_{s_2,t_2,t_2'})$. Other cases are similar.

  It is clear that the above proof, when interpreted in the way
  outlined before, implicitly defines the Skolem function~$f$. Also,
  in each case a local prefix production function is implicitly
  defined. From Corollary~\ref{cor_unique_solution} it follows that to
  justify the well-definedness of the implicit Skolem function it
  suffices to bound a local prefix production function for each case
  separately. If the definition is guarded in a given case, the
  well-definedness argument for this case is left implicit. Otherwise,
  a justification is needed.

  Note that for a coinductive proof to implicitly define a Skolem
  function, the elements exhibited for existential statements
  \emph{must not depend on the (implicit) stage~$\alpha$}. In other
  words, the Skolem functions must be the same for all~$\alpha$. This
  is the reason why Theorem~\ref{thm_coinduction} does not generalize
  to arbitrary formulas in the first place. However, it is usually the
  case that there is no dependency on~$\alpha$, and thus the
  justification of this is typically left implicit. But the necessity
  of this requirement should be kept in mind.\hfill$\Box$
\end{example}

\begin{example}
  We now give an example of an \emph{incorrect} coinductive
  argument. Let~$\to$ and~$T$ be like in the previous
  example. Define~$\to_i$ \emph{inductively} by the rules $(1)-(4)$
  from Example~\ref{ex_coinductive_def}. We show: if $s \to t$ and $s
  \to t'$ then there exists~$s'$ such that $t \to s'$ and $t' \to_i
  s'$. By inspecting the proof in the previous example one sees that
  it also works for the new statement. Simply, we need to change~$\to$
  to~$\to_i$ in certain places. The proof is still correct -- we just
  no longer need to track stages for the occurences of~$\to$ replaced
  by~$\to_i$.

  What is wrong with this argument? The modified coinductive step is
  indeed correct, but the formula we show is no longer standard, so
  Theorem~\ref{thm_coinduction} cannot be applied. Formally, we now
  show~$\varphi(\to^\alpha)$ for each ordinal~$\alpha$, where
  $\varphi(X) \equiv \forall s, t, t' \in T . \exists s' \in T . (s
  \to t \land s \to t') \to (X(t, s') \land t' \to_i s')$
  and~$\to^\alpha$ is the approximant of~$\to$ at stage~$\alpha$. In
  fact, $\varphi(\to^0)$ is false -- e.g.~if~$t'$ is infinite then
  there is no~$s'$ such that $t' \to_i s'$. \hfill$\Box$
\end{example}

We finish this section with a complex example of a proof of the
diamond property of a certain relation on infinitary
$\epsilon$-$\lambda$-terms.

\begin{definition}
  The binary relation~$\to_1$ on infinitary
  $\epsilon$-$\lambda$-terms~$\Lambda$ from
  Example~\ref{ex_epsilon_lambda_terms} is defined by the following
  coinductive rules.
  \[
  \infer=[(1)]{x \to_1 x}{}\quad
  \infer=[(2)]{st \to_1 s't'}{s \to_1 s' & t \to_1 t'}\quad
  \infer=[(3)]{\lambda x . s \to_1 \lambda x . s'}{s \to_1 s'}\quad
  \infer=[(4)]{(\lambda x . s)t \to_1 \epsilon(s'[t'/x])}{
    s \to_1 s' & t \to_1 t'}\quad
  \]
  \[
  \infer=[(5)]{\epsilon(t) \to_1 \epsilon(t')}{t \to_1 t'}
  \]
\end{definition}

\begin{lemma}\label{lem_to_1_refl}
  For $t \in \Lambda$ we have $t \to_1 t$.
\end{lemma}

\begin{proof}
  Coinduction. If $t \equiv x$ then $t \to_1 t$ by rule~$(1)$. If $t
  \equiv t_1t_2$ then $t_1 \to_1 t_1$ and $t_2 \to_1 t_2$ by
  the~CH. Thus $t \to_1 t$ by rule~$(2)$. Other cases are analogous.
\end{proof}

\newcommand{\FV}{\mathrm{FV}}

\begin{lemma}\label{lem_subst_equiv}
  If $y \notin \FV(t)$ then $s_1[s_2/y][t/x] \equiv
  s_1[t/x][s_2[t/x]/y]$.
\end{lemma}

\begin{proof}
  By coinduction, implicitly using Theorem~\ref{thm_bisimilarity}. If
  $s_1 \equiv y$ with $x \ne y$, then $s_1[s_2/y][t/x] \equiv s_2[t/x]
  \equiv s_1[t/x][s_2[t/x]/y]$, because $s_1[t/x] \equiv y[t/x] \equiv
  y$. If $s_1 \equiv x$ then $s_1[s_2/y][t/x] \equiv x[t/x] \equiv t
  \equiv s_1[t/x] \equiv s_1[t/x][s_2[t/x]/y]$, because $y \notin
  \FV(t)$. If $s_1 \equiv u_1u_2$ then $u_i[s_2/y][t/x] \equiv
  u_i[t/x][s_2[t/x]/y]$ by the~CH. Hence
  \[
  \begin{array}{rcl}
    s_1[s_2/y][t/x] &\equiv& (u_1[s_2/y][t/x])(u_2[s_2/y][t/x]) \\ &\equiv&
    (u_1[t/x][s_2[t/x]/y])(u_2[t/x][s_2[t/x]/y]) \\ &\equiv&
    s_1[t/x][s_2[t/x]/y].
  \end{array}
  \]
  If $s_1 \equiv \lambda z . s_1'$ with\footnote{Recall that we assume
    bound variables to be distinct from the free ones.} $z \ne x,y$
  then $s_1'[s_2/y][t/x] \equiv s_1'[t/x][s_2[t/x]/y]$ by the~CH. Thus
  \[
  s_1[s_2/y][t/x] \equiv \lambda z . s_1'[t/x][s_2[t/x]/y] \equiv
  \lambda z . s_1'[t/x][s_2[t/x]/y] \equiv s_1[t/x][s_2[t/x]/y].
  \]
  If $s_1 \equiv \epsilon(s_1')$ then the proof is analogous.
\end{proof}

\begin{lemma}\label{lem_to_1_subst}
  If $s \to_1 s'$ at~$\alpha$ and $t \to_1 t'$ at~$\alpha$ then
  $s[t/x] \to_1 s'[t'/x]$ at~$\alpha$.
\end{lemma}

\begin{proof}
  We proceed by coinduction. The coinductive hypothesis is: for all
  $s,s',t,t' \in \Lambda$, $x \in V$, if $s \to_1 s'$ at~$\alpha$ and
  $t \to_1 t'$ at~$\alpha$ then $s[t/x] \to_1 s'[t'/x]$
  at~$\alpha$. The statement that we need to show in the inductive
  step is: for all $s,s',t,t' \in \Lambda$, $x \in V$, if $s \to_1 s'$
  at~$\alpha+1$ and $t \to_1 t'$ at~$\alpha+1$ then $s[t/x] \to_1
  s'[t'/x]$ at~$\alpha+1$.

  So assume $s \to_1 s'$ at~$\alpha+1$ and $t \to_1 t'$
  at~$\alpha+1$. If $s \equiv s' \equiv x$ then $s[t/x] \equiv t \to_1
  t' \equiv s'[t'/x]$ at~$\alpha+1$. If $s \equiv s' \equiv y$ with $x
  \ne y$ then $s[t/x] \equiv y \equiv s'[t'/x]$, so $s[t/x] \to_1
  s'[t'/x]$ at~$\alpha+1$ by Lemma~\ref{lem_to_1_refl}. If $s \equiv
  s_1s_2$ and $s' \equiv s_1's_2'$ with $s_1 \to_1 s_1'$ at~$\alpha$
  and $s_2 \to_1 s_2'$ at~$\alpha$, then\footnote{Recall that $t \to_1
    t'$ at~$\alpha+1$ implies $t \to_1 t'$ at~$\alpha$, by
    Lemma~\ref{lem_sequence}.} $s_1[t/x] \to_1 s_1'[t'/x]$ at~$\alpha$
  and $s_2[t/x] \to_1 s_2'[t'/x]$ at~$\alpha$ by the~CH. Thus $s[t/x]
  \equiv (s_1[t/x])(s_2[t/x]) \to_1 (s_1'[t'/x])(s_2'[t'/x]) \equiv
  s'[t'/x]$ at~$\alpha+1$ by rule~$(2)$. If $s \equiv \lambda y
  . s_1$, $s' \equiv \lambda y . s_1'$ and $s_1 \to_1 s_1'$
  at~$\alpha$, then $s_1[t/x] \to_1 s_1'[t'/x]$ at~$\alpha$ by
  the~CH. Thus $s[t/x] \equiv \lambda y . s_1[t/x] \to_1 \lambda y
  . s_1'[t'/x] \equiv s'[t'/x]$ at~$\alpha+1$ by rule~$(3)$. If $s
  \equiv (\lambda y . s_1) s_2$ and $s' \equiv \epsilon(s_1'[s_2'/y])$
  with $s_1 \to_1 s_1'$ at~$\alpha$ and $s_2 \to_1 s_2'$ at~$\alpha$,
  then $s_1[t/x] \to_1 s_1'[t'/x]$ at~$\alpha$ and $s_2[t/x] \to_1
  s_2'[t'/x]$ at~$\alpha$ by the~CH. By Lemma~\ref{lem_subst_equiv} we
  have $s'[t'/x] \equiv \epsilon(s_1'[s_2'/y][t'/x]) \equiv
  \epsilon(s_1'[t'/x][s_2'[t'/x]/y])$. Thus $s[t/x] \equiv (\lambda y
  . s_1[t/x]) s_2[t/x] \to_1 \epsilon(s_1'[t'/x][s_2'[t'/x]/y]) \equiv
  s'[t'/x]$ at~$\alpha+1$ by rule~$(4)$. Finally, if $s \equiv
  \epsilon(s_1)$, $s' \equiv \epsilon(s_1')$ and $s_1 \to_1 s_1'$
  at~$\alpha$, then $s_1[t/x] \to_1 s_1'[t'/x]$ at~$\alpha$ by
  the~CH. Thus $s[t/x] \equiv \epsilon(s_1[t/x]) \to_1
  \epsilon(s_1'[t'/x]) \equiv s'[t'/x]$ at~$\alpha+1$ by rule~$(5)$.
\end{proof}

\begin{proposition}
  If $s \to_1 t$ and $s \to_1 t'$ then there exists~$s'$ with $t \to_1
  s'$ and $t' \to_1 s'$.
\end{proposition}

\begin{proof}
  By coinduction. If $s \equiv t \equiv t \equiv x$ then take $s'
  \equiv x$. If $s \equiv s_1s_2$, $t \equiv t_1t_2$ and $t' \equiv
  t_1't_2'$ with $s_i \to_1 t_i$ and $s_i \to_1 t_i'$, then by the~CH
  we obtain~$s_1'$ and~$s_2'$ with $t_i \to_1 s_i'$ and $t_i' \to_1
  s_i'$. Thus $t_1t_2 \to_1 s_1's_2'$ and $t_1't_2' \to_1 s_1's_2'$ by
  rule~$(2)$, and we may take $s' \equiv s_1's_2'$.

  If $s \equiv (\lambda x . s_1)s_2$, $t \equiv (\lambda x . t_1)t_2$
  and $t' \equiv \epsilon(t_1'[t_2'/x])$ with $s_i \to_1 t_i$ and $s_i
  \to_1 t_i'$, then by the~CH we obtain~$s_1'$ and~$s_2'$ with $t_i
  \to_1 s_i'$ at~$\alpha$ and $t_i' \to_1 s_i'$ at~$\alpha$. We have
  $t \equiv (\lambda x . t_1) t_2 \to_1 \epsilon(s_1'[s_2'/x])$
  at~$\alpha+1$ by rule~$(4)$. By Lemma~\ref{lem_to_1_subst} we have
  $t_1'[t_2'/x] \to_1 s_1'[s_2'/x]$ at~$\alpha$, so $t' \equiv
  \epsilon(t_1'[t_2'/x]) \to_1 \epsilon(s_1'[s_2'/x])$ at~$\alpha+1$
  by rule~$(5)$. Therefore take $s' \equiv \epsilon(s_1'[s_2'/x])$. It
  remains to justify the well-definedness of the implicit Skolem
  function in this case -- note that its definition is not guarded
  because we apply the substitution operation to results of
  corecursive invocations ($s_1'$,$s_2'$). However, a local prefix
  production function for this case is $\xi(n,m) = \eta_\subst(n,m) +
  1 = \min(n,m) + 1 > \min(n,m)$ and well-definednes follows.

  Assume $s \equiv (\lambda x . s_1)s_2$, $t \equiv
  \epsilon(t_1[t_2/x])$ and $t' \equiv \epsilon(t_1'[t_2'/x])$ with
  $s_i \to_1 t_i$ and $s_i \to_1 t_i'$. By the~CH we
  obtain~$s_1'$,$s_2'$ with $t_i \to_1 s_i'$ at~$\alpha$ and $t_i'
  \to_1 s_i'$ at~$\alpha$. By Lemma~\ref{lem_to_1_subst} we have
  $t_1[t_2/x] \to_1 s_1'[s_2'/x]$ at~$\alpha$ and $t_1'[t_2'/x] \to_1
  s_1'[s_2'/x]$ at~$\alpha$. Thus $t \equiv \epsilon(t_1[t_2/x]) \to_1
  \epsilon(s_1'[s_2'/x])$ at~$\alpha+1$ and $t' \equiv
  \epsilon(t_1'[t_2'/x]) \to_1 \epsilon(s_1'[s_2'/x])$ at~$\alpha+1$,
  by rule~$(5)$. Therefore take $s' \equiv \epsilon(s_1'[s_2'/x])$. A
  local prefix production function for this case is $\xi(n,m) =
  \eta_\subst(n,m) + 1 = \min(n,m) + 1 > \min(n,m)$, which implies
  well-definedness.

  Other cases are similar and left to the reader.
\end{proof}

Note that the two last cases considered in the proof above would not
go through if rule~$(4)$ was simply
\[
\infer={(\lambda x . s) t \to_1 s'[t'/x]}{s \to_1 s' & t \to_1 t'}
\]

\subsection{Nested induction and
  coinduction}\label{sec_nested_coinduction}

It is often useful to mix coinduction with induction, or to nest
coinductive definitions. For instance, the definition
from~\cite{EndrullisHansenHendriksPolonskySilva2018,EndrullisHansenHendriksPolonskySilva2015}
of infinitary reduction of arbitrary ordinal length in infinitary term
rewriting systems uses mixed induction-coinduction. Some other
examples may be found
in~\cite{AltenkirchDanielsson2009,NakataUustalu2010,BezemNakataUustalu2012}. In
this section we give a few example proofs and definitions which nest
induction and/or coinduction.

\begin{example}
  Define the set~$T$ coinductively:
  \[
  T \Coloneqq A T \parallel B T
  \]
  For $X \subseteq T$, we define the relation $R(X) \subseteq T$
  coinductively.
  \[
  \infer={A t \in R(X)}{t \in X}\quad\quad\infer={B t \in R(X)}{t \in R(X)}
  \]
  For $X \subseteq T$, the relation $S(X) \subseteq T$ is defined
  inductively.
  \[
  \infer{A t \in S(X)}{t \in S(X)}\quad\quad\infer{B t \in S(X)}{t \in
    X}
  \]
  Both~$R$ and~$S$ are monotone in~$X$, i.e., $X \subseteq Y$ implies
  $R(X) \subseteq R(Y)$ and $S(X) \subseteq S(Y)$. Hence, the
  following definitions of~$Q_1,Q_2 \subseteq T$ make sense.
  \[
  \begin{array}{c}
    \infer={A t \in Q_1}{t \in S(Q_1)}\quad\quad\infer={B t \in Q_1}{t \in
      Q_1}
    \\ \\
    \infer{A t \in Q_2}{t \in Q_2}\quad\quad\infer{B t \in Q_2}{t \in
      R(Q_2)}
  \end{array}
  \]
  Intuitively, $t \in Q_1$ means that~$t$ contains infinitely
  many~$B$s, and $t \in Q_2$ means that~$t$ contains only finitely
  many~$A$s.

  First, we show $Q_1 \subseteq S(Q_1)$. Let $t \in Q_1$. If $t \equiv
  A t'$ then $t' \in S(Q_1)$, so $A t' \in S(Q_1)$. If $t \equiv B t'$
  then $t' \in Q_1$, so $B t' \in S(Q_1)$.

  Now we show that if $t \in Q_2$ then $t \in Q_1$. The proof proceeds
  by induction on the length of derivation of $t \in Q_2$. Let $t \in
  Q_2$. If $t \equiv A t'$ then $t' \in Q_2$, so $t' \in Q_1$ by the
  inductive hypothesis. Since $Q_1 \subseteq S(Q_1)$ we have $t \equiv
  A t' \in Q_1$. If $t \equiv A t'$ then $t' \in R(Q_2')$ where~$Q_2'$
  is the set of~$s \in Q_2$ with shorter derivations than $t \in
  Q_2$. By nested coinduction we show that if $t' \in R(Q_2')$ then
  $t' \in Q_1$. This actually follows from the inductive hypothesis
  (which implies $Q_2' \subseteq Q_1$), the monotonicity of~$R$, and
  $R(Q_1) \subseteq Q_1$, but we give a direct proof. If $t' \equiv A
  t''$ then $t'' \in Q_2'$. So $t'' \in Q_1$ by the inductive
  hypothesis. Thus $t'' \in S(Q_1)$ and $t' \equiv A t'' \in Q_1$. If
  $t' \equiv B t''$ then $t'' \in R(Q_2')$. By the coinductive
  hypothesis $t'' \in Q_1$. Hence $t' \equiv B t'' \in Q_1$.\hfill$\Box$
\end{example}

\begin{example}
  Let $Q_1$ and~$T$ be as in the previous example. Consider the
  following corecursive definition of a function~$e : Q_1 \to T$ which
  erases all~$A$s:
  \[
  \begin{array}{rcl}
    e(A t) &=& e(t) \\
    e(B t) &=& B (e(t))
  \end{array}
  \]
  Formally, to make the definition of~$e$ consistent with our theory
  we should also specify~$e(t)$ for $t \in T \setminus Q_1$, but in
  this case we may simply take~$e(t)$ to be an arbitrary element
  of~$T$.

  One shows by induction that a function $e : Q_1 \to T$ satisfies the
  above equations if and only if it satisfies
  \[
  e(A \ldots A B t) = B (e(t))
  \]
  where $A$ occurs a finite number of times (possibly~$0$). But this
  definition of~$e$ is guarded, so we conclude that there exists a
  unique function $e : Q_1 \to T$ satisfying the original
  equations.\hfill$\Box$
\end{example}

\begin{example}\label{ex_mutual_coinduction}
  Define the set~$T$ of coterms coinductively:
  \[
  T \Coloneqq A(T) \parallel B(T) \parallel C(T) \parallel
  D(T) \parallel E(T)
  \]
  We define the relations~$\to_1$ and~$\to_2$ by mutual coinduction.
  \[
  \begin{array}{c}
    \infer={t \to_1 t}{} \quad\quad
    \infer={A(t) \to_1 C(s)}{t \to_2 s} \quad\quad
    \infer={B(t) \to_1 D(s)}{t \to_1 s} \\ \\
    \infer={C(t) \to_1 C(s)}{t \to_2 s} \quad\quad
    \infer={D(t) \to_1 D(s)}{t \to_1 s} \quad\quad
    \infer={E(t) \to_1 E(s)}{t \to_1 s} \\ \\
    \infer={t \to_2 t}{} \quad\quad
    \infer={A(t) \to_2 C(s)}{t \to_1 s} \quad\quad
    \infer={B(t) \to_2 E(s)}{t \to_2 s} \\ \\
    \infer={C(t) \to_2 C(s)}{t \to_1 s} \quad\quad
    \infer={D(t) \to_2 D(s)}{t \to_2 s} \quad\quad
    \infer={E(t) \to_2 E(s)}{t \to_2 s}
  \end{array}
  \]
  Intuitively, the reduction~$\to_1$ changes~$A$ to~$C$, and~$B$
  either to~$D$ or~$E$, starting with~$D$ and switching when
  encountering~$A$ or~$C$. For instance
  \[
  B(B(A(B(C(B(B(t))))))) \to_1 D(D(C(E(C(D(D(t))))))).
  \]

  Formally, the above rules define in an obvious way a
  monotone\footnote{We use the product ordering, i.e., pairs of sets
    are compared with~$\subseteq$ componentwise.} function
  \[
  F : \Pow{T \times T} \times \Pow{T \times T} \to \Pow{T\times T}
  \times \Pow{T \times T}
  \]
  such that $\la {\to_1}, {\to_2}\ra$ is the greatest fixpoint
  of~$F$. Setting $F(X,Y) = \la F_1(X,Y), F_2(X,Y)\ra$, by the
  Beki\v{c} principle (see
  e.g.~\cite[Lemma~1.4.2]{ArnoldNiwinski2001}) we have\footnote{For
    monotone~$f$ we use the notation $\nu x . f(x)$ to denote the
    greatest fixpoint of~$f$.}
  \[
  \begin{array}{rcl}
    {\to_1} &=& \nu X . F_1(X, \nu Y . F_2(X, Y)) \\
    {\to_2} &=& \nu Y . F_2(\nu X . F_1(X, Y), Y).
  \end{array}
  \]
  In other words, one may also think of~$\to_1$ as the greatest
  fixpoint of the monotone function $G : \Pow{T\times T} \to \Pow{T
    \times T}$ defined by $G(X) = F_1(X,H(X))$ where $H(X) = \nu Y
  . F_2(X, Y)$, i.e., $\nu G$ is defined by the coinductive rules
  for~$\to_1$ but instead of the premises $t \to_2 s$ we use $\la t, s
  \ra \in H(\to_1)$, and~$H(X)$ is defined by the coinductive rules
  for~$\to_2$ but with the premises $t \to_1 s$ replaced by $\la
  t,s\ra \in X$. Analogous considerations apply to the definition
  of~$\to_2$.

  We shall now give an example by showing by coinduction that if $t
  \to_i t_1$ and $t \to_i t_2$ then there is~$s$ with $t_1 \to_i s$
  and $t_2 \to_i s$, for $i=1,2$. The proof is rather
  straightforward. If $t \equiv t_1$ then we may take $s \equiv
  t_2$. If $t \equiv A(t')$, $t_1 \equiv C(t_1')$ and $t \to_1 t_1$,
  then also $t_2 \equiv C(t_2')$, $t' \to_2 t_1'$ and $t' \to_2
  t_2'$. By the coinductive hypothesis we obtain~$s'$ such that $t_1'
  \to_2 s'$ and $t_2' \to_2 s'$. Thus $t_1 \equiv C(t_1') \to_1 C(s')$
  and $t_2 \equiv C(t_2') \to_1 C(s')$, so we may take $s \equiv
  C(s')$. Other cases are similar.

  Formally, in the above proof we show the statement:
  \[
  \begin{array}{rl}
    \forall t,t_1,t_2 \in T .\, \exists s_1,s_2 \in T . & ((t \to_1 t_1 \land
    t \to_1 t_2) \To (t_1 \to_1 s_1 \land t_2 \to_1 s_1)) \land \\ & ((t \to_2
    t_1 \land t \to_2 t_2) \To (t_1 \to_2 s_2 \land t_2 \to_2 s_2))
  \end{array}
  \]
  So, after skolemizing, we actually show
  \[
  \begin{array}{rl}
    \forall t,t_1,t_2 \in T . & ((t \to_1 t_1 \land t \to_1 t_2) \To (t_1
    \to_1 f_1(t_1,t_2) \land t_2 \to_1 f_1(t_1,t_2))) \land \\ & ((t \to_2 t_1
    \land t \to_2 t_2) \To (t_1 \to_2 f_2(t_1,t_2) \land t_2 \to_2
    f_2(t_1,t_2)))
  \end{array}
  \]
  for appropriate $f_1, f_2 : T \times T \to T$. The mutually
  corecursive definitions of~$f_1$ and~$f_2$ follow from the
  proof. Formally, we define a corecursive function $f : T \times T
  \to T \times T$ such that $f(t,s) = \la f_1(t,s), f_2(t,s)\ra$ for
  $t,s\in T$. The cartesian product $T \times T$ may be treated as a
  set of coterms~$\Tc_p$ of a special sort~$p$. Then the
  projections~$\pi_1$ and~$\pi_2$ are destructors with a production
  function $\eta_d(n) = \max(0,n - 1)$. The pair-forming operator $\pi
  : T \times T \to \Tc_p$, defined by $\pi(t,s) = \la t,s\ra$, is then
  a constructor with a production function $\eta_c(n,m) = \min(n,m) +
  1$. Thus formally we have for instance $f(C(t),C(s)) = \la
  C(\pi_2(f(t,s))), C(\pi_1(f(t,s))) \ra$. Hence, strictly speaking,
  the definition of~$f$ is not guarded, but it is easily seen to be
  correct nonetheless. Indeed, each clause of the definition of~$f$
  has the form $f(c_1(t),c_2(s)) = \la c_3(\pi_i(f(t,s))),
  c_4(\pi_j(f(t,s))) \ra$, where~$c_1,c_2,c_3,c_4$ are constructors
  and $i,j \in \{1,2\}$, so the prefix production function is
  \[
  \eta(n,m) = \min(n-1+1,m-1+1) + 1 > \min(n,m)
  \]
\end{example}

The above example of mutually corecursive functions is generalized in
the following.

\begin{definition}\label{def_mutual_corecursion}
  We say that functions $f_1,\ldots,f_n : S \to Q$ are defined by
  \emph{mutual corecursion} from $h_j : S \times Q^{m_j} \to Q$ and
  $g_i^j : S \to S$, and $k_{i}^j \in \{1,\ldots,n\}$, $j=1,\ldots,n$,
  $i=1,\ldots,m_j$, if for a function~$f : S \to Q^n$ defined by
  corecursion from
  \[
  \lambda x \vec{y_1} \ldots \vec{y_n} . \la h_1(x,
  \pi_{k_{1}^1}(y_1^1), \ldots, \pi_{k_{1}^{m_1}}(y_1^{m_1})), \ldots,
  h_n(x, \pi_{k_{n}^{1}}(y_n^1), \ldots, \pi_{k_{n}^{m_n}}(y_n^{m_n})) \ra
  \]
  and~$g_i^j$ we have
  \[
  f(x) = \la f_1(x),\ldots,f_n(x) \ra
  \]
  for $x \in S$. We say that a definition by mutual corecursion is
  \emph{guarded} if each~$h_j$ is defined by cases from some
  constructor-guarded functions.
\end{definition}

It follows from our theory that a guarded mutually corecursive
definition uniquely determines the functions~$f_1,\ldots,f_n$. In
coinductive proofs, if the Skolem functions are defined by guarded
mutual corecursion then their well-definedness justifications may be
left implicit.

\clearpage
\phantomsection
\addcontentsline{toc}{section}{References}
\bibliography{biblio}{}
\bibliographystyle{plain}
\clearpage

\appendix

\section{Extending final coalgebras to sized
  CPOs}\label{sec_coalgebra}

In this section we relate our method from
Section~\ref{sec_corecursion} for defining corecursive functions to
the well-established method of finding unique morphisms into the final
coalgebra of a functor. We show a theorem which says that for every
final coalgebra in the category of sets there exists a ``canonical''
sized~CPO. The proof of this theorem is an adaptation of the
construction in~\cite[Theorem~4]{AdamekKoubek1995}. First, we need
some background on the coalgebraic approach to coinduction.

\newcommand{\Set}{\mathrm{Set}}

\subsection{Coalgebraic foundations of coinduction}

In this section we provide a brief overview of coalgebraic foundations
of coinduction. Familiarity with basic category theory is assumed, in
particular with the notions of functor, final object, cone and
limit. We consider only functors in the category of sets. For an
introduction to category theory see e.g.~\cite{Awodey2010}. For more
background on the coalgebraic approach to coinduction see
e.g.~\cite{JacobsRutten2011,Rutten2000}.

\begin{definition}
  A \emph{coalgebra} of an endofunctor $F : \Set \to \Set$, or
  \emph{$F$-coalgebra}, is a pair
  \[
  \la A, f : A \to F A \ra
  \]
  where~$A$ is the carrier set of the coalgebra. A \emph{homomorphism}
  of $F$-coalgebras $\la A, f\ra$ and $\la B, g\ra$ is a morphism $h :
  A \to B$ such that $F h \circ f = g \circ h$, i.e., the following
  diagram commutes:

  \centerline{
    \xymatrix{
      A \ar[r]_{h} \ar[d]_{f} & B \ar[d]_{g} \\
      F A \ar[r]_{F h} & F B
    }
  }

\medskip

  A \emph{final $F$-coalgebra} is a final object in the category of
  $F$-coalgebras and $F$-homomorphisms.

  The \emph{final sequence} of an endofunctor $F : \Set \to \Set$ is
  an ordinal-indexed sequence of sets $\la A_\alpha \ra_\alpha$ with
  morphisms $(w_\gamma^\beta : A_\beta \to A_\gamma)_{\gamma\le\beta}$
  uniquely defined by the conditions:
  \begin{itemize}
  \item $A_{\beta+1} = F(A_{\beta})$,
  \item $w_{\gamma+1}^{\beta+1} = F(w_\gamma^\beta)$,
  \item $w_\beta^\beta = \id$,
  \item $w_\delta^\beta = w_\delta^\gamma \circ w_\gamma^\beta$ for
    $\delta \le \gamma \le \beta$,
  \item if $\beta$ is a limit ordinal then the cone $(w_\gamma^\beta :
    A_\beta \to A_\gamma)_{\gamma<\beta}$ is the limit of the cochain
    $\la A_\gamma\ra_{\gamma<\beta}$, i.e., of the diagram $\la
    \{A_\gamma\}_{\gamma<\beta}, (w_\gamma^\delta : A_\delta \to
    A_\gamma)_{\gamma\le\delta<\beta} \ra$.
  \end{itemize}
\end{definition}

It follows by transfinite induction that the final sequence is indeed
well-defined by the given conditions. See e.g.~\cite{Worrell2005} for
the (easy) proof.

The following two theorems were shown by Adamek and Koubek
in~\cite{AdamekKoubek1995}.

\begin{theorem}\label{thm_final_sequence_coalgebra}
  Suppose the final sequence~$\la A_\alpha\ra_\alpha$ of~$F$
  stabilizes at~$\zeta$, i.e., $w_\zeta^{\zeta+1}$ is an isomorphism.
  Then $\la A_\zeta, (w_\zeta^{\zeta+1})^{-1} \ra$ is a final
  $F$-coalgebra.
\end{theorem}

\begin{theorem}\label{thm_final_sequence_stabilizes}
  If a set-functor has a final coalgebra, then its final sequence
  stabilizes.
\end{theorem}

\subsection{The theorem}

The following theorem shows that for every final coalgebra in the
category of sets there exists a ``canonical'' sized CPO. Moreover, it
is always, in principle, possible to define any morphism into the
final coalgebra as a unique fixpoint of an appropriate monotone
function. This shows that the method of defining corecursive functions
as fixpoints of monotone functions, using an underlying sized~CPO, is
fairly general. The construction in
Theorem~\ref{thm_final_coalgebra_cpo} is an adaptation of the
construction in~\cite[Theorem~4]{AdamekKoubek1995}.

\begin{theorem}\label{thm_final_coalgebra_cpo}
  Let $\la A, t \ra$ be the final coalgebra for a
  set-functor~$T$. There exists a sized~CPO $\la \Ab, \zeta, s, \tcut
  \ra$ with $\Max(\Ab) = A$, such that for any set~$S$ and any
  function $f : S \to T S$, the unique morphism $u : S \to A$ from~$f$
  into the final coalgebra~$\la A, t\ra$ is the unique fixpoint of
  some monotone function $F : \Ab^S \to \Ab^S$ satisfying
  \begin{equation}\label{eq_min}
    \min_{x\in S}s(F(g)(x)) > \min_{x \in S} s(g(x))
  \end{equation}
  for non-maximal $g \in \Ab^S$.
\end{theorem}

\begin{proof}
  Let $\la A_\alpha \ra_\alpha$ with $(w_\beta^\alpha : A_\alpha \to
  A_\beta)_{\beta\le\alpha}$ be the final sequence of~$T$. Since~$T$
  has a final coalgebra, by
  Theorem~\ref{thm_final_sequence_stabilizes} the final sequence
  stabilizes at some ordinal~$\zeta$. By
  Theorem~\ref{thm_final_sequence_coalgebra} we may assume without
  loss of generality that $\la A, t \ra = \la A_\zeta,
  (w_\zeta^{\zeta+1})^{-1} \ra$ (otherwise we just need to compose
  some morphisms below with the isomorphism between~$A$
  and~$A_\zeta$). Without loss of generality we may identify~$A$ with
  $\{\zeta\} \times A$ (otherwise the definition of~$\Ab$ below just
  needs to be complicated slightly by taking the carrier set to be
  e.g.~$A \cup (\{A\} \times \amalg_{\alpha<\zeta}A_\alpha)$ and
  adjusting the definition of~$\qle$ accordingly). If~$p$ is a pair,
  then by~$p_1$ we denote the first and by~$p_2$ the second component
  of~$p$. Take $\Ab = \la \amalg_{\alpha\le\zeta}A_\alpha, \qle \ra$
  with $p \qle q$ iff $p_1 \le q_1$ and $w_{p_1}^{q_1}(q_2) = p_2$. It
  follows from the definition of the final sequence of an endofunctor
  that~$\qle$ is a partial order.

  We show that~$\Ab$ is a~CPO. The bottom of~$\Ab$ is $\la 0, \bot
  \ra$ where~$\bot$ is the sole element of~$A_0$. Let $D \subseteq
  \Ab$ be a directed set. First, we show that~$D$ is in fact a
  chain. Let $p,q \in D$ with $p_1 \le q_1$. Because~$D$ is directed
  there is $r \in D$ with $p,q \qle r$, i.e., $p_1 \le q_1 \le r_1$,
  $w_{q_1}^{r_1}(r_2) = q_2$ and $w_{p_1}^{r_1}(r_2) = p_2$. Because
  $w_{p_1}^{r_1} = w_{p_1}^{q_1} \circ w_{q_1}^{r_1}$ we have
  $w_{p_1}^{q_1}(q_2) = w_{p_1}^{q_1}(w_{q_1}^{r_1}(r_2)) =
  w_{p_1}^{r_1}(r_2) = p_2$. Hence $p \qle q$.

  Let~$\alpha$ be the least upper bound of $D_1 = \{p_1 \mid p \in
  D\}$. If there is $p \in D$ with $p_1 = \alpha$, i.e., $\alpha \in
  D_1$ is the largest element of~$D_1$, then~$p$ is the largest
  element of~$D$, and thus the supremum. Indeed, let $q \in
  D$. Since~$D$ is a chain, $q \qle p$ or $p \qle q$. If $p \qle q$
  then $q_1 = \alpha$, because~$p_1=\alpha$ is the largest element
  of~$D_1$. But this implies $q = p$, because $w_\alpha^\alpha = \id$.

  So assume $\alpha \notin D_1$. Then~$\alpha$ must be a limit
  ordinal. So the cone $C = (w_\beta^\alpha : A_\alpha \to
  A_\beta)_{\beta<\alpha}$ is the limit of the cochain $\la
  A_\beta\ra_{\beta<\alpha}$. Let $A_\alpha' = A_\alpha \cup \{a\}$
  where $a \notin A_\alpha$. We define functions $f_\beta^\alpha :
  A_\alpha' \to A_\beta$ for $\beta<\alpha$ as follows: $f_\beta(x) =
  w_\beta^\alpha(x)$ if $x \ne a$, and $f_\beta(a) =
  w_\beta^\gamma(z_2^\beta)$ for the element~$z^\beta \in D$ such that
  $z_1^\beta = \gamma \ge \beta$ is smallest in $\{ \gamma \in D_1
  \mid \gamma \ge \beta \}$. The element~$z^\beta$ is uniquely
  defined, because distinct elements of~$\Ab$ with the same first
  components are pairwise incomparable, and~$D$ is a chain with
  elements with first components arbitrarily close to~$\alpha$, and
  $\beta < \alpha$. We show that $(f_\beta : A_\alpha' \to A_\beta)$
  is a cone over the cochain $\la A_\beta\ra_{\beta<\alpha}$, i.e.,
  over the diagram \mbox{$\la \{A_\beta\}_{\beta<\alpha},
    (w_\beta^\gamma : A_\gamma \to A_\beta)_{\beta\le\gamma<\alpha}
    \ra$}. Let $\gamma \ge \beta$. We have
  $w_\beta^\gamma(f_\gamma(a)) = w_\beta^{\gamma_1}(z_2^{\gamma})$
  where $\gamma_1 \ge \gamma$ and~$z_2^{\gamma}$ are such that
  $f_\gamma(a) = w_\gamma^{\gamma_1}(z_2^{\gamma})$. Let $\beta_1 \ge
  \beta$ be such that $f_\beta(a) =
  w_\beta^{\beta_1}(z_2^{\beta})$. Then $\beta_1 \le \gamma_1$, so
  $z_2^{\beta} \qle z_2^{\gamma}$, because~$D$ is a chain. Thus
  $w_{\beta_1}^{\gamma_1}(z_2^{\gamma}) = z_2^{\beta}$, so
  $w_{\beta}^{\gamma_1}(z_2^{\gamma}) =
  w_{\beta}^{\beta_1}(z_2^{\beta})$. Hence
  $w_\beta^\gamma(f_\gamma(a)) = w_\beta^{\gamma_1}(z_2^{\gamma}) =
  w_\beta^{\beta_1}(z_2^{\beta}) = f_\beta(a)$. For $x \in A_\alpha$
  the condition $f_\beta(x) = w_\beta^\gamma(f_\gamma(x))$ follows
  directly from definitions. Therefore \mbox{$(f_\beta : A_\alpha' \to
    A_\beta)_{\beta<\alpha}$} is a cone, and since~$C$ is the limit,
  there exists a unique $u : A_\alpha' \to A_\alpha$ such that
  $f_\beta = w_\beta^\alpha \circ u$ for $\beta < \alpha$. We show
  that $\bar{a} = \la \alpha, u(a) \ra$ is the supremum of~$D$. To
  prove that~$\bar{a}$ is an upper bound, it suffices to show that if
  $d \in D$ then $w_{d_1}^\alpha(u(a)) = d_2$. But this holds because
  $w_{d_1}^\alpha(u(a)) = f_{d_1}(a) = d_2$. So suppose~$\bar{b}$ is
  also an upper bound. Then so is $\la \alpha,
  w_\alpha^{\bar{b}_1}(\bar{b}_2) \ra$, hence we may assume $\bar{b}_1
  = \alpha$. Define $u' : A_\alpha' \to A_\alpha$ by: $u'(x) = u(x)$
  if $x \ne a$, and $u'(a) = \bar{b}_2$. Since $w_{d_1}^\alpha(u'(a))
  = d_2$ for $d \in D$, we have $f_\beta(a) =
  w_\beta^{\gamma}(z_2^\beta) =
  w_\beta^{\gamma}(w_{\gamma}^\alpha(u'(a))) = w_\beta^\alpha(u'(a))$
  for $\beta < \alpha$, where $\gamma = z_1^\beta$. This implies
  $f_\beta = w_\beta^\alpha \circ u'$ for $\beta < \alpha$. Thus $u' =
  u$, because $u : A_\alpha' \to A_\alpha$ is unique such that
  $f_\beta = w_\beta^\alpha \circ u$ for $\beta < \alpha$. Hence
  $\bar{b} = \bar{a}$. So~$\bar{a}$ is the supremum of~$D$. Therefore,
  $\Ab$ is a~CPO.

  It is clear that $\Max(\Ab) = A (= \{\zeta\} \times A)$. The size
  function $s : \Ab \to \On(\zeta)$ is defined by $s(x) = x_1$ for $x
  \in \Ab$. It is obviously surjective. That~$s$ is continuous follows
  from the construction of supremums we have given in the previous
  paragraph. Of course, $s(x) = \zeta$ iff $x \in \Ab$ is maximal. The
  cut-function $\tcut : \On(\zeta) \times \Ab \to \Ab$ is defined by:
  \begin{itemize}
  \item $\tcut(\alpha,x) = \la \alpha, w_\alpha^{x_1}(x_2) \ra$ if
    $x_1 \ge \alpha$,
  \item $\tcut(\alpha,x) = x$ otherwise.
  \end{itemize}
  It follows from definitions that~$\tcut$ is monotone in both
  arguments. Therefore, $\la \Ab, \zeta, s, \tcut\ra$ is a sized~CPO
  with $\Max(\Ab) = A$. To save on notation, from now on we confuse $x
  \in \Ab$ with~$x_2$, using~$s(x)$ to denote the first component.

  Let~$S$ be a set and let $f : S \to T S$. Suppose $u : S \to A$ is
  the unique morphism from~$f$ into the final coalgebra $\la A,
  t\ra$. For $g : S \to \Ab$ define $m(g) = \min_{x \in S} s(g(x))$,
  and define $g^* : S \to \Ab$ by $g^*(x) = w_{m(g)}^{s(g(x))}(g(x))$
  for $x \in \Ab$. Note that $g^* : S \to A_{m(g)}$, so $T g^* : T S
  \to A_{m(g)+1}$, and if $m(g) = \zeta$ then $g^* = g$. Let $F :
  \Ab^S \to \Ab^S$ be defined by
  \[
  F(g) =
  \left\{
    \begin{array}{cl}
      T g^* \circ f & \text{ if } m(g) < \zeta \\
      t^{-1} \circ T g \circ f & \text{ otherwise }
    \end{array}
  \right.
  \]
  for $g \in \Ab^S$. For non-maximal $g \in \Ab^S$ we have $m(g) <
  \zeta$, and thus
  \[
  \min_{x \in S}F(g)(x) = \min_{x\in S} T g^*(f(x)) = m(g) + 1 > m(g)
  = \min_{x \in S}(s(g(x)))
  \]
  so~\eqref{eq_min} is satisfied. We show that~$F$ is monotone. So let
  $g,h \in \Ab^S$ with $g \qle h$, i.e., $g(x) \qle h(x)$ for all $x
  \in S$. Then $m(g) \le m(h)$. We may assume $m(g) < \zeta$, because
  if $m(g) = m(h) = \zeta$ then $g = h$. We have $g^*(x) \qle h^*(x)$
  for all $x \in S$. Indeed, for $x \in S$ we have $g(x) =
  w_{s(g(x))}^{s(h(x))}(h(x))$ and thus
  \[
  \begin{array}{rcl}
    g^*(x) &=& w_{m(g)}^{s(g(x))}(g(x)) \\
    &=& w_{m(g)}^{s(g(x))}(w_{s(g(x))}^{s(h(x))}(h(x))) \\
    &=& w_{m(g)}^{s(h(x))}(h(x)) \\
    &=& w_{m(g)}^{m(h)}(w_{m(h)}^{s(g(x))}(h(x))) \\
    &=& w_{m(g)}^{m(h)}(h^*(x)).
  \end{array}
  \]
  So $g^* = w_{m(g)}^{m(h)} \circ h^*$, and hence $T g^* = T
  w_{m(g)}^{m(h)} \circ T h^*$. We have $T w_{m(g)}^{m(h)} =
  w_{m(g)+1}^{m(h)+1}$, so
  \[
  F(g) = T g^* \circ f = T w_{m(g)}^{m(h)} \circ T h^* \circ f
  = w_{m(g)+1}^{m(h)+1} \circ T h^* \circ f.
  \]
  If $m(h) < \zeta$ then this implies $F(g) = w_{m(g)+1}^{m(h)+1}
  \circ F(h)$, so $F(g) \qle F(h)$. If $m(h) = \zeta$ then $F(g) =
  w_{m(g)+1}^{\zeta+1} \circ T h \circ f = w_{m(g)+1}^\zeta \circ
  w_\zeta^{\zeta+1} \circ T h \circ f = w_{m(g)+1}^\zeta \circ F(h)$
  because $t^{-1} = w_\zeta^{\zeta+1}$. So then also $F(g) \qle
  F(h)$. Therefore~$F$ is monotone.

  It remains to show that~$u$ is the unique fixpoint of~$F$. Let~$v$
  be a fixpoint of~$F$. By~\eqref{eq_min} we must have $v \in
  A^S$. Then $F(v) = t^{-1} \circ T v \circ f$, so $t^{-1} \circ T v
  \circ f = v$. This implies $T v \circ f = t \circ v$, so~$v$ is a
  morphism from the coalgebra $\la S, f \ra$ into the final coalgebra
  $\la A, t \ra$. Therefore $v = u$.
\end{proof}

\end{document}